\documentclass{CSML}

\pdfoutput=1

\usepackage{lastpage}
\def\dOi{13(4:25)2017}
\lmcsheading%
{\dOi}
{1--\pageref{LastPage}}
{}
{}
{Dec.~31, 2015}
{Dec.~\phantom08, 2017}
{}

\usepackage[utf8]{inputenc}

\usepackage{bussproofs}
\EnableBpAbbreviations 

\usepackage{etoolbox} % Needed for toggles and theorem-appendix.tex
\newtoggle{techreport}
\toggletrue{techreport} %%% WARNING: toggle false before LNCS submission!!!
\newtoggle{lmcs}
\toggletrue{lmcs}

\usepackage{appendix}
\usepackage{xspace}

\usepackage{amsmath}
\usepackage{amssymb}
\usepackage{stmaryrd}

\usepackage{graphicx}
\usepackage{marvosym}

\usepackage{color}
\usepackage[usenames,dvipsnames]{xcolor}

\usepackage{pdfsync}
\usepackage{wasysym}
\usepackage{xifthen} % needed for isempty
\usepackage{proof}                  % Logic proofs

\usepackage{alltt} % for math mode within verbatim
\usepackage{enumitem}

\newlist{inlinelist}{enumerate*}{1}
 \setlist*[inlinelist,1]{%
  label=(\roman*),
}

\newlist{inlinelistarabic}{enumerate*}{1}
 \setlist*[inlinelistarabic,1]{%
  label=(\arabic*),
}

\newlist{inlinelistalph}{enumerate*}{1}
 \setlist*[inlinelistalph,1]{%
  label=(\alph*),
}

\usepackage{fixltx2e} % Text mode subscript and superscript
\usepackage{xspace} % Smart spaces after commands

\usepackage{nicefrac}
\usepackage{environ,thmtools}  % Advanced theorem handling

\graphicspath{{figures/}} \DeclareGraphicsExtensions{.pdf,.jpg}
\usepackage[hidelinks]{hyperref} % Load cleveref last!
\usepackage{cleveref}

\usepackage{mathtools}% needed for xRightarrow

\usepackage[final,nomargin,inline,index]{fixme} % Simplified management of FIXME's
\fxusetheme{color}

\FXRegisterAuthor{MM}{anMM}{\color{cyan} {\underline{Maurizio}}}
\FXRegisterAuthor{bart}{anbart}{\color{magenta} {\underline{bart}}}
\FXRegisterAuthor{tizi}{antizi}{\color{green} {\underline{tizi}}}

\newcommand{\codefont}{\fontsize{10}{10}\selectfont}
\newcommand{\code}[1]{{\tt\codefont {#1}}}

\usepackage{tikz} % for graph of module importation --- load LAST!
\input{macro.sty}

\begin{document}

\title[Timed session types]{Timed session types\rsuper*}

\author[M.~Bartoletti]{Massimo Bartoletti\rsuper a} %required
\address{{\lsuper{a,b,c}}Universit\`a degli Studi di Cagliari, Italy}	%required
\email{\{bart, t.cimoli\}@unica.it}  %optional
% \thanks{thanks 1, optional.}	%optional

\author[T.~Cimoli]{Tiziana Cimoli\rsuper b} %required
\address{\vspace{-18 pt}}	%required
%\email{t.cimoli@unica.it}  %optional
% \thanks{thanks 1, optional.}	%optional

\author[M.~Murgia]{Maurizio Murgia\rsuper c} %required
\address{{\lsuper c}University of Kent, UK}	%required
\email{M.Murgia@kent.ac.uk}  %optional
% \thanks{thanks 1, optional.}	%optional

%% mandatory lists of keywords and classifications:
%\ACMCCS{[{\bf Theory of computation}]: Models of computation; Concurrency --
%Semantics and reasoning -- Program reasoning -- Program specifications -- Program verification;
%Semantics and reasoning -- Program semantics -- Operational semantics.}
\titlecomment{{\lsuper *}Full version of an Extended Abstract presented at FORTE'15.} 
% \titlecomment{} 
% OPTIONAL comment concerning the title, \eg, if a variant
% or an extended abstract of the paper has appeared elsewehere
%%%%%%%%%%%%%%%%%%%%%%%%%%%%%%%%%%%%%%%%%%%%%%%%%%%%%%%%%%%%%%%%%%%%%%%%%%%

\pagestyle{headings}    % switches on printing of running heads

\begin{abstract}
  Timed session types 
  formalise timed communication protocols between
  two participants at the endpoints of a session.
  They feature a decidable compliance relation, 
  which generalises to the timed setting the progress-based 
  compliance between untimed session types.
  We show a sound and complete technique to decide 
  when a timed session type admits a compliant one. 
  Then, we show how to construct
  the most precise session type compliant with a given one,
  according to the subtyping preorder induced by compliance.
  Decidability of subtyping follows from these results.
\end{abstract}

\maketitle              % typeset the title of the contribution

%%% main text

\section{Introduction} \label{sec:introduction}

Session types are formal descriptions of interaction protocols
involving two or more participants over a network~\cite{HondaVK98CommProg,TakeuchiHK94}.
They can be used to specify the behavioural interface
of a service or a component, and to statically check 
through a (session-)type system that this conforms to its implementation,
so enabling compositional verification of distributed applications.
Session types support formal definitions of compatibility or \emph{compliance}
(when two or more session types, composed together, behave correctly), 
and of substitutability or \emph{subtyping}
(when a service can be safely replaced by another one, while preserving the 
interaction capabilities with the context).
Since these notions are often decidable and computationally tractable 
(for synchronous session types),
or safely approximable (for asynchronous ones),
session typing is becoming a particularly attractive approach 
to the problem of correctly designing distributed applications.
This is witnessed by a steady flow of foundational studies~\cite{Bartoletti15jlamp,Castagna09ppdp,Denielou13icalp,Dezani09wsfm,Huttel16csur}
and of tools~\cite{Ancona16ftpl,Corin08secure,Yoshida13scribble}
based on them in the last few years.

In the simplest setting, session types are terms of a process algebra
featuring a selection construct (an \emph{internal choice} among a set of branches),
a branching construct (an \emph{external choice} offered to the environment), 
and recursion.
In this basic form, session types cannot faithfully capture a natural and relevant
aspect of interaction protocols, namely the timing constraints
among the communication actions.
While formal methods for time have been studied for at least a couple of decades,
they have approached the realm of session types 
very recently~\cite{Bocchi14concur,Neykova14beat},
with the goal of extending compositional verification techniques~\cite{Honda08MPS,Honda16jacm}.
These approaches introduce time into an already sophisticated setting,
featuring multiparty session types with asynchronous communication (via unbounded buffers).

We think that studying timed session types in a basic setting
(synchronous communication between two endpoints, as in the seminal untimed version)
is worthy of attention.
From a theoretical point of view, the objective is to lift to the timed case
some decidability results, like those of compliance and subtyping.
Some intriguing problems arise: unlike in the untimed case, 
a timed session type not always admits a compliant.
Hence, besides deciding if two session types \emph{are compliant},
it becomes a relevant problem whether a session type \emph{has a compliant}.
From a more practical perspective, decision procedures for timed session types,
like those for compliance and for dynamic verification,
would enable the implementation of programming tools and infrastructures
for the development of safe communication-oriented distributed applications.

\paragraph{\textbf{Timed session types.}}

In this paper we present a theory of binary timed session types (TSTs).
For instance, we describe as the following TST
the protocol of a service % which offers a remote function
taking as input a zip code, and then either providing as output the current weather,
or aborting:
\begin{equation}
  \label{eq:intro:p}
  \tsbP \; = \; \TsumE{\atomIn{zip}}{\clockX}{
    \left(
      \TsumI{\atomOut{weather}}{5 < \clockX < 10}{} \sumInt \TsumI{\atomOut{abort}}{\clockX < 1}{}
    \right)
  }
\end{equation}
The prefix $\TsumE{\atomIn{zip}}{\clockX}{}$ states that the service 
can receive a $\atom{zip}$ code, and then reset a clock~$\clockX$. 
The continuation is an \emph{internal choice} between two outputs:
either send $\atom{weather}$ in a time window of $(5,10)$
time units, or $\atom{abort}$ the protocol within $1$ time unit.

A standard notion of compliance in the untimed setting
is the one which relates two session types whenever their interaction 
never reaches a deadlock~\cite{Bartoletti15plabs}.
For instance, consider the two session types:
\begin{equation}
  \label{eq:intro:untimed-weather}
  \TsumE{\atomIn{zip}}{}{
    \left(
      \TsumI{\atomOut{weather}}{}{} \sumInt \TsumI{\atomOut{abort}}{}{}
    \right)
  }
  \qquad\qquad
  \TsumI{\atomOut{zip}}{}{
    \left(
      \TsumE{\atomIn{weather}}{}{} \sumExt \TsumE{\atomIn{abort}}{}{}
    \right)
  }
\end{equation}

After the synchronisation on $\atom{zip}$,  
the leftmost session type can internally choose between one of the outputs
$\atom{weather}$ and $\atom{abort}$, and each one of these choices
is matched by the external choice in the rightmost session type.
Therefore, the two session types in~\eqref{eq:intro:untimed-weather} 
are considered compliant.
In the \emph{untimed} setting, the notion of compliance is decidable:
a procedure to decide whether two session types are compliant
can just explore the finite-state LTS which describes their interaction,
and return true whenever there are not deadlocks.

In the timed setting, 
we can use the same notion of compliance as in the untimed case, 
but applied on the \emph{timed} LTS that describes the interaction of the two TSTs.
For instance, the TST $\tsbP$ in~\eqref{eq:intro:p} is \emph{not} compliant with:
\begin{equation}
  \label{eq:intro:q}
  \tsbQ \; = \; \TsumI{\atomOut{zip}}{\clockY}{
    \left(
      \TsumE{\atomIn{weather}}{\clockY < 7}{} \sumExt \TsumE{\atomIn{abort}}{\clockY < 5}{}
    \right)
  }
\end{equation}
because $\tsbQ$ is available to receive $\atom{weather}$ until $7$ time units
since it has sent the $\atom{zip}$ code, while $\tsbP$ can choose
to send $\atom{weather}$ until $10$ time units
(note that $\tsbP$ and $\tsbQ$, cleaned from all time annotations, 
are compliant in the untimed setting).
Differently from the untimed case, 
the timed LTS of TSTs is
infinite-state, hence compliance is not trivially decidable.
A relevant problem is then the decidability of compliance in the timed case.

A further difference from the untimed case is that not every TST admits
a compliant one. 
In the untimed case, a session type is always compliant to its syntactic dual,
which is obtained by swapping outputs with inputs, and internal with external choices,
as in~\eqref{eq:intro:untimed-weather}.
For instance, consider the client protocol:
\[
\tsbQi \; = \; \TsumI{\atomOut{zip}}{\clockY < 10}{
  \left(
    \TsumE{\atomIn{weather}}{\clockY < 7}{} \sumExt \TsumE{\atomIn{abort}}{\clockY < 5}{}
  \right)
}
\]
No service can be compliant with $\tsbQi$, because if $\tsbQi$
sends the $\atom{zip}$ code, \eg, at time $8$, 
one cannot send $\atom{weather}$ or $\atom{abort}$
in the given time constraints.
This observation gives rise to other two questions.
How to decide if a TST admits a compliant one?
In such case, can we effectively construct a TST which is compliant with the given one?

Another relevant notion in the untimed setting
is that of semantic subtyping~\cite{Barbanera15mscs}.
Roughly, a session type $\tsbPi$ is a subtype of $\tsbP$ whenever
all the session types which are compliant with $\tsbP$ are compliant also with $\tsbPi$.
This notion can be used to detect when a service can be safely substituted by another one.
Indeed, a service with type $\tsbP$ can be replaced by
one with a subtype $\tsbPi$ of $\tsbP$,
guaranteeing that all the services which interacted correctly
with the old one will do the same with the new one. 
Establishing decidability of subtyping for TSTs
would allow to check substitutability of time-aware services.

\paragraph{\textbf{Contributions.}}

We summarize the main contributions of the present work as follows:
\begin{enumerate}

\item \label{intro:contrib:semantics}
  We give the syntax and semantics of TSTs.
  Their semantics is a conservative extension of the
  synchronous semantics of untimed session types~\cite{Barbanera15mscs},
  adding \emph{clock valuations} to associate each clock with a positive real.

\item %
  We develop a sound and complete decision procedure 
  for verifying compliance between TSTs (\Cref{th:compliance:decidable}).
  To do that, we reduce this problem to that of model-checking deadlock freedom  
  in timed automata~\cite{Alur94theory}, which is decidable.

\item We develop a procedure to detect whether a TST admits a compliant one.
  This procedure takes the form of a kind system which associates,
  to each TST $\tsbP$, a set of clock valuations under 
  which $\tsbP$ admits a compliant. %
  The kind system is sound and complete 
  (\Cref{th:dual-sound} and~\Cref{th:dual-complete}),
  and kind inference is decidable (\Cref{th:kind-inference:decidable}).
  From this we infer a decidable (sound and complete) procedure
  for the existence of a compliant.

\item We exploit the kind system to define the \emph{canonical compliant} 
  of a TST (\Cref{def:dual}).
  From the decidability results on the kind system 
  we obtain an effective procedure to construct the canonical compliant 
  (\Cref{th:dual:computable}).

\item We study the semantic subtyping preorder for TSTs.
  Building upon the decidability of compliance and of kind inference,
  we prove that semantic subtyping between TSTs is decidable (\Cref{th:subtyping-decidable}).
  We also show that the canonical compliant of a TST is the \emph{greatest} TST compliant with the original one,
  according to the preorder (\Cref{lem:dual-minimal}). 

\item We provide TSTs with a \emph{monitoring semantics} 
  (\Cref{def:tst:monitoring-semantics}),
  which detects when a participant does not respect its TST.
  This semantics % enjoys some desirable properties:
  is deterministic (\Cref{lem:tst-monitoring:determinism}), 
  and it is coherent with the semantics of item~\eqref{intro:contrib:semantics}
  with respect to compliance (\Cref{lem:st:mcompliance}).
  Further, it guarantees that in each state of an interaction, 
  either we have reached success, 
  or someone is in charge of a move, 
  or not respecting its TST (\Cref{lem:onduty-culpable}).

\item We develop a suite of tools which implement the primitives discussed above, 
  including a compliance checker, the canonical compliant construction,
  and an execution monitor based on TSTs. 
  Our tools are available at \href{http://co2.unica.it}{co2.unica.it}.

\end{enumerate}

\paragraph{\textbf{Applications of timed session types.}}

A theory of timed session types with the primitives outlined above 
can be applied to the design and implementation of distributed applications. 
For instance, the message-oriented middleware in~\cite{CO2} 
exploits our theory to allow disciplined interactions 
between mutually distrusting services. %
The idea is a contract-oriented, bottom-up composition, 
where only those services with compliant contracts can interact
via binary sessions. %
The middleware makes available a global store where services can advertise
contracts, in the form of TSTs.

The middleware guarantees that sessions are established only between services with compliant TSTs.
More specifically, assume that a service $\pmvA$ 
advertises a contract $\tsbP$ to the store
(this is only possible if $\tsbP$ admits a compliant).
A session between $\pmvA$ and another service $\pmvB$ can be established if %
\begin{inlinelist}
\item $\pmvB$ advertises a contract $\tsbQ$ compliant with $\tsbP$, or
\item $\pmvB$ \emph{accepts} the contract $\tsbP$ 
  (in this case, the contract of $\pmvB$ is the canonical compliant of $\tsbP$).
\end{inlinelist}
In the first case, the middleware exploits the decision procedure 
for compliance (\Cref{th:compliance:decidable}),
while in the second one it uses the canonical compliant 
construction (\Cref{def:dual}).

The middleware also implements an execution monitor for TSTs, 
to check if the actions performed by services conform to their protocol, 
and --- otherwise --- %
to detect which services have caused the violation. %
More specifically, 
when the session is established, $\pmvA$ and $\pmvB$
can interact by sending/receiving messages through the session.
During the interaction, all their actions are monitored 
(according to~\Cref{def:tst:monitoring-semantics}),
and possible misbehaviours are detected
(according to~\Cref{def:culpable}) and sanctioned. %
Sanctions negatively affect the reputation of a service, 
and consequently its chances to establish new sessions. %

In systems of \emph{honest} services that always respect their contracts,
compliance ensures progress of the whole application. 
In systems with dishonest services, dynamic verification
of all the exchanged messages guarantees safe executions,
and the sanction mechanism
automatically marginalizes the dishonest services.

The middleware APIs (available at~\href{http://co2.unica.it}{co2.unica.it})
implement the following primitives:
\begin{inlinelistalph}
\item \code{tell} advertises a TST (say, $\tsbP$) to the middleware.
Firstly, the middleware checks if $\tsbP$ admits a compliant one
(otherwise, it rejects $\tsbP$).
Then, it searches the store for a TST compliant with $\tsbP$. 
When such TST is found, the middleware starts a new session 
between the respective services, sending them the session identifier. % 
\item \code{accept} is used by a service $\pmvA$ 
to directly establish a session with a service $\pmvB$, 
knowing the identifier of a TST $\tsbP$ advertised by $\pmvB$.
In this case $\pmvA$ does not advertise a TST, and its contract in the 
new session will be the canonical compliant of $\tsbP$.
\item \code{send} and \code{receive} 
perform outputs and inputs of messages in an already established session.
The middleware monitors all the exchanged messages,
also keeping track of the passing of time, to detect and sanction contract violations.
\end{inlinelistalph}

\paragraph{\textbf{Structure of the paper.}}

We start in~\Cref{sec:tst} by giving the syntax and semantics of TSTs.
In~\Cref{sec:tst-compliance} we define a compliance relation between TSTs,
extending to the timed setting the progress-based notion of compliance between untimed session types.
In~\Cref{sec:tst-duality} we study when a TST admits a compliant one, 
and we define a construction to obtain the canonical compliant of a TST.
Decidability of this construction is established in~\Cref{sec:comput-dual}.
In~\Cref{sec:tst-subtyping} we study the semantic subtyping preorder for TSTs. %
In~\Cref{sec:ex-paypal-full} we report a case study of a user agreement policy, 
showing how to model it as a TST.
We detail the encoding of TSTs into timed automata in~\Cref{sec:tst-to-ta}.
In~\Cref{sec:tst-monitoring} we address the problem of
dynamically monitoring interactions regulated by TSTs. 
Finally, in~\Cref{sec:related-work} we discuss some related work.
The proofs of all our statements are either in the main text, 
or in~\Cref{sec:proofs-compliance,sec:proofs-duality,sec:proofs-comput-dual,sec:proofs-tst-to-ta,sec:proofs-monitoring}.

\section{Timed session types: syntax and semantics} 
\label{sec:tst}
\label{sec:tst-semantics}

Let $\Act$ be a set of \emph{actions},
ranged over by $\atomA, \atomB, \ldots$.
We denote with $\ActOut$ the set $\setcomp{\atomOut{a}}{\atomA \in \Act}$
of \emph{output actions},
with $\ActIn$ the set $\setcomp{\atomIn{a}}{\atomA \in \Act}$
of \emph{input actions},
and with $\BLab = \ActOut \cup \ActIn$ the set of \emph{branch labels},
ranged over by $\labL,\labLi,\ldots$.
We use $\delta, \delta', \ldots$ to range over the set $\Realpos$
of positive real numbers including zero,
and $d,d',\ldots$ to range over 
\iftoggle{techreport}{the set $\Nat$ of natural numbers}{$\Nat$}.
Let $\Clocks$ be a set of
\emph{clocks}, namely variables in $\Realpos$, ranged over by
$\clockT, \clockTi, \ldots$.  We use $\resetR,\resetT,\ldots \subseteq
\Clocks$ to range over sets of clocks. 

\newcommand{\guardsyntax}{%
  \guardG 
  \; ::= \;
  \guardTrue 
  \ \sep \ 
  \neg \guardG 
  \ \sep \ 
  \guardG  \land \guardG  
  \ \sep \
  \clockT \circ d
  \ \sep \
  \clockT - \clockTi \circ d
}

\begin{defi}[\textbf{Guard}]
  \label{def:guards}
  We define the set $\GuardG[\Clocks]$ of \emph{guards} 
  over clocks $\Clocks$ as follows:
  \[
  \guardsyntax
  \tag*{where $\circ \in \setenum{<,\leq,=,\geq,>}$}
  \]
\end{defi}\bigskip

\noindent A TST $\tsbP$ models the behaviour of a single participant
involved in an interaction (\Cref{def:tst:syntax}).
To give some intuition, we consider two participants, 
Alice ($\pmvA$) and Bob ($\pmvB$), 
who want to interact. %
$\pmvA$ advertises an \emph{internal choice}
$\TSumInt[i]{\atomOut[i]{a}}{\guardG[i],\resetR[i]}{\tsbP[i]}$ 
when she wants to do one of the outputs $\atomOut[i]{a}$ 
in a time window where $\guardG[i]$ is true;
further, the clocks in $\resetR[i]$ will be reset after the output is performed.
Dually, $\pmvB$ advertises an \emph{external choice} 
$\TSumExt[i]{\atomIn[i]{a}}{\guardG[i],\resetR[i]}{\tsbQ[i]}$
to declare that he is available to receive
each message $\atom[i]{a}$ in \emph{any instant} 
within the time window defined by $\guardG[i]$
(and the clocks in $\resetR[i]$ will be reset after the input).

\begin{defi}[\textbf{Timed session type}] 
  \label{def:tst:syntax}
  \emph{Timed session types} $\tsbP,\tsbQ,\ldots$ 
  are terms of the following grammar:
  \begin{align*}
    \tsbP \;\; & ::= \;\;
    % \textstyle
    % \cnil
    % \ \sep \
                 \success
                 \ \sep \
                 \TSumInt[i \in I]{\atomOut[i]{a}}{\guardG[i],\resetR[i]}{\tsbP[i]} 
                 \ \sep  \
                 \TSumExt[i \in I]{\atomIn[i]{a}}{\guardG[i],\resetR[i]}{\tsbP[i]}
                 \ \sep \
                 \rec \tsbX \tsbP
                 \ \sep \
                 \tsbX
  \end{align*}
  where   
  \begin{inlinelist}
  \item the set $I$ is finite and non-empty,
  \item the actions in internal/external choices are pairwise distinct,
  \item recursion is guarded 
    (\eg, we forbid both $\rec \tsbX \tsbX$ and $\rec{\tsbX}{\rec{\tsbY}{\tsbP}}$).
  \end{inlinelist}  
\end{defi}

Except where stated otherwise, we consider TSTs up-to unfolding of recursion.
A TST is \emph{closed} when it has no recursion variables.
If $\tsbQ = \TSumInt[i \in I]{\atomOut[i]{a}}{\guardG[i],\resetR[i]}{\tsbP[i]}$ and $0 \not\in I$,
we write $\atomOut[0]{a}.\tsbP[0] \sumInt \tsbQ$ for 
$\TSumInt[i \in I \cup \setenum{0}]{\atomOut[i]{a}}{\guardG[i],\resetR[i]}{\tsbP[i]}$
(the same for external choices).
True guards, empty resets, 
and trailing occurrences of the \emph{success state} $\success$ can be omitted.

\begin{exa}[\textbf{Simplified PayPal}] \label{ex-paypal}
  Along the lines of PayPal User Agreement~\cite{PayPal},
  we specify the protection policy for buyers of a simple on-line payment
  platform, called PayNow 
  (see~\iftoggle{techreport}{\Cref{sec:ex-paypal-full}}{\cite{tst-long}}
  for the full version). 
  PayNow helps customers in on-line purchasing, providing protection
  against seller misbehaviours. %
  In case a buyer has not received what he
  has paid for, he can open a dispute within $180$ days 
  from the date the buyer made the payment. %
  After opening of the dispute, the buyer and
  the seller may try to come to an agreement. %
  If this is not the case,
  within $20$ days, the buyer can escalate the dispute to a claim. %
  However, the buyer must wait at least $7$ days from the date of
  payment to escalate a dispute. %
  Upon not reaching an agreement, if
  still the buyer does not escalate the dispute to a claim within $20$
  days, the dispute is considered aborted. %
  During a claim procedure, PayNow will ask the buyer to provide
  documentation to certify the payment,
  within $3$ days of the date the dispute was escalated to a claim. %
  After that, the payment will be refunded within 7 days.
  PayNow's agreement is described by the following TST $\tsbP$:
  \begin{align*}
    \tsbP \; = \; & \atomIn{pay} \setenum{\clockT[pay]}.
                    \left( 
                    \atomIn{ok}
                    \sumExt \; 
                    \atomIn{dispute}\setenum{\clockT[pay]<180,\clockT[d]}.\; \tsbPi
                    \right)
                    \tag*{where}
    \\
    \tsbPi \; = \; &
                     \atomIn{ok} \setenum{\clockT[d] < 20}
                     \;\sumExt 
    \\
                  & \atomIn{claim}\setenum{\clockT[d] < 20 \land \clockT[pay] > 7, \clockT[c]}.
                    \atomIn{rcpt}\setenum{\clockT[c]<3,\clockT[c]}.\atomOut{refund}\setenum{\clockT[c]<7}                
                    \;\sumExt 
    \\
                  & \atomIn{abort}
  \end{align*}
\end{exa}

\paragraph{Semantics.}

To define the behaviour of TSTs we use \emph{clock valuations},
which associate each clock with its value.
The state of the interaction between two TSTs is described
by a \emph{configuration} $(\tsbP,\clockN) \mid (\tsbQ,\clockE)$,
where the clock valuations $\clockN$ and $\clockE$ 
record (keeping the same pace) 
the time of the clocks in $\tsbP$ and $\tsbQ$, respectively.
The dynamics of the interaction is formalised as a transition relation
between configurations (\Cref{def:tst:semantics}).
This relation describes all and only the \emph{correct} interactions:
\eg, we do not allow time passing to
make unsatisfiable all the guards in an internal choice,
since doing so would prevent a participant from respecting her protocol.
In~\Cref{sec:tst-monitoring} we will study another semantics, 
which can also describe the behaviour of \emph{dishonest} participants
who do not respect their protocols.

\begin{defi}[\textbf{Clock valuations}]
  \label{def:clock-valuation}
We denote with $\Val = \Clocks \rightarrow \Realpos$ 
the set of \emph{clock valuations}. 
We use meta-variables $\clockN, \clockE, \ldots$ to range over $\Val$, %
and we denote with $\clockN[0], \clockE[0]$ the \emph{initial} clock valuations, 
which map each clock to zero.
% $\clockN : \Clocks \rightarrow \Realpos$,
% 
Given a clock valuation $\clockN$,
we define the following valuations:
\begin{itemize}

\item $\clockN + \delta$ 
increases each clock in $\clockN$ by the delay $\delta \in \Realpos$, \ie: 
\[
  (\clockN + \delta)(\clockT) = \clockN(\clockT) + \delta
  \tag*{for all $\clockT \in \Clocks$}
\]

\item $\reset{\clockN}{\resetR}$ \emph{resets} 
all the clocks in the set $\resetR \subseteq \Clocks$, \ie:
\[
\reset{\clockN}{\resetR}(\clockT) = \begin{cases}
  0 & \text{if $\clockT \in \resetR$} \\
  \clockN(\clockT) & \text{otherwise}
\end{cases}
\]

\end{itemize}
\end{defi}

\begin{defi}[\textbf{Semantics of guards}]
  For all guards $\guardG$,
  we define the set of clock valuations $\sem{\guardG}$
  inductively as follows, where $\circ \in \setenum{<,\leq,=,\geq,>}$:
  \[
  \begin{array}{lcl}
    \sem{\guardTrue} = \Val
    &
    \sem{\neg\guardG} = \Val \setminus \sem{\guardG} 
    \qquad
    &
    \sem{\guardG[1] \land \guardG[2]} = \sem{\guardG[1]} \cap \sem{\guardG[2]} 
    \\[5pt]
    \sem{\clockT \circ d} = \setcomp{\clockN}{\clockN(\clockT) \circ d} 
    \quad
    & & 
    \sem{\clockT - \clockTi \circ d} = \setcomp{\clockN}{\clockN(\clockT) - \clockN(\clockTi) \circ d} 
  \end{array}
  \]
\end{defi}\bigskip

\noindent Before defining the semantics of TSTs, 
we recall from \cite{BengtssonY03} 
some basic operations on \emph{sets} of clock valuations
(ranged over by~$\kindK,\kindKi,\ldots \subseteq \Val$).

\begin{defi}[\textbf{Past and inverse reset}]
  \label{def:past-invreset}
  For all sets $\kindK$ of clock valuations, 
  the set of clock valuations $\past{\kindK}$ (the \emph{past} of $\kindK$)
  and $\invReset{\kindK}{\resetT}$ (the \emph{inverse reset} of~$\kindK$)
  are defined as:
  \iftoggle{techreport}{%
  \[
  \past{\kindK} 
  \; = \; 
  \setcomp{\clockN}{\exists \delta \geq 0 : \clockN+\delta \in \kindK}
  \hspace{50pt}
  \invReset{\kindK}{\resetT} 
  \; = \; 
  \setcomp{\clockN}{\reset{\clockN}{\resetT} \in \kindK}
  \]
  }
  {%
  $\past{\kindK} = \setcomp{\clockN}{\exists \delta \geq 0 : \clockN+\delta \in \kindK}$,
  $\invReset{\kindK}{\resetT} = \setcomp{\clockN}{\reset{\clockN}{\resetT} \in \kindK}$.
  }
\end{defi}

\begin{defi}[\textbf{Semantics of TSTs}] \label{def:tst:semantics}
  \label{def:rdy}
  A \emph{configuration} is a term of the form 
  $(\tsbP,\clockN) \mid (\tsbQ,\clockE)$,
  where $\tsbP, \tsbQ$ are TSTs extended with \emph{committed choices}
  $\todo[\TsumI{\atomOut{a}}{\guardG,\resetR}{}] \tsbP$.
  The semantics of TSTs is a labelled relation $\smove{}$ over
  configurations (\Cref{fig:tst:s_semantics}), whose
  labels are either silent actions $\tau$,
  delays $\delta$, or branch labels,
  and where we define the set of clock valuations $\rdy{\tsbP}$ as: 
  \[
  \rdy{\tsbP} \; = \;
  \begin{cases}
    \past{\bigcup \sem{\guardG[i]}}
    & \text{if } \tsbP = \TSumInt[i \in I]{\atomOut[i]{a}}{\guardG[i],\resetR[i]}{\tsbP[i]}
    \\
    \Val
    & \text{if } \tsbP =  \SumExtRaw \cdots \text{ or } \tsbP =  \success 
    \\ 
    \emptyset
    & \text{otherwise }
  \end{cases}
  \]
  As usual, we write $\tsbP \smove{\alpha} \tsbPi$ as a shorthand for $(\tsbP,\alpha,\tsbPi) \in \smove{}$, with $\alpha \in \BLab \cup \setenum{\tau} \cup \Realpos$.
  Given a relation $\rightarrow$, we denote with $\rightarrow^*$ 
  its reflexive and transitive closure.
\end{defi}

\begin{figure}[t]
  \[
  \begin{array}{c}
    \begin{array}{cll}
      {( {\TsumI{\atomOut{a}}{\guardG,\resetR}{\tsbP}} \sumInt \tsbPi, \; \clockN)
        \;\smove{\tau}\;
        ({\todo[\TsumI{\atomOut{a}}{\guardG,\resetR}{}]{\tsbP}},\; \clockN)
      }
      \hspace{20pt}
      & \text{if } \clockN \in \sem{\guardG} 
      & \smallnrule{[$\sumInt$]}
      \\[10pt] 
      {({\todo[\TsumI{\atomOut{a}}{\guardG,\resetR}{}]{\tsbP}}, \; \clockN)
        \;\smove{\atomOut{a}}\;
        (\tsbP,\; \reset{\clockN}{\resetR})
      }
      & % \qquad \text{if } \clockN \models \guardG 
      & \smallnrule{[\bang]}    
      \\[10pt]
      {(\TsumE{\atomIn{a}}{\guardG,\resetR}{\tsbP} + \tsbPi, \; \clockN)
        \;\smove{\atomIn{a}}\;
        (\tsbP, \; \reset{\clockN}{\resetR})}
      & \text{if } \clockN \in \sem{\guardG} 
      & \smallnrule{[\qmark]}    
      \\[10pt]
      (\tsbP,\; \clockN)\smove{\; \delta \; }(\tsbP,\; \clockN+\delta)
      & \text{if }  \delta> 0 \ \land \ \clockN + \delta \in \rdy{\tsbP}
      \hspace{10pt}
      & \smallnrule{[Del]}    
    \end{array}
    \\[45pt]
    \irule{(\tsbP,\clockN) \smove{\; \tau \; } (\tsbPi,\clockNi)}
    {(\tsbP,\clockN) \mid (\tsbQ,\clockE)\smove {\; \tau \;}(\tsbPi,\clockNi) \mid (\tsbQ,\clockE)}
    \;\smallnrule{[S-$\sumInt$]}
    \hspace{30pt}
    \irule
    {(\tsbP,\clockN) \smove{\; \delta \; } (\tsbP,\clockNi) \quad 
      (\tsbQ,\clockE) \smove{\; \delta \; } (\tsbQ,\clockEi)}
    {(\tsbP,\clockN) \mid (\tsbQ,\clockE) \smove{\; \delta \;} 
      (\tsbP,\clockNi) \mid (\tsbQ,\clockEi)}
    \;\smallnrule{[S-Del]}
    \\[15pt]
    \irule
    {(\tsbP,\clockN) \smove{\; \atomOut{a} \; } (\tsbPi,\clockNi) \quad 
      (\tsbQ,\clockE) \smove{\; \atomIn{a} \; } (\tsbQi,\clockEi)}
    {(\tsbP,\clockN) \mid (\tsbQ,\clockE) \smove{\; \tau \;} 
      (\tsbPi,\clockNi) \mid (\tsbQi,\clockEi)}
    \;\smallnrule{[S-$\tau$]}
    \iftoggle{techreport}{}{%
    \\[15pt]
    \rdy{\TSumInt{\atomOut[i]{a}}{\guardG[i],\resetR[i]}{\tsbP[i]}} = \past{\bigcup \sem{\guardG[i]}}
    \hspace{12pt}
    \rdy{\SumExtRaw \cdots} = \rdy{\success} = \Val
    \hspace{12pt}
    \rdy{\todo[\TsumI{\atomOut{a}}{\guardG,\resetR}{}] \tsbP} = \emptyset
    }
  \end{array}
  \]
  \caption{Semantics of timed session types (symmetric rules omitted).}
  \label{fig:tst:s_semantics}
\end{figure}

We now comment the rules in~\Cref{fig:tst:s_semantics}.
The first four rules describe the behaviour of a TST in isolation. %
Rule~\nrule{[$\sumInt$]} allows a TST to commit to the branch
$\atomOut{a}$ of her internal choice, provided that the corresponding
guard is satisfied in the clock valuation $\clockN$.
This results in the term
$\todo[\TsumI{\atomOut{a}}{\guardG,\resetR}{}] \tsbP$,
which can only fire $\atomOut{a}$ 
through rule~\nrule{[\bang]},
without making time pass.
This term % $\todo[\TsumI{\atomOut{a}}{\guardG,\resetR}{}] \tsbP$ 
represents a state where the endpoint has committed to 
branch $\atomOut{a}$ in a specific time instant%
\footnote{
This is quite similar to the handling of internal choices 
in~\cite{Barbanera10ppdp,BCPZ15jlamp,BSZ14concur}.
In these works, an internal choice $\atomOut{a}.\tsbP \sumInt \tsbQ$
first commits to one of the branches (say, $\atomOut{a}.\tsbP$)
through an internal action, taking a transition to
a singleton internal choice $\atomOut{a}.\tsbP$.
In this state, only the action $\atomOut{a}$ is enabled ---
as in our $\todo[\TsumI{\atomOut{a}}{\guardG,\resetR}{}] \tsbP$.
}.
Rule~\nrule{[\qmark]} allows an external choice to fire any of its
input actions whose guard is satisfied.
Rule~\nrule{[Del]} allows time to pass;
this is always possible for external choices and success term, 
while for an internal choice we require that at least one of the 
guards remains satisfiable;
this is obtained through the function $\rdy{}$.
The last three rules deal with configurations.
Rule~\nrule{[S-$\sumInt$]} allows a TST
to commit in an internal choice.
Rule~\nrule{[S-$\tau$]} is the standard synchronisation rule \emph{\`a la} CCS.
Rule~\nrule{[S-Del]} allows time to pass, equally for both endpoints.

\begin{exa}
  \label{ex:tst:committed-choice}
  Let 
  $
  \tsbP = 
  \TsumI{\atomOut{a}}{}{} \sumInt 
  \TsumI{\atomOut{b}}{\clockT > 2}{}
  $,
  let 
  $
  \tsbQ = 
  % \TsumE{\atomIn{a}}{}{} \sumExt 
  \TsumE{\atomIn{b}}{\clockT > 5}{}
  $,
  and consider the computations:
  \begin{align}
    % \nonumber
    \label{eq:tst:committed-choice:0}
    (\tsbP,\clockN[0]) \mid (\tsbQ,\clockE[0]) 
    \smove{\, 7 \,} \smove{\, \tau \,} \,
    & (\todo[\TsumI{\atomOut{b}}{\clockT>2}{}],\clockN[0] + 7) \mid 
    (\tsbQ,\clockE[0] + 7)
    \smove{\, \tau \,} \,
    (\success,\clockN[0] + 7) \mid 
    (\success,\clockE[0] + 7)
    \\
    \label{eq:tst:committed-choice:1}
    (\tsbP,\clockN[0]) \mid (\tsbQ,\clockE[0]) 
    \smove{\, \delta \,} \smove{\, \tau \,} \,
    & (\todo[\TsumI{\atomOut{a}}{}{}],\clockN[0] + \delta) \mid 
    (\tsbQ,\clockE[0] + \delta)
    \\
    \label{eq:tst:committed-choice:2}
    (\tsbP,\clockN[0]) \mid (\tsbQ,\clockE[0]) 
    \smove{\, 3 \,} \smove{\, \tau \,} \,
    & (\todo[\TsumI{\atomOut{b}}{\clockT>2}{}],\clockN[0] + 3) \mid 
    (\tsbQ,\clockE[0] + 3)
  \end{align}
  The computation in~\eqref{eq:tst:committed-choice:0} reaches success,
  while the other two computations reach a deadlock state.
  In~\eqref{eq:tst:committed-choice:1},
  $\tsbP$ commits to the choice $\atomOut{a}$ after some delay $\delta$;
  at this point, time cannot pass 
  (because the leftmost endpoint is a committed choice),
  and no synchronisation is possible 
  (because the other endpoint is not offering~$\atomIn{a}$).
  In~\eqref{eq:tst:committed-choice:2},
  $\tsbP$ commits to $\atomOut{b}$ after $3$ time units;
  here, the rightmost endpoint would offer $\atomIn{b}$,
  but not in the time chosen by the leftmost endpoint.
  Note that, were we allowing time to pass in committed choices,
  then we would have obtained \eg that 
  $(\TsumI{\atomOut{b}}{\clockT > 2}{},\clockN[0]) \mid (\tsbQ,\clockE[0])$ 
  never reaches deadlock --- contradicting our intuition that
  these endpoints should not be considered compliant.
\end{exa}

Note that, even when $\tsbP$ and $\tsbQ$ have shared clocks,
the rules in~\Cref{fig:tst:s_semantics} ensure that there is no interference between them.
For instance, if a transition of $(\tsbP,\clockN)$ % via rule~\nrule{[?]}
resets some clock $\clockT$, this has no effect on a clock with the same name
in $\tsbQ$, \ie on a transition of $(\tsbP,\clockN) \mid (\tsbQ,\clockE)$.
Thus, without loss of generality  
we will assume that the clocks in $\tsbP$ and in $\tsbQ$ are disjoint.

\section{Compliance between TSTs} 
\label{sec:tst-compliance}

We extend to the timed setting
the standard progress-based compliance between (untimed) session types
\cite{Barbanera15mscs,Bartoletti15plabs,Castagna09toplas,Laneve07concur}.
If $\tsbP$ is compliant with $\tsbQ$, then
whenever an interaction between~$\tsbP$ and~$\tsbQ$ becomes stuck,
it means that both participants have reached the success state.
Intuitively, when two TSTs are compliant, 
the interaction between services correctly implementing%
\footnote{%
  The notion of ``correct implementation'' of a TST
  is orthogonal to the present work.
  A possible instance of this notion
  can be obtained by extending to the timed setting 
  the \emph{honesty} property of~\cite{BSTZ16lmcs}.
}
them will progress (without communication and time errors), 
until both services reach a success state.

\begin{defi}[\textbf{Compliance}] 
  \label{def:deadlock}
  \label{def:compliance}
  We say that $(\tsbP,\clockN) \mid (\tsbQ,\clockE)$
  is \emph{deadlock} whenever
  $(i)$ it is not the case that both $\tsbP$ and $\tsbQ$ are $\success$,
  % $(p \neq \success \;\lor\; q \neq \success)$
  and $(ii)$ there is no $\delta$ such that
  $(\tsbP,\clockN+\delta) \mid (\tsbQ,\clockE+\delta) \smove{\tau}$.
  We then write $(\tsbP,\clockN) \compliant (\tsbQ,\clockE)$
  whenever:
  \[
  (\tsbP,\clockN) \mid (\tsbQ,\clockE)
  \smove{}^* 
  (\tsbPi,\clockNi) \mid (\tsbQi,\clockEi)
  \quad \text{ implies } \quad 
  (\tsbPi,\clockNi) \mid (\tsbQi,\clockEi)
  \text{ not deadlock}
  \]
  We say that $\tsbP$ and $\tsbQ$ are \emph{compliant}
  whenever $(\tsbP,\clockN[0]) \compliant (\tsbQ,\clockE[0])$
  (in short, $\tsbP \compliant \tsbQ$).
\end{defi}

Note that item $(ii)$ of the definition of deadlock can be equivalently phrased as follows:
$(\tsbP,\clockN) \mid (\tsbQ,\clockE) \not\smove{\;\;\tau\,}$ 
(\ie, the configuration cannot do a $\tau$-move in the \emph{current} clock valuation),
and there does not exist any $\delta > 0$ such that
$(\tsbP,\clockN) \mid (\tsbQ,\clockE) \, \smove{\,\delta\,} \, \smove{\,\tau\,} \,$.

\begin{exa}
  The TSTs $\tsbP = \atomIn{a}\setenum{\clockT < 5} . \atomOut{b}\setenum{\clockT < 3}$
  and
  $\tsbQ = \atomOut{a}\setenum{\clockT < 2}.\atomIn{b}\setenum{\clockT < 3}$
  are compliant, 
  but $\tsbP$ is \emph{not} compliant with 
  $\tsbQi = \atomOut{a}\setenum{\clockT < 5}.\atomIn{b}\setenum{\clockT < 3}$.
  Indeed, if $\tsbQi$ outputs $\atomA$ at, say, time $4$, the configuration will reach a state where no actions are possible, and time cannot pass. 
  This is a deadlocked state, according to~\Cref{def:deadlock}.
\end{exa}

\begin{exa}
  Consider a customer of PayNow (\Cref{ex-paypal}) %
  who is willing to wait 10 days to receive the item she has paid for, 
  but after that she will open a claim. %
  Further, she will instantly provide PayNow with
  any documentation required. %
  The customer contract is described by the following TST,
  which is compliant with PayNow's $\tsbP$ in~\Cref{ex-paypal}: %
  \[
  \begin{array}{rl}
    \atomOut{pay}\setenum{\clockT[pay]}.
    (
    & \!\!\atomOut{ok}\setenum{\clockT[pay] < 10}
    \;\sumInt 
    \\
    & \!\!\atomOut{dispute}\setenum{\clockT[pay] = 10}.\atomOut{claim}\setenum{\clockT[pay] = 10}.\atomOut{rcpt}
    \setenum{\clockT[pay] = 10}.\atomIn{refund}
    )
  \end{array}
  \]
  % Indeed, the customer and PayNow's contracts are compliant.
\end{exa}\bigskip

\noindent Compliance between TSTs is more liberal than the untimed notion,
as it can relate terms which, when cleaned from all the time annotations,
would not be compliant in the untimed setting.
For instance, the following~\namecref{ex:tst-compliance:rec-vs-sumext} shows
that a recursive internal choice can be compliant with 
a \emph{non}-recursive external choice --- 
which can never happen in untimed session types. 

\begin{exa}
  \label{ex:tst-compliance:rec-vs-sumext}
  {%
    Let
    $
    \tsbP = \rec \tsbX {\big(
      \TsumI{\atomOut{a}}{}{} 
      \;\sumInt\; 
      \TsumI{\atomOut{b}}{\clockX \leq 1}{\TsumE{\atomIn{c}}{}{\tsbX}}
      \big)}
    $,
    $
    \tsbQ = \TsumE{\atomIn{a}}{}{} 
    \;\sumExt\; 
    \TsumE{\atomIn{b}}{\clockY \leq 1}
    {\TsumI{\atomOut{c}}{\clockY > 1}{\TsumE{\atomIn{a}}{}{}}}
    $.}
  We have that $\tsbP \compliant \tsbQ$.
  Indeed, if $\tsbP$ chooses the output $\atomOut{a}$,
  then $\tsbQ$ has the corresponding input, and they both succeed;
  instead, if $\tsbP$ chooses $\atomOut{b}$, then it will read $\atomIn{c}$
  when $\clockX>1$, and so at the next loop it is forced to choose
  $\atomOut{a}$, since the guard of $\atomOut{b}$ has become
  unsatisfiable.
\end{exa}

\Cref{def:coind-compliance} and~\Cref{lem:coind-compliance} below
coinductively characterise compliance between TSTs,
by extending to the timed setting the coinductive compliance 
for untimed session types in~\cite{Barbanera10ppdp}. %
Intuitively, an internal choice $\tsbP$ is compliant with $\tsbQ$ when 
\begin{inlinelist}
\item $\tsbQ$ is an external choice, 
\item for each output $\atomOut{a}$ that $\tsbP$ can fire after $\delta$
  time units,
  there exists a corresponding input $\atomIn{a}$
  that $\tsbQ$ can fire after $\delta$ time units, and
\item their continuations are coinductively compliant.
\end{inlinelist}
The case where $\tsbP$ is an external choice is symmetric. %

\begin{defi}
  \label{def:coind-compliance}
  We say $\relR$ is a \emph{coinductive compliance}
  iff $(\tsbP,\clockN) \relR (\tsbQ,\clockE)$ implies:
  \begin{enumerate}
    
  \item \label{def:coind-compliance:i} $ 
    \tsbP = \success \iff \tsbQ = \success
    $
    
    \vspace{2pt}
    
  \item\label{def:coind-compliance:ii} $
    \tsbP = \TSumInt[i \in I]{\atomOut[i]{a}}{\guardG[i],\resetR[i]}{\tsbP[i]} \implies 
    \clockN \in \rdy{\tsbP} 
    \;\land\;
    \tsbQ = 
    \TSumExt[j \in J]{\atomIn[j]{a}}{\guardG[j],\resetR[j]}{\tsbQ[j]} \;\land\; 
    $ \\
    $
    \forall \delta,i:~\clockN + \delta \in \sem{\guardG[i]} \implies \exists j : \atom[i]{a} = \atom[j]{a} 
    \land \clockE + \delta \in \sem{\guardG[j]} \land (\tsbP[i],\clockN+\delta[\resetR[i]]) \relR
    (\tsbQ[j],\clockE+\delta[\resetR[j]])
    $

    \vspace{2pt}

  \item\label{def:coind-compliance:iii} $
    \tsbP = \TSumExt[j \in J]{\atomIn[j]{a}}{\guardG[j],\resetR[j]}{\tsbP[j]} \implies 
    \clockE \in \rdy{\tsbQ}
    \;\land\;
    \tsbQ = 
    \TSumInt[i \in I]{\atomOut[i]{a}}{\guardG[i],\resetR[i]}{\tsbQ[i]} \;\land\;
    $ \\
    $
    \forall \delta,i:~\clockE + \delta \in \sem{\guardG[i]} \implies \exists j : \atom[i]{a} = \atom[j]{a} 
    \land \clockN + \delta \in \sem{\guardG[j]} \land (\tsbP[j],\clockN+\delta[\resetR[j]]) \relR
    (\tsbQ[i],\clockE+\delta[\resetR[i]])
    $
  \end{enumerate}
\end{defi}

\newcommand{\lemcoindcompliance}{%
  $
  \tsbP  \compliant \tsbQ \iff 
  \exists \relR \text{coinductive compliance} \; : \; 
  (\tsbP,\clockN[0]) \relR (\tsbQ,\clockE[0])
  $
}
\begin{reslemma}{lem:coind-compliance}
  \lemcoindcompliance
\end{reslemma}

The following~\namecref{th:compliance:decidable} establishes
decidability of compliance.
To prove it, we reduce the problem of checking $\tsbP \compliant \tsbQ$
to that of model-checking % with Uppaal~\cite{Uppaal04tutorial} 
deadlock freedom in a network of timed automata
constructed from $\tsbP$ and $\tsbQ$.

\begin{thm} \label{th:compliance:decidable}
  Compliance between TSTs is decidable.
\end{thm}
\begin{proof}
  We defer the actual proof after~\Cref{th:compliance} in~\Cref{sec:tst-to-ta}.
\end{proof}

\section{Admissibility of a compliant} 
\label{sec:tst-duality}

In the untimed setting, each session type $\tsbP$ admits a compliant, 
\ie there exists some $\tsbQ$ such that $\tsbP \compliant \tsbQ$.
For instance, we can compute $\tsbQ$
by simply swapping internal choices with external ones 
(and inputs with outputs) in~$\tsbP$
(this $\tsbQ$ is called the \emph{canonical dual} of $\tsbP$ in some papers \cite{Castagna09ppdp,GayV10}). %
A na\"{\i}ve attempt to extend this construction to TSTs can be
to swap internal with external choices, as in the untimed case,
and leave guards and resets unchanged. %
% Unfortunately, 
This construction does not work as expected,
as shown by the following~\namecref{ex:duality}.

\begin{exa}\label{ex:duality}
  Consider the following TSTs:
  \[
  \begin{array}{l}
    \tsbP[1] = \TsumI{\atomOut{a}}{\clockX \leq 2}{\TsumI{\atomOut{b}}{\clockX \leq 1}{}}
    \hspace{80pt}
    \tsbP[2] = \TsumI{\atomOut{a}}{\clockX \leq 2}{} \sumInt
    \TsumI{\atomOut{b}}{\clockX \leq 1}{\TsumE{\atomIn{a}}{\clockX \leq 0}{}}
    \\[5pt]
    \tsbP[3] = \rec {\tsbX}{\TsumE {\atomIn{a}} {\clockX \leq 1 \land \clockY \leq 1}
      {\TsumI{\atomOut{a}}{\clockX \leq 1,\setenum{\clockX}}{\tsbX} 
        % \sumInt \TsumI{\atomOut{b}}{\guardTrue,\setenum{\clockX,\clockY}}{\tsbX}
      }}
  \end{array}
  \]
  The TST $\tsbP[1]$ is not compliant with
  its na\"{\i}ve dual
  $\tsbQ[1] = \TsumE{\atomIn{a}}{\clockX \leq 2}{\TsumE{\atomIn{b}}{\clockX \leq 1}{}}$:
  even though $\tsbQ[1]$ can do the input $\atomIn{a}$ in the required
  time window, $\tsbP[1]$ cannot perform $\atomOut{b}$ 
  if $\atomOut{a}$ is performed after 1 time unit.
  For this very reason, no TST is compliant with $\tsbP[1]$. %
  Note instead that 
  $
  \tsbQ[1] \compliant
  \TsumI{\atomOut{a}}{\clockX \leq 1}{\TsumI{\atomOut{b}}{\clockX \leq 1}{}}
  $,
  which is \emph{not} its na{\"i}ve dual. %
  In $\tsbP[2]$,
  a similar deadlock situation occurs if the $\atomOut{b}$ branch is chosen,
  and so also $\tsbP[2]$ does not admit a compliant.
  The reason why $\tsbP[3]$ does not admit a compliant is more subtle:
  actually, $\tsbP[3]$ can loop until the clock $\clockY$
  reaches the value $1$; after this point, the guard $\clockY \leq 1$
  can no longer be satisfied, and then $\tsbP[3]$ reaches a deadlock.
\end{exa}

To establish when a TST admits a compliant,
we define a kind system which associates to each $\tsbP$ 
a set of clock valuations~$\kindK$ (called \emph{kind of $\tsbP$}).
The kind of a TST is unique, and each closed TST is kindable 
(\Cref{lem:every-tst-kindable}).
If $\tsbP$ has kind $\kindK$, then 
there exists some $\tsbQ$ such that, for all $\clockN \in \kindK$, 
the configuration $(\tsbP,\clockN) \mid (\tsbQ,\clockN)$
never reaches a deadlock (\Cref{th:dual-sound}).
Also the converse statement holds: if, for some $\tsbQ$, 
$(\tsbP,\clockN) \mid (\tsbQ,\clockN)$
never reaches a deadlock, then $\clockN \in \kindK$ (\Cref{th:dual-complete}).
Therefore, $\tsbP$ admits a compliant whenever 
the initial clock valuation $\clockN[0]$ belongs to~$\kindK$. 
We give a constructive proof of the correctness of the kind system, 
by showing a TST $\co{\tsbP}$ which we call the \emph{canonical compliant} of~$\tsbP$.

\begin{figure}[t]
  \iftoggle{techreport}{%
  \[
  \begin{array}{cl}
    {\Gamma \vdash \success: \Val} & \smallnrule{[T-$\success$]}
    \\[8pt]
    \irule{\Gamma \vdash \tsbP[i]: \kindK[i] \qquad \forall i \in I}
    {\Gamma \vdash \TSumExt[i \in I]{\atomIn[i]{a}}{\guardG[i],\resetT[i]}{\tsbP[i]}: 
      \bigcup_{i \in I} \past{\big(\sem {\guardG[i]} \cap \invReset{\kindK[i]}{\resetT[i]}\big)}} 
    & \smallnrule{[T-$\sumExt$]} 
    \\[15pt]
    \irule{\Gamma \vdash \tsbP[i]: \kindK[i] \qquad \forall i \in I}
    {\Gamma \vdash \TSumInt[i \in I]{\atomOut[i]{a}}{\guardG[i],\resetT[i]}{\tsbP[i]}:
      \big(\bigcup_{i \in I}\past{\sem{\guardG[i]}}\big) \setminus 
      \big(\bigcup_{i \in I}\past{(\sem{\guardG[i]} \setminus
      \invReset{\kindK[i]}{\resetT[i]})\big)}} 
    \hspace{10pt}
    & \smallnrule{[T-$\sumInt$]} 
    \\[15pt]
    {\Gamma,\tsbX:\kindK \vdash \tsbX:\kindK} & \smallnrule{[T-Var]}
    \\[8pt]
    \irule{\Gamma,\tsbX:\kindK \vdash \tsbP : \kindKi}
    {\Gamma \vdash \rec \tsbX \tsbP:\bigcup{\setcomp{\kindK[0]}{\exists \kindK[1]:~\Gamma,\tsbX:\kindK[0] 
          \vdash \tsbP : \kindK[1] \land \kindK[0] \subseteq \kindK[1]}}} & \smallnrule{[T-Rec]}
    % \\[15pt]
  \end{array}
  \]
  }
  {%
  \[
  \begin{array}{c}
    {\Gamma \vdash \success: \Val} \;\; \smallnrule{[T-$\success$]}
    \hspace{10pt}
    \irule{\Gamma \vdash \tsbP[i]: \kindK[i] \qquad \text{for }i \in I}
    {\Gamma \vdash \TSumExt[i \in I]{\atomIn[i]{a}}{\guardG[i],\resetT[i]}{\tsbP[i]}: 
      \bigcup_{i \in I} \past{\big(\sem {\guardG[i]} \cap \invReset{\kindK[i]}{\resetT[i]}\big)}} 
    \; \smallnrule{[T-$\sumExt$]} 
    \\[10pt]
    \irule{\Gamma \vdash \tsbP[i]: \kindK[i] \qquad \text{for }i \in I}
    {\Gamma \vdash \TSumInt[i \in I]{\atomOut[i]{a}}{\guardG[i],\resetT[i]}{\tsbP[i]}:
      \big(\bigcup_{i \in I}\past{\sem{\guardG[i]}}\big) \setminus 
      \big(\bigcup_{i \in I}\past{(\sem{\guardG[i]} \setminus
        \invReset{\kindK[i]}{\resetT[i]})\big)}} 
    \; \smallnrule{[T-$\sumInt$]} 
    \\[15pt]
    {\Gamma,\tsbX:\kindK \vdash \tsbX:\kindK} \;\; \smallnrule{[T-Var]}
    \hspace{10pt}
    \irule{\Gamma,\tsbX:\kindK \vdash \tsbP : \kindKi}
    {\Gamma \vdash \rec \tsbX \tsbP:\bigcup{\setcomp{\kindK[0]}{\exists \kindK[1]:~\Gamma,\tsbX:\kindK[0] 
          \vdash \tsbP : \kindK[1] \land \kindK[0] \subseteq \kindK[1]}}}
    \; \smallnrule{[T-Rec]}
  \end{array}
  \]
  \iftoggle{techreport}{}{\vspace{-10pt}}
  }
  \caption{Kind system for TSTs.}
  \label{fig:tsb:type-system}
  \iftoggle{techreport}{}{\vspace{-10pt}}
\end{figure}

\begin{defi}[\bf Kind system for TSTs]
  Kind judgements $\Gamma \vdash \tsbP: \kindK$
  are defined in \Cref{fig:tsb:type-system},
  where $\Gamma$ is a partial function which associates kinds to
  recursion variables. %
\end{defi}

Rule \nrule{[T-$\success$]} says that the success TST $\success$ admits 
a compliant in every $\clockN$: 
indeed, $\success$ is compliant with itself. %
The kind of an external choice is the union of the kinds of its branches
(rule \nrule{[T-$\sumExt$]}), where the kind of a branch is the past of those clock valuations
which satisfy both the guard and, after the reset, the kind of their continuation.
Internal choices are dealt with by rule \nrule{[T-$\sumInt$]}, 
which computes the difference between the union of the past of the guards and a set of error
clock valuations. %
The error clock valuations are those which can satisfy a guard but not
the kind of its continuation. Rule \nrule{[T-Var]} is standard.
Rule \nrule{[T-Rec]} looks for a kind which is preserved by unfolding of recursion 
(hence a fixed point).

\begin{exa}
  Recall $\tsbP[2]$ from~\Cref{ex:duality}.
  We have the following kinding derivation:
  \[
    \AXC{$\vdash \success: \Val$}
    \AXC{$\vdash \success: \Val$}
    \RL{\nrule{[T-$\sumExt$]}}
    \UIC{$\vdash \TsumI{\atomOut{a}}{\clockX \leq 0}{}:\past{\sem{\clockX \leq 0}} \cap \Val = \sem{\clockX \leq 0}$}
    \RL{\nrule{[T-$\sumInt$]}}
    \BIC{$\vdash \tsbP[2]:
      \big(\past{\sem{\clockX \leq 2}} \cup \past{\sem{\clockX \leq 1}}\big) \setminus
      \big(\past{\sem{\clockX \leq 2} \setminus \Val)} \cup \past{\sem{\clockX \leq 1} \setminus \sem{\clockX \leq 0}} \big)
      = \kindK$}
    \DisplayProof
  \]
  where $\kindK = \sem{(\clockX > 1) \land (\clockX \leq 2)}$.
  As noted in~\Cref{ex:duality},
  intuitively $\tsbP[2]$ has no compliant;
  this will be asserted by~\Cref{th:dual-complete} below,
  as a consequence of the fact that $\clockN[0] \not\in \kindK$. % 
  However, since $\kindK$ is non-empty, 
  \Cref{th:dual-sound} guarantees that
  there exists $\tsbQ$ such that
  $(\tsbP[2],\clockN) \compliant (\tsbQ,\clockN)$, 
  for all clock valuations $\clockN \in \kindK$. %
\end{exa}

The following~\namecref{lem:every-tst-kindable} 
states that \emph{every} closed TST is kindable, as well as uniqueness of kinding. %
We stress that being kindable does not imply admitting a compliant: %
this holds if and only if the initial clock valuation $\clockN[0]$ belongs to the kind.
% (see~\Cref{th:dual-sound} and~\Cref{th:dual-complete}).
Note that uniqueness of kinding holds at the \emph{semantic} level,
but the same kind can be represented syntactically in different ways.
In~\Cref{sec:comput-dual} we show that uniqueness of kinding 
may be obtained also at the \emph{syntactic} level,
by representing kinds as guards in normal form~\cite{BengtssonY03}.

\newcommand{\lemeverytstkindable}{%
  For all $\tsbP$ and $\Gamma$ with $\fv{\tsbP} \subseteq \dom{\Gamma}$, 
  there exists unique $\kindK$ such that $\Gamma \vdash \tsbP:\kindK$.
}

\begin{restheorem}[\textbf{Uniqueness of kinding}]{lem:every-tst-kindable}
  \lemeverytstkindable
\end{restheorem}

By exploiting the kind system we define the \emph{canonical compliant} of kindable TSTs. %
Roughly, we turn internal choices into external ones 
(without changing guards nor resets),
and external into internal, 
changing the guards so that the kind of continuations is preserved. %

\begin{defi}[\textbf{Canonical compliant}] \label{def:dual}
  For all kinding environments $\Gamma$ and $\tsbP$ kindable in $\Gamma$, 
  we define the TST $\dual[\Gamma]{\tsbP}$ 
  in~\Cref{fig:dual}.
  We will abbreviate $\dual[\Gamma]{\tsbP}$ as $\dual{\tsbP}$ when $\Gamma = \emptyset$.
\end{defi}

\begin{figure}[t]
  \[
  \begin{array}{rcll}
    \dual[\Gamma]{\success} & = & \success 
    \\[2pt]
    \dual[\Gamma]{\TSumExt[i \in I]{\atomIn[i]{a}}{\guardG[i],\resetT[i]}{\tsbP[i]}} 
    & = &
    \TSumInt[i \in I]{\atomOut[i]{a}}{\guardG[i] \land \invReset{\kindK[i]}{\resetT[i]},\resetT[i]}
    {\dual[\Gamma]{\tsbP[i]}} 
    & \text{ if } \Gamma \vdash \tsbP[i]:\kindK[i]
    \\[2pt]
    \dual[\Gamma]{\TSumInt[i \in I]{\atomOut[i]{a}}{\guardG[i],\resetT[i]}{\tsbP[i]}} 
    & = &
    \TSumExt[i \in I]{\atomIn[i]{a}}{\guardG[i],\resetT[i]}{\dual[\Gamma]{\tsbP[i]}} 
    \\[2pt]
    \dual[\Gamma]{\tsbX} 
    & = & 
    \tsbX 
    & \text{ if } \Gamma(\tsbX) \text{ defined}
    \\[2pt]
    \dual[\Gamma]{\rec {\tsbX}{\tsbP}} & = & \rec {\tsbX}{\dual[\Gamma\setenum{\bind{\tsbX}{\kindK}}]
      {\tsbP}} & \text{ if } \Gamma \vdash \rec {\tsbX}{\tsbP}:\kindK
  \end{array}
  \]
  \iftoggle{techreport}{}{\vspace{-10pt}}
  \caption{Canonical compliant of a TST.}
  \label{fig:dual}
  \iftoggle{techreport}{}{\vspace{-10pt}}
\end{figure}

The following~\namecref{th:dual-sound} 
states the soundness of the kind system: %
is particular, if the initial clock valuation $\clockN[0]$ belongs 
to the kind of $\tsbP$, then $\tsbP$ admits a compliant.

\newcommand{\thdualsound}{%
  If $\,\vdash \tsbP : \kindK$ and $\clockN \in \kindK$, 
  then
  $(\tsbP,\clockN) \compliant (\dual{\tsbP},\clockN)$.
}

\begin{restheorem}[\textbf{Soundness}]{th:dual-sound}
  \thdualsound
\end{restheorem}

\begin{exa} \label{ex:dual}
  Recall 
  $\tsbQ[1] = \TsumE{\atomIn{a}}{\clockX \leq 2}{\TsumE{\atomIn{b}}{\clockX \leq 1}{}}$ 
  from~\Cref{ex:duality}. %
  We have
  $
  \dual{\tsbQ[1]} 
  \; = \;
  \TsumI{\atomOut{a}}{\clockX \leq 1}{\TsumI{\atomOut{b}}{\clockX \leq 1}{}}
  $.
  Since $\vdash \tsbQ[1] : \kindK = \sem{\clockX\leq1}$ 
  and $\clockN[0] \in \kindK$, 
  by~\Cref{th:dual-sound} we have that
  $\tsbQ[1] \compliant \dual{\tsbQ[1]}$,
  as anticipated in~\Cref{ex:duality}.
\end{exa}

The following~\namecref{th:dual-complete} 
states the kind system is also \emph{complete}:
in particular, if $\tsbP$ 
admits a compliant, then the clock valuation $\clockN[0]$ belongs to 
the kind of $\tsbP$.

\newcommand{\thdualcomplete}{%
  If $\,\vdash \tsbP : \kindK$
  and $\exists \tsbQ,\clockE .\; (\tsbP,\clockN) \compliant (\tsbQ,\clockE)$,
  then $\clockN \in \kindK$.
}

\begin{restheorem}[\textbf{Completeness}]{th:dual-complete}
  \thdualcomplete
\end{restheorem}

Compliance is not transitive, in general:
however, \Cref{th:dual-transitive} below states 
that transitivity holds when passing through the canonical compliant.

\begin{reslemma}{lem:compliance-trans}
  For all $\tsbP,\tsbQ,\clockN,\clockE$ and $\tsbPi,\clockNi$ such that $\vdash \tsbPi:\kindK$ and
  $\clockNi \in \kindK$:
  \[
  (\tsbP,\clockN)\compliant(\tsbPi,\clockNi) \;\land\; (\dual{\tsbPi},\clockNi)\compliant(\tsbQ,\clockE) 
  \;\implies\;
  (\tsbP,\clockN)\compliant(\tsbQ,\clockE) 
  \]
\end{reslemma}

\newcommand{\thdualtransitive}{%
  If $\tsbP \compliant \tsbPi$ and $\dual{\tsbPi} \compliant \tsbQ$, 
  then $\tsbP \compliant \tsbQ$.
}

\begin{thm}[\bf Transitivity of compliance] \label{th:dual-transitive}
  \thdualtransitive
\end{thm}
\begin{proof}
  Straightforward after~\Cref{lem:compliance-trans}.
\end{proof}

\section{Computability of the canonical compliant}
\label{sec:comput-dual}

In this~\namecref{sec:comput-dual} we show that 
the canonical compliant construction is computable. 
To prove this, we first show the decidability of kind inference
(\Cref{th:kind-inference:decidable}). %
This fact is not completely obvious, because
the cardinality of the set of kinds is $2^{2^{\aleph_0}}$; %
however, the kinds constructed by our inference rules
can always be represented syntactically by guards. % (as in~\Cref{def:guards}).

\begin{figure}
  \[
  \begin{array}{c}
    \Gamma \vdash_I \success: \Val
    \;\; \nrule{[I-$\success$]}
    \hspace{10pt}
    \irule{\Gamma \vdash_I \tsbP[i]: \kindK[i] \qquad \forall i \in I}
    {\Gamma \vdash_I \TSumInt[i \in I]{\atomOut[i]{a}}{\guardG[i],\resetT[i]}{\tsbP[i]}:
      \big(\bigcup_{i \in I}\past{\sem{\guardG[i]}}\big) \setminus 
      \big(\bigcup_{i \in I}\past{(\sem{\guardG[i]} \setminus
        \invReset{\kindK[i]}{\resetT[i]})\big)}} 
    \; \nrule{[I-$\sumInt$]} 
    \\[12pt]
    \Gamma,\tsbX:\kindK \vdash_I \tsbX:\kindK 
    \;\; \nrule{[I-Var]}
    \hspace{15pt}
    \irule{\Gamma \vdash_I \tsbP[i]: \kindK[i] \qquad \forall i \in I}
    {\Gamma \vdash_I \TSumExt[i \in I]{\atomIn[i]{a}}{\guardG[i],\resetT[i]}{\tsbP[i]}: 
      \bigcup_{i \in I} \past{\big(\sem {\guardG[i]} \cap \invReset{\kindK[i]}{\resetT[i]}\big)}} 
    \; \nrule{[I-$\sumExt$]} 
    \\[12pt]
    \Gamma \vdash_I \rec \tsbX \tsbP : \bigsqcap_{i \geq 0} \hat{F}_{\Gamma,\tsbX,\tsbP}^i (\Val)
    \;\; \nrule{[I-Rec]}
    \quad \textit{where } \hat{F}_{\Gamma,\tsbX,\tsbP}(\kindK) = \kindKi \text{ iff }\Gamma,\tsbX:\kindK \vdash_I \tsbP:\kindKi
  \end{array}
  \]
  \caption{Kind inference rules.}
  \label{fig:tsb:kind-inference}
\end{figure}

To prove the decidability of kind inference we will use the fundamental notion of clock region, 
introduced in \cite{Alur94theory}
(although our definition is more similar to the one in \cite{HenzingerNSY94}).

\begin{defi}[\textbf{Region}]
  For all $d \in \Nat$, a \emph{$d$-region} is a set of clock valuations $\kindK$ where:
  \begin{enumerate}
  \item there exists a guard $\guardG$, with all constants bounded by $d$, such that
    $\sem{\guardG} = \kindK$;
  \item if $\kindK$ properly contains a $d$-region $\kindKi$, then $\kindKi$ is empty.
  \end{enumerate}
  We call \emph{$d$-zones}\footnote{The term ``zone'' is often referred to convex sets of clocks, 
    while our definition includes non-convex ones.}
  (ranged over by $\zoneZ,\zoneZi,\hdots$) the finite unions of $d$-regions.
\end{defi}

It is well known that, for all $d$, the set of $d$-regions is finite \cite{Alur94theory}, and 
so also the induced set of $d$-zones.
Further, $d$-zones are closed under past, inverse reset, union, intersection and difference,
and for every guard $\guardG$ with all constants below $d$, 
$\sem{\guardG}$ is a $d$-zone~\cite{HenzingerNSY94,BengtssonY03}.
By the above, it follows immediately that the set of $d$-zones, ordered by set inclusion, 
forms a finite complete lattice.
In our case, $d$ will be the greatest constant in the TST under analysis.
% , and we will omit it. 

The following~\namecref{lem:domain-theory} recalls some well known 
facts about functions over complete lattices and their fixed points.
These facts will be used to prove that kind inference is decidable.

\begin{reslemma}{lem:domain-theory}
  Let $(D,\leq)$ be a complete lattice with top $\top$, 
  and let $f$ be a monotonic endofunction over $D$. 
  Then:
  \begin{enumerate}
    
  \item \label{item:finite-monotonic-is-cocontinuous}
    if $D$ is finite, then $f$ is cocontinuous.

  \item \label{item:gfp-smaller-glb}
    \iftoggle{lmcs}{%
      $\gfp f \leq \bigsqcap_i f^i(\top)$.
    }
    {%
      \[
      \gfp f \leq \bigsqcap_i f^i(\top)
      \]
    }

  \item \label{item:sangiorgi}
    if $f$ is cocontinuous, then
    \iftoggle{lmcs}{\;%
      $\gfp f \; = \; \bigsqcap_i f^i(\top)$
    }
    {%
      \[
      \textstyle
      \gfp f \; = \; \bigsqcap_i f^i(\top)
      \]
    }

  \end{enumerate}
\end{reslemma}

We start by showing that the kind obtained by rule~\nrule{[T-Rec]} in~\Cref{fig:tsb:type-system} can be characterised
as the greatest fixed point (over the lattice $(2^{\Val},\subseteq)$) 
of a functional, defined below.

\begin{defi}\label{def:kinding-functional}
  Given $\Gamma,\tsbX,\tsbP$, with $\tsbX$ \emph{not} in the domain of $\Gamma$, we define: 
  \[F_{\Gamma,\tsbX,\tsbP}(\kindK) = \kindKi 
  \;\; \text{ whenever } \;\; 
  \Gamma,\tsbX:\kindK \vdash \tsbP:\kindKi\]
  The subscript of $F_{\Gamma,\tsbX,\tsbP}$ will be omitted when clear from the context.
\end{defi}  

Note that, by uniqueness of kinding, $F$ is a function;
further, \Cref{lem:every-tst-kindable} ensures that $F$ is total
when $\fv{\tsbP} \subseteq \dom{\Gamma} \cup \setenum{\tsbX}$.
The following~\namecref{lem:gamma-monotonicity} states that $F$ is monotonic.
Then, the Knaster-Tarski fixed point theorem~\cite{tarski1955}
ensures that~\nrule{[T-Rec]} yields the gfp of $F$.

\begin{reslemma}{lem:gamma-monotonicity}
  The function $F_{\Gamma,\tsbX,\tsbP}$ is monotonic,
  for all $\Gamma,\tsbX,\tsbP$ with $\fv{\tsbP} \subseteq \dom{\Gamma} \cup \setenum{\tsbX}$.
\end{reslemma}

To prove decidability of the kinding relation, we introduce in~\Cref{fig:tsb:kind-inference}
an alternative set of rules, with judgements of the form $\Gamma \vdash_I \tsbP : \kindK$.
\Cref{lem:judments-equivalence} below shows that the kind relations $\vdash$ and $\vdash_I$
are equivalent.
The new set of rules can be exploited as a kind inference algorithm:
in particular, rule~\nrule{[I-Rec]} allows for computing 
the kind of a recursive TST $\rec \tsbX \tsbP$
by evaluating the non-increasing sequence $\hat{F}_{\Gamma,\tsbX,\tsbP}^i (\Val)$
until it stabilizes.
The following~\namecref{lem:inferred-kinds-are-zones} states that
the kinds inferred through the relation $\vdash_I$ are zones.
By~\cite{BengtssonY03} it follows that
the kind of a TST can always be represented as a guard.

\begin{lem}\label{lem:inferred-kinds-are-zones}
  If $\Gamma \vdash_I \tsbP : \kindK$,
  for some $\Gamma$ which maps variables to zones,
  then $\kindK$ is a zone.
\end{lem}
\begin{proof}
  Easy, by induction on the structure of $\tsbP$ and inspection of the kind inference rules (exploiting finiteness of the set of zones).
\end{proof}

\begin{reslemma}{lem:judments-equivalence}
  For all $\Gamma$ mapping variables to zones,
  and for all $\tsbP$ such that $\fv{\tsbP} \subseteq \dom{\Gamma}$:
  \[
  \Gamma \vdash \tsbP:\kindK \iff \Gamma \vdash_I \tsbP:\kindK
  \]
\end{reslemma}

\newcommand{\thkindinferencedecidable}{Kind inference is decidable.}
\begin{thm}\label{th:kind-inference:decidable}
  \thkindinferencedecidable
\end{thm}
\begin{proof}
  By~\Cref{lem:judments-equivalence}, kinds of closed TSTs 
  can be inferred by the rules in~\Cref{fig:tsb:kind-inference}. %
  All the operations between zones used in~\Cref{fig:tsb:kind-inference}
  (except $\hat{F}$) are well know to be computable~\cite{HenzingerNSY94,BengtssonY03}. %
  By finiteness of the set of zones, also $\hat{F}$ is computable. %
\end{proof}

The following~\namecref{th:dual:computable} states that 
the canonical compliant construction is computable.

\newcommand{\thdualcomputable}{The function $\dual{\cdot}$ is computable.}
\begin{thm} % [\bf Computability of the canonical compliant]
  \label{th:dual:computable}
  \thdualcomputable
\end{thm}
\begin{proof}
  It follows by the fact that all the operations in \Cref{def:dual} are computable.
\end{proof}

\section{Subtyping} \label{sec:tst-subtyping}

In this~\namecref{sec:tst-subtyping} we study the semantic subtyping preorder,
which is a sound and complete model of the Gay and Hole subtyping
relation (in reverse order) for untimed session types~\cite{Barbanera15mscs}.
Intuitively, $\tsbP$ is subtype of $\tsbQ$ if every $\tsbQi$ compliant
with $\tsbQ$ is compliant with $\tsbP$, too.

\begin{defi}[\textbf{Semantic subtyping}] 
  \label{def:subtyping}
  For all TSTs $\tsbP$, we define the set $\compset{(\tsbP,\clockN)}$ as:
  \[
  \compset{(\tsbP,\clockN)}
  \; = \;
  \setcomp{(\tsbQ,\clockE)}{(\tsbP,\clockN) \compliant (\tsbQ,\clockE)}
  \]
  Then, we define the preorder $\sqsubseteq$ between TSTs as follows:
  \[
  \tsbP \sqsubseteq \tsbQ
  \hspace{20pt}
  \text{whenever}
  \hspace{20pt}
  \compset{(\tsbP,\clockN[0])} \supseteq
  \compset{(\tsbQ,\clockN[0])}
  \]
\end{defi}

The following~\namecref{lem:dual-minimal} states that, 
as in the untimed setting, the canonical compliant of $\tsbP$
is the maximum (\ie, the most ``precise'') in the set of TSTs 
compliant with $\tsbP$. % (if not empty).

\newcommand{\lemdualminimal}{
  $
  \tsbQ \compliant \tsbP \implies \tsbQ \sqsubseteq \dual{\tsbP}
  $
}

\begin{thm} \label{lem:dual-minimal}
  \lemdualminimal
\end{thm}
\begin{proof}
  Assume that $\tsbQ \compliant \tsbP$
  and $\vdash \tsbP : \kindK$,
  and let $(\tsbR,\clockE) \in \compset{(\co{\tsbP},\clockN[0])}$.
  Since $(\tsbQ,\clockN[0]) \compliant (\tsbP,\clockN[0])$,
  by~\Cref{th:dual-complete} we obtain $\clockN[0] \in \kindK$.
  Therefore, by
  $(\co{\tsbP},\clockN[0]) \compliant (\tsbR,\clockE)$
  and \Cref{lem:compliance-trans}
  we conclude that
  $(\tsbQ,\clockN[0]) \compliant (\tsbR,\clockE)$.
  So, $(\tsbR,\clockE) \in \compset{(\tsbQ,\clockN[0])}$,
  from which the thesis follows.
\end{proof}

The following~\namecref{th:compliance-subtyping} 
reduces the problem of deciding $\tsbP \sqsubseteq \tsbQ$
to that of checking compliance between $\tsbP$ and $\dual{\tsbQ}$,
when $\tsbQ$ admits a compliant
(otherwise $\compset{(\tsbQ,\clockE[0])} = \emptyset$, so $\tsbQ$ is supertype of every $\tsbP$).
Since compliance, the canonical compliant construction, 
and checking the admissibility of a compliant 
are all decidable (\Cref{th:compliance:decidable}, \Cref{th:dual:computable}),
this implies decidability of subtyping (\Cref{th:subtyping-decidable}).

\newcommand{\thcompliancesubtyping}{%
  For all TSTs $\tsbP,\tsbQ$:
  \[
  \tsbP \sqsubseteq \tsbQ
  \;\; \iff \;\;
  \begin{cases}
    \tsbP \compliant \dual{\tsbQ}
    & \text{if $\tsbQ$ admits a compliant} \\
    \true 
    & \text{otherwise}
  \end{cases}
  \]
}

\begin{thm}\label{th:compliance-subtyping}
  \thcompliancesubtyping
\end{thm}
\begin{proof}
  If $\tsbQ$ does not admit a compliant then the thesis is trivial,
  so assume that $\compset{(\tsbQ,\clockE[0])} \neq \emptyset$.
  For the $(\Rightarrow)$ direction,
  assume that $\tsbP \sqsubseteq \tsbQ$.
  Since $\tsbQ$ admits a compliant, by~\Cref{th:dual-complete} there exists some
  $\kindK$ such that $\vdash \tsbQ : \kindK \ni \clockN[0]$.
  By~\Cref{th:dual-sound}, it follows that
  $\dual{\tsbQ} \compliant \tsbQ$.
  Then, by~\Cref{def:subtyping} we conclude that
  $\tsbP \compliant \dual{\tsbQ}$.
  For the $(\Leftarrow)$ direction,
  assume that $\tsbP \compliant \dual{\tsbQ}$,
  and let $\tsbQi$ be such that $\tsbQi \compliant \tsbQ$.
  Then, by~\Cref{th:dual-transitive} we conclude that
  $\tsbQi \compliant \tsbP$, from which the thesis follows.
\end{proof}

\begin{thm}[\textbf{Decidability of subtyping}]
  \label{th:subtyping-decidable}
  Subtyping between TSTs is decidable.
\end{thm}
\begin{proof}
  Immediate consequence of~\Cref{th:compliance:decidable} 
  and~\Cref{th:compliance-subtyping}.
\end{proof} 

Unlike in the untimed case, the canonical compliant construction is not involutive,
\ie $\dual{\dual{\tsbP}}$ is not equal to $\tsbP$, in general. 
However, $\tsbP$ and $\dual{\dual{\tsbP}}$ are still strongly related, 
as they have the same set of compliant TSTs, in every $\clockN$ in the kind of $\tsbP$
(\Cref{th:weak-dual-convolutive}). %
By~\Cref{def:subtyping}, this implies that $\tsbP \sqsubseteq \sqsupseteq \dual{\dual{\tsbP}}$,
for all kindable $\tsbP$.

\newcommand{\thweakdualconvolutive}{
  Let $\,\vdash \tsbP : \kindK$ and $\clockN \in \kindK$. Then:
  $
  \compset{(\tsbP,\clockN)} = \compset{(\dual{\dual {\tsbP}},\clockN)}
  % \text{ and } (\dual{\dual {\tsbP}},\clockN) \sqsubseteq (\tsbP,\clockN)
  $.
}

\begin{thm} \label{th:weak-dual-convolutive}
  \thweakdualconvolutive
\end{thm}
\begin{proof}
  Suppose that $\,\vdash \tsbP : \kindK$ and $\clockN \in \kindK$. 
  By~\Cref{th:dual-sound}: 
  \begin{equation} \label{eq:weak-dual-convolutive:aux1}
    (\tsbP,\clockN) 
    \; \compliant \;
    (\dual{\tsbP},\clockN)
  \end{equation} 
  Assume that $\vdash \dual{\tsbP} : \kindKi$.
  By~\eqref{eq:weak-dual-convolutive:aux1} and~\Cref{th:dual-complete} 
  it follows that $\clockN \in \kindKi$. 
  By repeating the same argument twice, we also obtain that: 
  \begin{align} 
    \label{eq:weak-dual-convolutive:aux2}
    (\dual{\tsbP},\clockN) 
    & \; \compliant \;
    (\dual{\dual{\tsbP}},\clockN)
    \\
    \label{eq:weak-dual-convolutive:aux3}
    (\dual{\dual{\tsbP}},\clockN) 
    & \; \compliant \;
    (\dual{\dual{\dual{\tsbP}}},\clockN)
  \end{align}
  To prove $\compset{(\tsbP,\clockN)} \subseteq \compset{(\dual{\dual {\tsbP}},\clockN)}$,
  let $(\tsbQ,\clockE) \in \compset{(\tsbP,\clockN)}$.
  By applying~\Cref{lem:compliance-trans} on
  $(\tsbQ,\clockE) \compliant (\tsbP,\clockN)$
  and on~\eqref{eq:weak-dual-convolutive:aux2},
  we obtain
  $(\tsbQ,\clockE) \compliant (\dual{\dual {\tsbP}},\clockN)$.

  \smallskip\noindent
  To prove $\compset{(\tsbP,\clockN)} \supseteq \compset{(\dual{\dual {\tsbP}},\clockN)}$,
  let $(\tsbQ,\clockE) \in \compset{(\dual{\dual{\tsbP}},\clockN)}$.
  By applying~\Cref{lem:compliance-trans}
  on~\eqref{eq:weak-dual-convolutive:aux1} and~\eqref{eq:weak-dual-convolutive:aux3}
  (and using commutativity of compliance),
  we obtain:
  \begin{equation} \label{eq:weak-dual-convolutive:aux4}
    (\dual{\dual{\dual{\tsbP}}},\clockN) 
    \; \compliant \;
    (\tsbP,\clockN)
  \end{equation}
  Finally, by applying~\Cref{lem:compliance-trans} on 
  $(\tsbQ,\clockE) \compliant (\dual{\dual{\tsbP}},\clockN)$
  and on~\eqref{eq:weak-dual-convolutive:aux4},
  we conclude that
  $(\tsbQ,\clockE) \compliant (\dual{\dual {\tsbP}},\clockN)$.
\end{proof}

\section{Case study: Paypal User Agreement}
\label{sec:ex-paypal-full}

As a case study, we formalise as a TST (part of)
the ``protection for buyers'' section of the PayPal User Agreement~\cite{PayPal},
which regulates the interaction between Paypal and buyers
in trouble during online purchases. %
When a buyer has not received the item they have paid for ($\atom{inr}$), 
or if they have received something significantly
different from what was described ($\atom{snad})$, 
they can 
open a \emph{dispute}. %
The dispute can be opened within $180$ days ($\clockT[pay] < 180$) of
the payment date ($\atom{pay}$). %
After opening the dispute, the buyer and the seller may try to solve
the problem, or it might be the case that the item finally arrives;
otherwise, if an agreement ($\atom{ok}$) is not found within $20$ days
($\clockT[inr] < 20$), the buyer can escalate the dispute to a claim
($\atom{claimINR}$,$\atom{claimSNAD}$). %
However, in case of an item not received, the buyer must wait at least
$7$ days from the date of payment to escalate the dispute
($\clockT[pay] > 7$). %
Upon not reaching an agreement, if still the buyer does not escalate
the dispute to a claim within $20$ days ($\clockT[pay] > 20$), PayPal
will close the dispute ($\atom{close}$). %

During the claim process, PayPal may require the buyer to provide
documentation to support the claim, for instance receipts
($\atom{rcpt}$) or photos ($\atom{photo}$), and the buyer must comply
in a \emph{timely manner} to what they are required to do. %
For SNAD claims, if the claim is accepted, PayPal may require the
buyer to ship the item back to the Seller, to PayPal, or to a third
party and to provide proof of delivery. In case the item is
counterfeit, the item will be destroyed ($\atom{destroy}$) and not
shipped back to the seller($\atom{sendBack}$). After that, the buyer
will be refunded. %
In some cases, the buyer is not eligible for a refund ($\atom{notEligible}$).

We can formalise this agreement as the following TST:
\[
\small
\begin{array}{l}
  \atomIn{pay} \setenum{\guardTrue, \clockT[pay]}.(
  \hspace{5pt} \atomIn{ok} 
  \\ \hspace{10pt} + \; \atomIn{inr}\setenum{\clockT[pay]<180,\clockT[inr]}.(\atomIn{ok} \setenum{\clockT[inr]< 20}+  \; \atomIn{close}\setenum{\clockT[pay]\geq 20}
  \\\hspace{40pt} + \; \atomIn{claimINR}\setenum{\clockT[inr]< 20 \land \clockT[pay] > 7, \clockT[c]}.
  \hspace{0pt}\atomIn{rcpt}.(\atomOut{refund} \oplus \atomOut{notEligible} )  \\          
  \hspace{10pt} + \; \atomIn{snad}\setenum{\clockT[pay]<180,\clockT[snad]}.(\atomIn{ok}\setenum{\clockT[snad]< 20}+  \; \atomIn{close}\setenum{\clockT[pay]\geq 20}\\
  \hspace{40pt} + \;  \atomIn{claimSNAD}\setenum{\clockT[snad]<20, \clockT[c]}.\atomIn{photo}. \\
  \hspace{50pt} ( \, \atomOut{sendBack}.\atomIn{ackSendBack}.\atomOut{refund} 
   \oplus \;  \atomOut{destroy}.\atomIn{ackDestroy}.\atomOut{refund}  \oplus \atomOut{notEligible} ))
\end{array}
\]\medskip

\noindent Let us consider a possible buyer Alice, who wants to see if she may
entrust PayPal for her transactions. %
Alice is willing to wait $10$ days to receive the item she has paid
for, but after that she will open a claim.  She will readily provide
PayPal with every documentation they may need in order to issue the
refund. %
In case she receives an item significantly different from what she has
paid for, she will complain to PayPal by opening a claim as soon as
the item is received. %
Alice will timely comply to do whatever PayPal requires (either to
destroy the item or to send it back) in order to be refunded. %  
Alice's requirements can be formalised as the following TST: %
\[
\small
\begin{array}{l}
  \atomOut{pay}\setenum{\guardTrue, \clockT[pay]}.
  \hspace{0pt} (\; \atomOut{ok}\setenum{\clockT[pay] < 10}
  \\  \hspace{10pt} \oplus \atomOut{inr}\setenum{\clockT[pay] = 10}.\atomOut{claimINR}\setenum{\clockT[pay]= 10}.\atomOut{rcpt}\setenum{\clockT[pay]= 10}.(\atomIn{refund}+\atomIn{notEligible}) \; 
  \\  \hspace{10pt} \oplus \atomOut{snad}\setenum{\clockT[pay]<10,\clockT[snad]}.\atomOut{claimSNAD}\setenum{\clockT[snad]= 0}.\atomOut{photo}\setenum{\clockT[snad] = 0}.
  \\  \hspace{20pt} (\, \atomIn{sendBack}\setenum{\clockT[c]}.\atomOut{ackSendBack}\setenum{\clockT[c]<3}.\atomIn{refund} 
   \\ \hspace{20pt} + \; \atomIn{destroy}\setenum{\clockT[c]}.\atomOut{ackDestroy}\setenum{\clockT[c]<3}.\atomIn{refund} + \atomIn{notEligible}\, ) ) 
\end{array}
\]\medskip

\noindent Alice's and PayPal's TSTs are compliant, according to~\Cref{def:compliance}.
However, we can see that PayPal's TST lacks some important details:
what does it means \emph{timely comply to what is required}? 
And, most importantly: how long will it take for a buyer to be refunded?
Without a deadline on the $\atomOut{refund}$ action, 
Alice may possibly wait forever.

\section{Encoding TSTs into Timed Automata} \label{sec:tst-to-ta}

We define a semantic-preserving encoding 
of TSTs into timed automata (\Cref{mapping}),
and we exploit it to devise an effective procedure to 
decide compliance (\Cref{th:compliance}).

\subsection{Timed automata} \label{sec:ta}

A timed automaton (\TA) \cite{Alur94theory} is a
non-deterministic automaton with a finite set of clocks,
used to put constraints on when to take an edge, or to stay in a location.
We denote with $\LabAut = \ActOut \cup \ActIn \cup \setenum{\tau}$
the set of labels, ranged over by $\labLtau,\labLitau,\ldots$.

\begin{defi}[\textbf{Timed automaton~\cite{Alur94theory}}]
  \label{ta}
  A \TA is a tuple
  $\autA = (\Loc, \UrgLoc, \locL[0], \Edg, \Inv )$ where:
  $\Loc$ is a finite set of \emph{locations}; 
  $\UrgLoc \subset \Loc$ is the set of \emph{urgent} locations; 
  $\locL[0] \in \Loc$ is the \emph{initial} location;
  $\Edg \subseteq \Loc \times \LabAut \times \GuardG[\Clocks] \times \powset
  \Clocks \times \Loc$
  is a set of edges; and $\Inv : \Loc \rightarrow \GuardG[\Clocks]$ is the
  \emph{invariant} function.
\end{defi}

To define the semantics of \TA, we consider a \emph{network} of 
{\TA} which can take internal transitions or synchronize on channels. 
For our purposes, we just need networks of \emph{two} {\TA}
interacting through \emph{binary} channels, which correspond to actions $\atomA, \atomB, \ldots$.

\begin{defi}[\textbf{Semantics of {\TA}}]
  \label{net_sem}
  \label{umove}
  Let $\autA[i] = (\Loc[i], \UrgLoc[i], \locInit[i], \Edg[i],\Inv[i])$ be {\TA},
  for $i \in \setenum{1,2}$. 
  We define the behaviour of the network $\net{\autA[1]}{\autA[2]}$
  as the timed LTS $(S, s_0, \LabNet, \umove{})$, where: 
  $S = \Loc[1] \times \Loc[2] \times \Val$; 
  $s_0 = ( \locInit[1], \locInit[2],\clockN[0])$; 
  $\LabNet = \Chan\ \cup \setenum{\tau} \cup \Realpos$; 
  $\umove{}$ is defined in~\Cref{fig:net_sem}.
\end{defi}

\begin{figure}
  \[
  \small
  \begin{array}{c}
    \inference[\smallnrule{[TA1]}]
    {(\locL[1], \tau, \guardG, \resetR, \locLi[1]) \in \Edg[1] & 
      \clockN \in \sem{\guardG} &  
      \reset{\clockN}{\resetR} \in \sem{\Inv[1](\locLi[1]) \land \Inv[2](\locL[2])} }  
    {(\locL[1], \locL[2],\clockN) \umove{\quad \tau \quad } (\locLi[1] , \locL[2],  \reset{\clockN}{\resetR})}
    \\[7pt]
    \inference[\smallnrule{[TA2]}]
    { (\locL[1], \atomOut{a}, \guardG[1], \resetR[1], \locLi[1]) \in \Edg[1] &
      (\locL[2], \atomIn{a},  \guardG[2], \resetR[2], \locLi[2]) \in \Edg[2]  &
      \clockN \in \sem{\guardG[1] \land \guardG[2]} & 
      \resett{\clockN}{\resetR[1]}{\resetR[2]} \in \sem{\Inv[1](\locLi[1]) \land \Inv[2](\locLi[2]) }
    }
    {(\locL[1],\locL[2],\clockN) \umove{\quad \atom{a} \quad } (\locLi[1],\locLi[2],\resett{\clockN}{\resetR[1]}{\resetR[2]})} 
    \\[10pt]
    \inference[\smallnrule{[TA3]}]
    {\forall i \in \setenum{1,2} : (\clockN + \delta) \in \sem{\Inv[i](\locL[i])}  \land \locL[i] \not\in \UrgLoc[i]} 
    {(\locL[1],\locL[2],\clockN) \umove{\quad \delta \quad } (\locL[1],\locL[2],\clockN + \delta)}  
  \end{array}
  \]
  \vspace{-10pt}
  \caption{Semantics of networks of \TA (symmetric rules omitted).}
  \label{fig:net_sem}
\end{figure}

The state of a network is given by the locations of its automata, 
and a clock valuation. 
Rule~\nrule{[TA1]} allows one of the \TA to take an internal edge, 
if its guard is satisfied and if, after resetting the clocks in $\resetR$, 
the invariants of the target locations are satisfied.
Rule~\nrule{[TA2]} allows two \TA to synchronize on a channel $\atomA$ if 
(i) their current locations are outbound edges
with complementary labels (\eg, $\atomOut{a}$ and $\atomIn{a}$);
(ii) the guards of those edges are satisfied by the clock valuation; and
(iii) the invariants of the target locations are satisfied after the clocks reset.
Rule~\nrule{[TA3]} allows time to pass, if the locations in
the current state are not urgent, and their invariants are satisfied
after the time delay.
Note that all the clocks progress with the same pace, and
taking an edge is an instantaneous action.

If the current locations of a state have no outgoing edges, 
then such state is called \emph{success},
while a state is called \emph{deadlock} if it is not success and 
no \emph{action}-transitions are possible
(neither in the current clock valuation, nor in the future).

\begin{defi}[\textbf{Deadlock freedom}] \label{def:netdeadlock}
  % Let $\netN = \net{\autA[1]}{\autA[2]}$ be a network of \TA s. 
  We say that a network state $s$ is \emph{deadlock} whenever: 
  (i)  $s$ is not \emph{success}, and 
  (ii) $\not \exists \delta \geq 0,\ \atom[\tau]{a} \in \Act \cup
  \setenum{\tau} : s \umove{\delta}\umove{ \atom[\tau]{a}}$.
  A network is \emph{deadlock-free} if none of its
  reachable states is deadlock.
\end{defi}

% \subsection{\bf \TA\  patterns.} % \label{sec:patterns}

We now introduce some operators to compose \TA.
The \emph{union} of a set of \TA collects all the locations and all the edges;
the invariant on a location is the conjunction of all the invariants 
defined on that location;
the initial location is specified as a parameter.
The \idle\ pattern creates a \TA with only a success location;
the \pfx\ pattern prefixes a location to a \TA;
finally, the \br\ pattern prefixes a location to a set of \TA,
using guarded edges.

\begin{figure}\centering
  \small
  \begin{tikzpicture}[scale=1.2]
    % \draw[help lines] (-1,-1) grid (1,1);
    \draw [] (0,0) circle (0.2);
    \draw [] (0,0) circle (0.15);
    \node [right] at (-0.4,0.25) {$\locInit$};
    % caption
    \node [right] at (-0.5,-1.5) {\scriptsize $\idle(\locInit)$};
  \end{tikzpicture}
  \hspace{20pt}
  \begin{tikzpicture}[scale=1.2]
    \draw [] (0,0) circle (0.15);
    \draw [] (0,0) circle (0.2);
    \node [right] at (-0.2, 0) {{\scriptsize U}};
    % triangle
    \draw (1,0)--(2,0.5) -- (2, -0.5) -- (1,0);
    \node [right] at (1.4, 0) {$\autA$};
    % edge
    \draw [->] (0.2,0) -- (1,0);
    \node [right] at (0.3,0.2) {$\labL, \; \resetR$};
    \node [right] at (-0.4,0.3) {$\locInit$};
    % caption
    \node [right] at (-0.5,-1.5) {\scriptsize $\pfx(\locInit, \labL, \resetR, \autA)$};
  \end{tikzpicture}
  \hspace{20pt}
  \begin{tikzpicture}[scale=1.2]
    % \small
    % \draw[help lines] (0,0) grid (5,2);
    % 
    \draw [] (0,0) circle (0.15);
    \draw [] (0,0) circle (0.2);
    \node [right] at (-0.3, -0.3) {$\guardG$};
    \node [right] at (-0.4,0.3) {$\locInit$};
    % 
    % triangle 1
    \draw (1,0.5)--(2,0.9)--(2, 0.1)--(1,0.5);
    \node [right] at (1.4, 0.5) {$\autA[1]$};
    \node [above] at (0.6,0.5) {$\guardG[1], \labL[1], \resetR[1]$};
    \draw [->] (0.2,0) -- (1,0.5);
    % dots
    \draw [dotted] (1.5,0.2) -- (1.5,-0.2);
    % triangle 2
    \draw (1,-0.5)--(2,-0.9)--(2, -0.1)--(1,-0.5);
    \node [right] at (1.4, -0.5) {$\autA[n]$};
    \node [below] at (0.6,-0.5) {$\guardG[n], \labL[n],\resetR[n]$};
    \draw [->] (0.2,0) -- (1,-0.5);
    % caption
    \node [right] at (-0.5,-1.5) {\scriptsize $\br(\locInit, \guardG, \setenum{ \labL[i], \guardG[i], \resetR[i],  \autA_i}_i)$};
  \end{tikzpicture}
  \caption{Patterns for \TA composition, represented as in~\cite{Dong08}. 
    Circles denote locations (those marked with $U$ are urgent),
    and arrows denote edges.
    Internal actions, $\guardTrue$ guards/invariants and empty resets 
    are left blank.
    A \TA is depicted as a triangle, 
    whose left vertex represents its initial location (a double circle). 
    An arrow from $\locInit$ to triangle $\autA$ represents 
    an edge from $\locInit$ to the initial location of~$\autA$.
  }
  \label{fig:pattern}
\end{figure}

\begin{defi}[\bf Union] \label{def:aut_composition}
  For all $i \in I$, let $\autA[i] =
  (\Loc[i],\UrgLoc[i],\locInit[i],\Edg[i],\Inv[i])$, 
  and let $\locL \in \bigcup_i \Loc[i]$. 
  % The union of all the $\autA[i]$ is a \TA defined as:
  We define
  $
  \sumAut{i \in I}{\; \locL}{\autA[i]} = (\bigcup_i  \Loc[i],  \bigcup_i \UrgLoc[i], \locL, \bigcup_i \Edg[i], \Inv)
  $,
  where $\Inv(\locL[j]) = \bigwedge_i \Inv[i] (\locL[j])$ for all $\locL[j] \in  \bigcup_i \Loc[i]$.
\end{defi}

\begin{defi}[\bf Idle pattern]{\label{idle_pattern}} 
  Let  $\locInit$ be a location. 
  Then, $\idle(\locInit) = (\setenum{\locInit},\emptyset,\locInit, \emptyset,\emptyset)$.
\end{defi}

\begin{defi}[\bf Prefix pattern]{\label{pfx_pattern}}
  Let $\autA = (\Loc[1],\UrgLoc[1],\locInit[1], \Edg[1], \Inv[1])$ be a \TA,
  let $\labLtau \in \LabAut$, $\locInit \not\in \Loc[1]$, 
  and $\resetR \subseteq \Clocks$.
  Then, 
  $
  \pfx(\locInit, \labLtau, \resetR, \autA) 
  = 
  (\Loc[1] \cup \setenum{\locInit},\UrgLoc[1] \cup \setenum{\locInit}, 
  \locInit, \Edg, \Inv[1]\setenum{\bind{\locInit}{\guardTrue}})
  $,
  where 
  $\Edg = \Edg[1] \cup \setenum{(\locInit, \labLtau, \guardTrue, \resetR, \locInit[1])}$.
\end{defi}

\begin{defi}[\bf Branch pattern]{\label{br_pattern}}
  For all $i \in I$, let $\autA[i] =
  (\Loc[i],\UrgLoc[i], \locInit[i],\Edg[i],\Inv[i])$, 
  and let $\guardG[i], \guardG \in \GuardG[\Clocks]$, 
  $\resetR[i] \subseteq \Clocks$,  
  $\labLtau[i] \in \LabAut$.  
  Then,
  $
  \br(\locInit, \guardG, \setcomp{(\labLtau[i],\guardG[i],\resetR[i],\autA[i])}{i \in I}) = 
  (\Loc,\UrgLoc, \locInit, \Edg, \Inv)
  $,
  where: $\Loc = \bigcup_i \Loc[i] \cup \setenum{\locInit}$,
  $\UrgLoc = \bigcup_i \UrgLoc[i]$,  $E = \bigcup_i E_i \cup
  \bigcup_i\setenum{(\locInit,\labLtau[i], \guardG[i], \resetR[i], \locInit[i])}$, 
  and
  $I = \bigcup_i I_i \cup \setenum{(\locInit, \guardG)}$.
\end{defi}

\subsection{Defining equations} \label{sec:def_eq}

The first step of our encoding from TSTs to \TA is to 
put TSTs in a normal form where 
recursive terms $\rec \tsbX \tsbP$ are replaced by 
\emph{defining equations}. % of the form $\tsbXi \eqdef \tsbP$. 
This alternative representation of infinite-state processes is quite common
in concurrency theory~\cite{Milner89}, 
hence we will defer some standard technicalities to~\Cref{sec:proofs-tst-to-ta}. 
In our normal form (called \denf) each process is represented as a pair,
composed of 
a recursion variable % (which represents the initial process),
and a set of defining equations of the form $\tsbX[i] \eqdef \tsbP[i]$, 
where $\tsbX[i]$ is a recursion variable and $\tsbP[i]$ is a term 
where every recursion variable is guarded by exactly one action (\Cref{normalform}).
 
\begin{defi}[\bf Defining-equation normal form]\label{normalform}
  % Let $\ide$ be a set of identifiers. Then, 
  A \denf is a pair $(\tsbX,D)$, 
  where $D$ is a set of
  defining equations of the form $\tsbXi \eqdef \tsbP$,
  where $\tsbP$ has the following syntax:
  \begin{align*}
    \tsbP \;\; & ::= \;\;
                 \success
                 \ \sep \
                 \TSumInt[i \in I]{\atomOut[i]{a}}{\guardG[i],\resetR[i]}{\tsbX[i]} 
                 \ \sep  \
                 \TSumExt[i \in I]{\atomIn[i]{a}}{\guardG[i],\resetR[i]}{\tsbX[i]}
  \end{align*}
  and \begin{inlinelist}
  \item the index set $I$ is finite and non-empty, and
  \item the actions in internal/external choices are pairwise distinct
  \end{inlinelist}
\end{defi}

Given a \denf $(\tsbX,D)$, we denote with 
$\urv{D}$ the set of recursion variables \emph{used} in~$D$, 
and with $\rv{D}$ the set of recursion variables \emph{defined} in $D$.
We say that $(\tsbX,D)$ is \emph{closed} when $\setenum{\tsbX} \cup \urv{D} \subseteq \rv{D}$, and 
each $\tsbX \in \urv{D}$ is defined exactly once.
Every TST $\tsbP$ can be translated into a \denf $\nf[V]{\tsbP}$ 
(where $V$ is a set of fresh recursion variables, see~\Cref{def:mapping:denf})
in a way that preserves closedness (\Cref{lem:tst-to-denf:closed})
and compliance (\Cref{lem:preserv_compliance}).
Hereafter, we will always assume closed \denf.

To define the semantics of \denf, we use a new set $\states$ of terms
(\Cref{demove}).
Intuitively, in the state $\tsbTauX$ 
the process has performed an internal move
(without choosing the branch), while in
$\todo[\atomOut{a} \setenum{\guardG, \resetR}]{\tsbX}$ 
it has committed to the branch $\atomOut{a}$. 

\begin{defi}[\bf Semantics of \denf]
  \label{def:states}
  \label{demove}
  Let  $\states$ be a set of terms of the form:
  \[
  \tsbT \;\;  ::= \;\;
  \tsbTauX  
  \sep \todo[\atomOut{a} \setenum{\guardG, \resetR}]{\tsbX}
  \sep \tsbX
  \tag{where $\atom{a} \in \Act$,  $\guardG \in \GuardG[\Clocks]$, and $\resetR \subseteq \Clocks$}
  \]
  and let $D$ be a set of defining equations. 
  Then, the relation $\demove{}$ is inductively defined by the set of rules 
  in~\Cref{fig:tst:s_semantics_bis}, where we use $\tsbT, \tsbS\ldots$ to range
  over $\states$.
\end{defi}

\begin{figure}[t]
  \small
  \[
  \begin{array}{c}
    \inference[\smallnrule{\erule}]
    {(\tsbX \eqdef \tsbP)  \in D &  p = \bigoplus \ldots &  \clockN \in \rdy{p}}
    {(\tsbX,\; \clockN) \;\demove{\tau}\; (\tsbTauX,\; \clockN)}
    \qquad
    \inference[\smallnrule{\crule}]
    {(\tsbX \eqdef {\TsumI{\atomOut{a}}{\guardG,\resetR}{\tsbY}} \oplus  \tsbP ) \in D & \clockN \in \sem{\guardG}}
    {(\tsbTauX,\; \clockN) \;\demove{\tau}\; ({\todo[\TsumI{\atomOut{a}}{\guardG,\resetR}{}]{\tsbY}},\; \clockN) }
    \\[10pt] 
    \inference[\smallnrule{\outrule}]
    {}
    {({\todo[\TsumI{\atomOut{a}}{\guardG,\resetR}{}]{\tsbY}}, \; \clockN)
    \;\demove{\atomOut{a}}\;
    (\tsbY,\; \reset{\clockN}{\resetR})
    }
    \qquad
    \inference[\smallnrule{\inrule}]
    {(\tsbX \eqdef \TsumE{\atomIn{a}}{\guardG,\resetR}{\tsbY} + \tsbP) \in D & \clockN \in \sem{\guardG}}
    {(\tsbX, \; \clockN)
    \;\demove{\atomIn{a}}\;
    (\tsbY, \; \reset{\clockN}{\resetR}) }
    \\[10pt] 
    \inference[\smallnrule{\edrule}]
    {(\tsbX \eqdef  \sum \ldots ) \in D  \; \lor \;   (\tsbX \eqdef   \success ) \in D}
    {(\tsbX, \; \clockN)\demove {\; \delta \;} (\tsbX, \clockN +\delta )} 
    \qquad
    \inference[\smallnrule{\idrule} ]
    {\tsbX \eqdef  \tsbP \in D & \clockN + \delta \in \rdy{\tsbP}}
    {(\tsbTauX,\; \clockN)\demove{\; \delta \; }(\tsbTauX,\; \clockN+\delta)}   
    \\[10pt]
    \inference[\smallnrule{\serule} ]
    {(\tsbS,\clockN) \demove{\; \tau \; } (\tsbSi,\clockN)}
    {(\tsbS,\clockN) \mid (\tsbT,\clockE) \demove {\; \tau \;}(\tsbSi,\clockN) \mid (\tsbT,\clockE)} 
    \qquad
    \inference[\smallnrule{\sdrule}]
    {(\tsbS,\clockN) \demove{\; \delta \; } (\tsbS,\clockNi) & 
    (\tsbT,\clockE) \demove{\; \delta \; } (\tsbT,\clockEi)}
    {(\tsbS,\clockN) \mid (\tsbT,\clockE) \demove{\; \delta \;} 
    (\tsbS,\clockNi) \mid (\tsbT,\clockEi)} 
    \\[10pt]  
    \inference[\smallnrule{\strule}]
    {(\tsbS,\clockN) \demove{\; \atomOut{a} \; } (\tsbSi,\clockNi) &
    (\tsbT,\clockE) \demove{\; \atomIn{a} \; } (\tsbTi,\clockEi)}
    {(\tsbS,\clockN) \mid (\tsbT,\clockE) \demove{\; \atom{a} \;} 
    (\tsbS',\clockNi) \mid (\tsbTi,\clockEi)}   
  \end{array}
  \]
  \caption{Semantics of \denf (symmetric rules omitted).}
  \label{fig:tst:s_semantics_bis}
\end{figure}

There is almost a one-to-one correspondence between the rules of 
$\smove{}$ and $\demove{}$,
aside from the syntactic differences between TST and \denf.
Rules~\nrule{\edrule} and~\nrule{\idrule} allow time to pass in the
same way as~\nrule{[Del]} does: they have been split in two 
only to accommodate with the new term $\tsbTauX$.
Rule~\nrule{\erule} forces every internal choice to perform an internal
step before committing to a branch. 
However, this is done only if $\clockN \in \rdy{p}$, 
\ie if at least one of the branches is available. 
Note that before the internal step has been performed, time cannot pass, 
hence the only possible move is via~\nrule{\crule}
(so, a sequence of rules~\nrule{\erule} and~\nrule{\crule}
corresponds to rule \nrule{[$\sumInt$]} in $\smove{}$).

\begin{defi}[\bf Compliance for \denf] 
  \label{def:dedeadlock}
  \label{def:decompliance}
  A state $(\tsbX,\clockN)\mid (\tsbY,\clockE)$ in $\demove{}$ 
  is \emph{deadlock} whenever
  $(i)$ it is not the case that both $\tsbX \eqdef \success$ and $\tsbY \eqdef \success$ are in $D$,
  and $(ii)$ there is no $\delta$ and no $\atom[\tau]{a} \in \Act \cup \setenum{\tau}$
  such that $(\tsbX,\clockN +\delta)\mid (\tsbY,\clockE + \delta) \demove{\; \atom[\tau]{a} \,}$.
  We write $(x,\clockN) \decompliant (y,\clockE)$ when:
  \[
  (x,\clockN)\mid(y,\clockE)
  \demove{}^* 
  (x',\clockNi) \mid (y',\clockEi)
  \quad \text{ implies } \quad 
  (x',\clockNi) \mid (y',\clockEi)
  \text{ not deadlock}
  \]
  and we write $(x,D') \compliant (y,D'')$
  whenever $(x,\clockN[0]) \decompliant[D' \cup D''] (y,\clockE[0])$.
\end{defi}

\begin{reslemma}{lem:preserv_compliance}
  Let $\tsbP$ and $\tsbQ$ be two closed TSTs with no shared clocks.  
  Let $V_1$ and $V_2$ be two sets of recursion variables not
  occurring in $\tsbP$ and $\tsbQ$ and such that
  $V_1 \cap V_2 = \emptyset$. 
  Then:
  \[
  \tsbP  \compliant \tsbQ \iff  \nf[V_1]{\tsbP}   \compliant \nf[V_2]{\tsbQ}  
  \]
\end{reslemma}

\subsection{Encoding  \denf into \TA }\label{sec:mapping}

In~\Cref{mapping} we transform \denf into \TA:
first, we use the function $\sem{\cdot}$ to transform each 
defining equation into a \TA; 
then, we compose all the resulting \TA with the union operator $\sumAut{}{}{}$
(introduced in~\Cref{def:aut_composition}).
Our encoding produces a location for every term in $\states$ 
reachable in $\demove{}$; it creates an edge for each move that
can be performed by a TST, so that, in the end, the moves
performed by the network associated to $(\tsbX,D)$ 
coincide with the moves of $\tsbX$ in $\demove{}$. 
To avoid some technicalities in proofs,
we assume that disjunctions never occur in guards 
(while they can occur in invariants);
in~\Cref{sec:upp_guards} we discuss how to deal with the general case.

\begin{defi}[\textbf{Encoding of \denf into \TA}] \label{mapping}
  For all closed \denf $(\tsbX,D)$,
  we define the function
  $
  \trn{(\tsbX,D)} = \sumAut{d \in D}{\tsbX}{\sem{d}}
  $, 
  where:
  \[
  \sem{\tsbX \triangleq \tsbP} 
  =
  \begin{cases}
    \idle(\tsbX)
    & \text{if $\tsbP = \success$}
    \\[4pt]
    \br(\tsbX, \rdy{\tsbP}, \setenum{ (\atomIn[i]{a}, \guardG[i], \resetR[i], \idle(\tsbX[i]))}_i)
    & \text{if $\displaystyle\tsbP = \sum_{i \in I} \atomIn{a}\setenum{\guardG[i], \resetR[i]}.\tsbX[i]$}
    \\
    \!\!\!\begin{array}{l}
      \pfx(\tsbX, \tau, \emptyset, \br(\tsbTauX, \rdy{\tsbP}, \setenum{(\tau, \guardG[i], \emptyset,  A_i)}_i), \text{where}
      \\
      A_i = \pfx( \todo[{\atomOut[i]{a}} \setenum{\guardG[i], \resetR[i]}]{\tsbX[i]}, \atomOut[i]{a}, \resetR[i], \idle(\tsbX[i]))
    \end{array}
    & \text{if $\displaystyle\tsbP = \bigoplus_{i \in I} \atomOut[i]{a}\setenum{\guardG[i], \resetR[i]}.\tsbX[i]$}
  \end{cases}
  \]
\end{defi}

\begin{figure}[t]
  \centering
  \begin{tikzpicture}[scale=1.2]
    \draw [] (0,0) circle (0.15);
    \draw [] (0,0) circle (0.2);
    \node [right] at (-0.2, 0.3) {{\tiny $\tsbX$}};
    \node [right] at (-0.2, 0) {{\footnotesize U}};
    \draw [->] (0.2,0) -- (0.8,0);
    \node [above] at (0.4,0) {};%$\tau$% 
    \draw [] (1,0) circle (0.2);
    \node [right] at (0.5, -0.3) {\tiny $\rdy{\tsbP}$};
    \node [right] at (0.7, 0.3) {{\tiny $\tsbTauX$}};
    % edge to f.u.n.
    \draw [dotted, -] (2.5,0.1) -- (2.5, 0.65); 
    \draw [rounded corners=1mm, ->] (1.2,0) -- (1.7, 0.6) --  (2.3, 0.6)-- (2.4, 0.72); 
    \node [right] at (1.7, 0.8) {$\guardG[1]$}; %guardia
    \draw [] (2.5,0.9) circle (0.2);
    \node [right] at (2.3 , 0.9) {{\footnotesize U}};
    \node [right] at (1.6 , 1.25) {{\tiny  $\todo[{\atomOut[1]{a}} \setenum{\guardG[1], \resetR[1]}]{\tsbX[1]}$}};
    \draw [rounded corners=1mm, ->] (2.6, 0.73) --  (2.7, 0.6) --  (3.9,0.6) --  (4,0.73) ;
    \node [right] at (2.8, 0.8) {$ \atomOut[1]{a} \, \resetR[1]$};
    \draw [] (4.1,0.9) circle (0.2);
    \node [right] at (3.9, 1.25) {{\tiny  ${\tsbX[1]}$ }};
    \draw [rounded corners=1mm, ->] (1.2,0) --  (1.7, -0.6)--  (2.3, -0.6) -- (2.4, -0.47); 
    \draw [] (2.5,-0.3) circle (0.2);
    \node [right] at (2.3 , -0.3) {{\footnotesize U}};
    \node [right] at (1.6 , 0.05) {{\tiny  $\todo[{\atomOut[n]{a}} \setenum{\guardG[n], \resetR[n]}]{\tsbX[n]}$}};
    \draw [rounded corners=1mm, ->] (2.6, - 0.47) --  (2.7, -0.6)  -- (3.9,-0.6) -- (3.98, -0.45);
    \node [right] at (2.8, -0.4) {$\atomOut[n]{a} \,  \resetR[n]$};
    \node [right] at (1.7, -0.4) {$\guardG[n]$}; %guardia
    \draw [] (4.1 ,-0.3) circle (0.2);
    \node [right] at (3.9 , 0.05) {\tiny $\tsbX[n]$};
    % caption
  \end{tikzpicture}
  \hspace{20pt}
  \begin{tikzpicture}[scale=1.2]
    % \small
    % 
    \draw [] (1.5,0) circle (0.15);
    \draw [] (1.5,0) circle (0.2);
    \node [right] at (1.3, 0.3) {\tiny $\tsbX$};
    \draw [rounded corners=1mm, ->] (1.7,0) -- (2.2,0.6)-- (3.9,0.6) --  (4,0.73) ;
    \node [right] at (2.2, 0.8) {$\guardG[1] \;\; \atomIn[1]{a} \;\; \resetR[1]$};
    \draw [dotted, -] (3,-0.25) -- (3, 0.56); 
    \draw [] (4.1,0.9) circle (0.2);
    \node [right] at (3.9, 1.25) {{\tiny  ${\tsbX[1]}$ }};
    \draw [rounded corners=1mm,->]  (1.7,0) -- (2.2,-0.6) --  (3.9,-0.6) -- (3.98, -0.45);
    \node [right] at (2.2,-0.4) {$\guardG[n] \;\; \atomIn[n]{a} \;\; \resetR[n]$};
    % success automaton
    \draw [] (4.1,-0.3) circle (0.2);
    \node [right] at (3.9, 0.05) {\tiny $\tsbX[n]$};
  \end{tikzpicture}
  \caption{Encoding of an internal choice (left) and of an external choice (right).}
  \label{fig:mapping_choices}
\end{figure}\medskip

\noindent The encoding of $\tsbX \triangleq \success$ 
produces an \emph{idle} \TA (\Cref{idle_pattern}), 
with a single success location $\tsbX$.
The encoding of an \emph{external choice}
$\tsbX \triangleq \sum_{i \in I} \atomIn[i]{a}\setenum{\guardG[i],\resetR[i]}.\tsbX[i]$ 
generates a location $\tsbX$ (with $\guardTrue$ invariant), and % $\rdy{\tsbP} = \true$ and
outgoing edges towards all the locations $\tsbX[i]$: 
these edges have guards $\guardG[i]$, reset sets $\resetR[i]$, 
and synchronization labels $\atomIn[i]{a}$
(see~\Cref{fig:mapping_choices}, right).  
Basically, the \TA is listening on all its channels $\atomIn[i]{a}$, 
and since the location $\tsbX$ is not urgent, 
time can pass forever while in there.
The encoding of $\tsbX \triangleq \tsbP$ when
$\tsbP$ is an \emph{internal choice} 
$\bigoplus_{i \in I} \atomOut[i]{a}\setenum{\guardG[i], \resetR[i]}.\tsbX[i]$
is a bit more complex (see \Cref{fig:mapping_choices}, left). 
First, we generate an urgent location $\tsbX$,
and an edge leading to a non-urgent location $\tsbTauX$
with invariant $\rdy{\tsbP}$.
Note that, although $\rdy{\tsbP}$ is a semantic object (a set of clock valuations),
it can always be represented syntactically as a guard~\cite{HenzingerNSY94}. 
Second, we connect the latter location to $i$ urgent locations
(named $\todo[{\atomOut[i]{a}} \setenum{\guardG[i], \resetR[i]}]{\tsbX[i]}$)
via internal edges with guards~$\guardG[i]$;
each of these locations is connected to $\tsbX[i]$
through an edge with reset set $\resetR[i]$ and label $\atomOut[i]{a}$.
The resulting \TA can wait some time before deciding on which branch committing,
since time can pass in location $\tsbTauX$;
however, time passing cannot make all the guards on the branches unsatisfiable,
because the invariant $\rdy{\tsbP}$ on $\tsbTauX$ 
ensures that the location is left on time.
As soon as this happens, since the arriving location is urgent, time cannot pass
anymore, and a synchronization is performed (if possible).
The reason we use both locations $\tsbX$ and $\tsbTauX$ is that, in some
executions, all the guards of an internal choice may have \emph{already} expired. %
In this case, the invariant $\rdy{\tsbP}$ on location $\tsbTauX$ would be
false, and the system could not enter it, so preventing a previous
action (if any) to be done.  To avoid this problem, we have put an
extra location ($\tsbX$) before $\tsbTauX$.

Since location names are in $\states$, every defined/used variable
$\tsbX$ has a location called $\tsbX$, in \emph{every} \TA which
defines/uses it. 
When the \TA obtained from different equations are composed with $\bigsqcup$, 
locations with the same name collapse, therefore
connecting together the call of a recursion variable with its definition.

\begin{exa} 
  \label{ex:mapping}
  Let
  $
  \tsbP[1] = \rec
  {\tsbX[1]}
  {\left( 
      \atomOut{a}\setenum{\clockFmt{c} = 2,\setenum{\clockFmt{c}}}.\tsbX[1] \oplus \atomOut{b}\setenum{\clockT < 7}
    \right)}
  $,
  and let $V_1$  be a set of recursion variables not in $\tsbP[1]$.
  The \denf of $\tsbP[1]$ is $\nf[V_1]{\tsbP[1]} = (\tsbX, D)$, 
  where:
  \[
  D \; = \;
  \setenum{
    \;\;
    \tsbX \eqdef \tsbP, 
    \;\;
    \tsbW \eqdef \success
    \;\;
    \tag*{where $\tsbP = \atomOut{a}\setenum{\clockFmt{c} = 2,\setenum{\clockFmt{c}}}.\tsbX \oplus \atomOut{b}\setenum{\clockT < 7}.\tsbW$}
  }
  \]
  \Cref{fig:ex:mapping} (left) shows the \TA $\trn{(\tsbX,D)}$.
  All the locations but $\tsbTauX$ and $\tsbW$ are urgent;
  all the invariants are $\guardTrue$, except for 
  $
  \Inv( \tsbTauX) = \rdy{\tsbP} 
  = \past{(\sem{\clockFmt{c} = 2} \cup  \sem{\clockT < 7})}
  $
  (represented by the guard $\clockFmt{c} \leq 2 \lor \clockT < 7$).

  \smallskip
  Now, let
  $
  \tsbQ[1] = \rec
  {\tsbY[1]}
  {\atomIn{a}\setenum{\clockFmt{r} > 1 \land \clockFmt{r} < 5,
      \setenum{\clockFmt{r}}}.\tsbY[1] + \atomIn{b}\setenum{\clockFmt{r} < 7}}
  $,
  and let $V_2$ be a set of recursion variables not in $\tsbQ[1]$.
  The \denf of $\tsbQ[1]$ is $\nf[V_2] {\tsbQ[1]} = (\tsbY, F)$, where: 
  \[
  F \; = \;
  \setenum{
    \;\;
    \tsbY \eqdef \tsbQ,
    \;\;
    \tsbZ \eqdef \success
    \;\;
  }
  \tag*{where $\tsbQ = \atomIn{a}\setenum{r > 1 \land r < 5 ,\setenum{r}}.\tsbY + \atomIn{b}\setenum{r < 7}.\tsbZ$}
  \]
  \Cref{fig:ex:mapping} (right) shows the \TA $\trn{(\tsbY,F)}$. %
  Its locations are $\tsbY$ and $\tsbZ$ (both non-urgent), with invariants
  $\Inv(\tsbZ)  = \guardTrue$ and 
  $
  \Inv(\tsbY) = \rdy{\tsbQ} = \Val
  $
  (represented by the guard $\guardTrue$). 
\end{exa}

\begin{figure}[t]
  \centering
  \begin{tikzpicture}[scale=1.2]
    % \small
    % \draw[help lines] (0,0) grid (5,2); 
    % first node
    \draw [] (0,0) circle (0.15);
    \draw [] (0,0) circle (0.2);
    \node [right] at (-0.35, 0.3) {{\tiny $\tsbX$}};
    \node [right] at (-0.2, 0) {{\footnotesize U}};
    % edge from first to second
    \draw [->] (0.2,0) -- (0.8,0);
    \node [above] at (0.4,0) {};%$\tau$% 
    % second node
    \draw [] (1,0) circle (0.2);
    \node [right] at (0.5, -0.3) {\tiny $\clockT<7 \; \lor$};
    \node [right] at (0.5, -0.5) {\tiny $\clockFmt{c} \leq 2$};
    \node [right] at (0.7, 0.3) {{\tiny $\tsbTauX$}};
    % edge to f.u.n.
    \draw [rounded corners=1mm, ->] (1.2,0) -- (1.7, 0.6) --  (2.3, 0.6)-- (2.4, 0.72); ; 
    \node [right] at (1.3, 0.8) {$\clockFmt{c}=2$}; %guardia
    \draw [] (2.5,0.9) circle (0.2);
    \node [right] at (2.3 , 0.9) {{\footnotesize U}};
    \node [right] at (1.3 , 1.25) {{\tiny  $ \todo[\atomOut{a}\setenum{\clockFmt{c}=2,\setenum{\clockFmt{c}}}]{\tsbX}$}};
    \draw [rounded corners=1mm, ->] (2.6, 0.73) --  (2.7, 0.6) --  (4.1,0.6)  -- (4.1, 1.5) -- (0,1.5) -- (0,0.2);
    \node [right] at (2.9, 0.8) {$ \atomOut{a}\; \setenum{\clockT}$};
    % edge to f.s.n.
    \draw [rounded corners=1mm, ->] (1.2,0) --  (1.7, -0.6)--  (2.3, -0.6) -- (2.4, -0.47); 
    \node [right] at (1.5, -0.8) {$\clockT<7$}; % \, \tau$};
    \draw [] (2.5,-0.3) circle (0.2);
    \node [right] at (2.3 , -0.3) {{\footnotesize U}};
    \node [right] at (1.7 , 0) {{\tiny  $\todo[\atomOut{b}\setenum{\clockT<7,\emptyset}]{\tsbW}$}};
    \draw [rounded corners=1mm, ->] (2.6, - 0.47) --  (2.7, -0.6)  --  (3.72,-0.6);
    \node [right] at (3, -0.4) {$\atomOut{b}$};
    % success automaton
    \draw [] (3.8 + 0.1 ,-0.5) circle (0.2);
    % \node [right] at (3.4 + 0.1 , -0.5) {{\footnotesize U}};
    \node [right] at (3.6 + 0.1 , -0.2) {\tiny $\tsbW$};
    % caption
  \end{tikzpicture}
  \hspace{20pt}
  \begin{tikzpicture}[scale=1.2]
    % \small
    % first node
    \draw [] (1.5,0) circle (0.15);
    \draw [] (1.5,0) circle (0.2);
    \node [right] at (1.1, 0.3) {\tiny $\tsbY$};
    % loop edge to init
    \draw [rounded corners=1mm, ->] (1.7,0) -- (2.2,0.6)--(5.9,0.6) -- (5.9, 1.6) -- (1.5,1.6) -- (1.5,0.2);
    \node [right] at (2.2, 0.8) {$(\clockFmt{r} > 1 \land \clockFmt{r}<5) \;\; \atomIn{a} \;\; \setenum{\clockFmt{r}} $};
    % edge to  success state
    \draw [rounded corners=1mm,->]  (1.7,0) -- (2.2,-0.6) -- (4,-0.6);
    \node [right] at (2.2,-0.4) {$\clockFmt{r}<7 \;\; \atomIn{b}$};
    % success automaton
    \draw [] (4.2,-0.5) circle (0.2);
    % \node [right] at (4, -0.5) {{\footnotesize U}};
    \node [right] at (4, -0.2) {\tiny $\tsbZ$};
    % caption
    \node [right] at (1, -0.8) {$ $};
  \end{tikzpicture}
  \caption{Encoding of the TSTs in~\Cref{ex:mapping}.}
  \label{fig:ex:mapping}
\end{figure}

\Cref{lem:bsim} shows a strict correspondence between the timed LTSs $\demove[D]{}$ and $\umove{}$:
matching states are strongly bisimilar.
To make explicit that some state $q$ belongs to a LTS $\rightarrow$, 
we write it as a pair $(q,\rightarrow)$.

\begin{reslemma}{lem:bsim}
  Let $(\tsbX,D')$ and $(\tsbY,D'')$ be \denf
  such that $\rv{D'} \cap \rv{D''} = \emptyset$.
  Let $\netN = \net{\trn{(\tsbX,D')}}{\trn{(\tsbY,D'')}}$. %
  Then:
  \[  
  ( (\tsbX, \clockN[0]) \mid (\tsbY, \clockE[0]), \demove[D'\cup D'']{}) 
  \sim 
  ((\tsbX, \tsbY,\clockN[0] \sqcup \clockE[0]), \umove{} )
  \]
\end{reslemma}

\subsection{Decidability of compliance}\label{sec:model-checking}

To prove that compliance between TSTs is decidable, 
we reduce such problem to that of checking if a network of \TA is deadlock-free ---
which is known to be decidable~\cite{Alur94theory}.

\begin{thm} \label{th:compliance}
  Let $\tsbP$ and $\tsbQ$ be two closed TSTs. 
  Let $V_1$ and $V_2$ be two sets of recursion variables not
  occurring in $\tsbP$ and $\tsbQ$ and such that
  $V_1 \cap V_2 = \emptyset$.  Then:
  \[
  \tsbP \compliant \tsbQ 
  \iff  
  \net{\trn{\nf[V_1]{\tsbP}}}{\trn{\nf[V_2]{\tsbQ}}} 
  \text{ is deadlock-free}
  \]
\end{thm}
\begin{proof}
  Let $(\tsbX,D') = \nf[V_1]{\tsbP}$ and $(\tsbY,D'') = \nf[V_2]{\tsbQ}$.
  We show that:
  \begin{equation} \label{eq:compliance:denf}
    (\tsbX,D') \compliant (\tsbY,D'') 
    \;\; \iff \;\;
    \net{\trn{(\tsbX,D')}}{\trn{(\tsbY,D'')}} 
    % \net{\sem{\tsbX}}{\sem{\tsbY}} 
    \text{ is deadlock-free}
  \end{equation}
  For the $(\Rightarrow)$ direction, 
  assume by contradiction that $(\tsbX,D') \compliant (\tsbY,D'')$
  but the associated network $\netN$ is \emph{not} deadlock-free. 
  By~\Cref{def:netdeadlock}, there exist a reachable deadlocked state $s$, \ie
  $(\tsbX, \tsbY, \clockN[0] \sqcup \clockE[0]) \umove{}^* s = (x,y,\clockN \sqcup \clockE)$
  and there are no $ \delta \geq 0$ and
  $\atom[\tau]{a} \in \Act \cup \setenum{\tau}$ such that
  $s \umove{\delta} \umove{ \atom[\tau]{a}}$.
  By~\Cref{lem:bsim},
  $(\tsbX, \clockN[0]) \mid (\tsbY, \clockE[0]) \demove[D' \cup D'']{}^* (x,\clockN) \mid (y, \clockE)$,
  and the last state is bisimilar to $s$. %
  However, since $(\tsbX,D') \compliant (\tsbY,D'')$,
  the state $(x,\clockN) \mid (y, \clockE)$ is not deadlock according to~\Cref{def:dedeadlock}. % 
  Here we have two cases.
  \begin{inlinelist}
  \item If $x = \success$ and $y = \success$, then by~\Cref{mapping} $x$ and $y$ are success locations --- 
    contradiction, because $s$ is not success;
  \item $(x,\clockN+\delta) \mid (y, \clockE+\delta) \demove[D' \cup D'']{\; \atom[\tau]{a} \,}$ 
    for some $\delta$ and no $\atom[\tau]{a} \in \Act \cup \setenum{\tau}$.
    Then, by~\Cref{lem:bsim} also $s$ can fire that moves --- contradiction.
  \end{inlinelist}
  The direction $(\xLeftarrow{})$ is similar. % to the previous case.
  The main statement follows from~\eqref{eq:compliance:denf} and from~\Cref{lem:preserv_compliance}.
\end{proof}

\begin{exa} \label{ex:mapping:compliance}
  Let us consider the two TSTs in~\Cref{ex:mapping}.
  Since the associated network of \TA (\Cref{fig:ex:mapping}) is deadlock-free,
  by~\Cref{th:compliance} we conclude that $\tsbP[1] \compliant \tsbQ[1]$.
\end{exa}

Our implementation of TSTs (\url{co2.unica.it}) uses
\UPP~\cite{Uppaal04tutorial} to check compliance. 
\UPP can verify deadlock-freedom of a network of \TA  
through its query language, which is a simplified version of Time Computation Tree Logic. %
In \UPP, the special state formula \code{deadlock} is satisfied by all deadlocked states;
hence, a network is deadlock-free if none of its reachable states
satisfies \code{deadlock}.
Note that checking deadlock-freedom with the path formula \code{A[] not deadlock}
would not be correct, 
because~\Cref{def:netdeadlock} does not consider as deadlock the success states.
The actual \UPP query we use takes into account success states.
E.g., let $\netN = \net{\autA[1]}{\autA[2]}$ be a network of \TA,
with $\locL[1]$ success location of $\autA[1]$, 
and $\locL[2],\locL[3]$ success locations of $\autA[2]$.
The query 
\code{A[]  deadlock imply  (A1.l1  \&\& (A2.l2 || A2.l3))}
correctly checks deadlock freedom according to~\Cref{def:netdeadlock}.

\section{Runtime monitoring of TSTs} \label{sec:tst-monitoring}

In this~\namecref{sec:tst-monitoring} we study runtime monitoring
based on TSTs. % 
The setting is the following: two participants $\pmvA$
and $\pmvB$ want to interact according to two (compliant) TSTs
$\tsbP[\pmvA]$ and $\tsbP[\pmvB]$, respectively.  This interaction
happens through a server, which monitors all the messages exchanged
between $\pmvA$ and $\pmvB$, while keeping track of the passing of
time. %
If a participant (say, $\pmvA$) sends a message not
expected by her TST, then the monitor classifies $\pmvA$ as
\emph{culpable} of a violation. %
There are other two circumstances where $\pmvA$ is culpable: 
(i) $\tsbP[\pmvA]$ is an internal choice, but $\pmvA$ loses time
until all the branches become unfeasible, or 
(ii) $\tsbP[\pmvA]$ is an external choice, 
but $\pmvA$ does not readily receive an incoming message sent by $\pmvB$.

Note that the semantics in~\Cref{fig:tst:s_semantics}
cannot be directly exploited to define such a runtime monitor,
for two reasons. %
First, the synchronisation rule is purely symmetric, 
while the monitor outlined above assumes an asymmetry between
internal and external choices.
Second, in the semantics in~\Cref{fig:tst:s_semantics}
all the transitions (both messages and delays) 
must be allowed by the TSTs:
for instance, $(\TsumI{\atomOut{a}}{\guardFmt{t \leq 1}}{}, \clockN)$
cannot take any transitions (neither $\atomOut{a}$ nor $\delta$)
if $\clockN(t) > 1$. %
In a runtime monitor we want to avoid such kind of situations,
where no actions are possible, and the time is frozen. %
More specifically, our desideratum is that the runtime monitor 
acts as a \emph{deterministic} automaton, which reads a 
\emph{timed trace} (a sequence of actions and time delays)
and it reaches a unique state $\gamma$,
which can be inspected to find which of the two participants (if any)
is culpable. % of a violation. % 

To reach this goal, we define the semantics of the runtime monitor
on two levels.
The first level, specified by the relation $\cmove{}$,
deals with the case of honest participants;
however, differently from the semantics in~\Cref{sec:tst-semantics},
here we decouple the action of sending from that of receiving.
More precisely, if $\pmvA$ has an internal choice
and $\pmvB$ has an external choice, then we postulate that
$\pmvA$ must move first, by doing one of the outputs in her choice,
and then $\pmvB$ must be ready to do the corresponding input.
The second level, called \emph{monitoring semantics} 
and specified by the relation $\mcmove{}$, 
builds upon the first one.
Each move accepted by the first level is also accepted
by the monitor.
Additionally, the monitoring semantics defines transitions 
for actions not accepted by the first level, 
for instance unexpected input/output actions, 
and improper time delays.
In these cases, the monitoring semantics signals which of the two 
participants is culpable.

\begin{defi}[\textbf{Monitoring semantics of TSTs}]
  \label{def:tst:turn-based}
  \label{def:tst:monitoring-semantics}
  \emph{Monitoring configurations} $\gamma,\gamma',\ldots$
  are terms of the form $\tb{P} \mmid \tb{Q}$, where
  $\tb{P}$ and $\tb{Q}$ are triples $(\tsbP, c, \clockN)$, 
  such that
  \begin{inlinelist}
   \item $\tsbP$ is either a TST or $\cnil$, and
   \item $\qmvA$ is a one-position buffer (either empty or containing an output label).
   \end{inlinelist}
  The transition relations $\cmove{}$ and $\mcmove{}$
  over monitoring configurations,
  with labels 
  $\lambda, \lambda', \ldots \in (\setenum{\pmvA, \pmvB} \times \BLab) \cup \Realpos$,
  are defined in~\Cref{fig:tst:monitoring-semantics}.
\end{defi}

\begin{figure}[t]
  \[
  \begin{array}{cll}
    {(\TsumI{\atomOut{a}}{\guardG,\resetR} \tsbP \sumInt \tsbPi, \queue{} , \clockN)  \mmid (\tsbQ , \queue{} , \clockE)
      \cmove{\pmv A \says \atomOut{a}}
      (\tsbP, \queue{\atomOut{a}}, \reset{\clockN}{\resetR}) \mmid 
      (\tsbQ , \queue{} , \clockE)
    }
    \quad
    & \text{if } \clockN \in \sem{\guardG}
    \;\;
    & \nrule{[M-$\sumInt$]}
    \\[10pt]
    { (\tsbP, \queue{\atomOut{a}}, \clockN) \mmid 
      (\TsumE{\atomIn{a}}{\guardG,\resetR}{\tsbQ} \sumExt \tsbQi}, \queue{}, \clockE)
    \cmove{\pmv B \says \atomIn{a}}
     (\tsbP, \queue{}, \clockN)  \mmid  (\tsbQ, \queue{},  \reset{\clockE}{\resetR})
    & \text{if } \clockN \in \sem{\guardG}
    & \nrule{[M-$\sumExt$]}
    \\[10pt]
    \irule{
      \clockN + \delta \in \rdy{\tsbP} \qquad 
      \clockE + \delta \in \rdy{\tsbQ}}
    {(\tsbP, \queue{}, \clockN) \mmid ( \tsbQ , \queue{}, \clockE) 
      \cmove{\delta} 
      (\tsbP, \queue{}, \clockN + \delta) \mmid (\tsbQ , \queue{}, \clockE + \delta)}
    & 
    & \nrule{[M-Del]}
    \\[15pt]
    \irule 
    {(\tsbP, \qmvA, \clockN) \mmid (\tsbQ,  \qmvB, \clockE) 
    \cmove{\lambda} 
    (\tsbPi, \qmvAi, \clockNi) \mmid (\tsbQi, \qmvBi, \clockEi)}
    {(\tsbP,\qmvA, \clockN) \mmid (\tsbQ, \qmvB, \clockE)  
    \mcmove{\lambda} 
    (\tsbPi, \qmvAi, \clockNi) \mmid (\tsbQi, \qmvBi, \clockEi)}
    & & \nrule{[M-Ok]}
    \\[15pt]
    \irule{(\tsbP, \qmvA, \clockN) \mmid (\tsbQ,  \qmvB, \clockE) 
    \not\cmove{\pmvA \says \labL} }
    {(\tsbP, \qmvA, \clockN) \mmid (\tsbQ,  \qmvB, \clockE) 
    \mcmove{\pmv A \says \labL}
    (\cnil, \qmvA, \clockN) \mmid (\tsbQ, \qmvB, \clockE)} 
    & & \nrule{[M-FailA]}
    \\[20pt]
    \irule{(\qmvB  =  \queue{} \land  \clockN + \delta \not\in \rdy{\tsbP}) \;\lor\; \qmvB  \neq \queue{}}
    {(\tsbP, \qmvA , \clockN) \mmid   (\tsbQ, \qmvB, \clockE) 
    \mcmove{\delta}  
    (\cnil, \qmvA, \clockN + \delta) \mmid   (\tsbQ, \qmvB, \clockE + \delta)}
    & & \nrule{[M-FailD]}
  \end{array}
  \]
  \caption{Monitoring semantics (symmetric rules omitted).}
  \label{fig:tst:monitoring-semantics}
  \label{fig:tst:turn-based}
\end{figure}

In the rules in~\Cref{fig:tst:turn-based},
we always assume that the leftmost TST is
governed by {\pmv A}, while the rightmost one
is governed by {\pmv B}.
In rule~\nrule{[M-$\sumInt$]}, $\pmvA$ has an internal choice, 
and she can fire one of her outputs $\atomOut{a}$,
provided that its buffer is empty, and the guard $\guardG$ is satisfied.
When this happens, 
the message $\atomOut{a}$ is written to the buffer,
and the clocks in $\resetR$ are reset.
Then, $\pmvB$ can read the buffer,
by firing $\atomIn{a}$ in an external choice
through rule~\nrule{[M-$\sumExt$]};
this requires that the buffer of $\pmvB$ is empty, 
and the guard $\guardG$ of the branch $\atomIn{a}$ is satisfied.
Rule~\nrule{[M-Del]} allows time to pass,
provided that the delay $\delta$ is permitted for both participants,
and both buffers are empty.
The last three rules specify the runtime monitor.  Rule~\nrule{[M-Ok]}
says that any move accepted by $\cmove{}$ is also accepted by the
monitor.  Rule~\nrule{[M-FailA]} is used when participant $\pmvA$ 
attempts to do an action not permitted by $\cmove{}$: 
this makes the monitor evolve to a configuration where $\pmvA$ is
culpable (denoted by the term $\cnil$). %
Rule~\nrule{[M-FailD]} 
makes $\pmvA$ culpable when time passes, in two cases:
either $\pmvA$ has an internal choice, 
but the guards are no longer satisfiable;
or she has an external choice, and there is an incoming message.

The following~\namecref{lem:tst-monitoring:determinism}
establishes that the monitoring semantics is deterministic,
\ie, if $\gamma \mcmove{\lambda} \gamma'$ 
and $\gamma \mcmove{\lambda} \gamma''$,
then $\gamma' = \gamma''$.
This is a very desirable property indeed,
because it ensures that the culpability of a participant at any
given time is uniquely determined by the past actions.
Further, for all finite timed traces $\vec{\lambda}$
(\ie, sequences of actions $\pmvA \says \labL$ or time delays $\delta$),
there exists some configuration $\gamma$ reachable
from the initial one.

\newcommand{\lemtstmonitoringdeterminism}{%
  Let $\gamma_0 = (\tsbP, \queue{}, \clockN[0]) \mmid (\tsbQ, \queue{}, \clockN[0])$, %
  where $\tsbP \compliant \tsbQ$.
  Then,
  for all finite timed traces~$\vec{\lambda}$
  there exists one and only one $\gamma$ such that
  $\gamma_0 \mcmove{\vec{\lambda}} \gamma$.
}

\begin{lem} \label{lem:tst-monitoring:determinism}
  \lemtstmonitoringdeterminism
\end{lem}
\begin{proof}
  Simple inspection of the rules in~\Cref{fig:tst:monitoring-semantics}.
\end{proof}

The goal of the runtime monitor is to detect, 
at any state of the execution, which of the two participants
is culpable (if any).
Further, we want to identify who is in charge of the next move.
This is formalised by the following~\namecref{def:onduty}.

\begin{defi}[\textbf{Duties \& culpability}]
  \label{def:onduty}
  \label{def:culpable}
  Let $\gamma = (\tsbP,\qmvA,\clockN) \mmid (\tsbQ,\qmvB,\clockE)$.
  We say that $\pmvA$ is \emph{culpable} in $\gamma$
  iff $\tsbP = \cnil$.
  We say that $\pmv A$ is \emph{on duty} in $\gamma$ if 
  \begin{inlinelist}
  \item {\pmv A} is not culpable in $\gamma$, and
  \item either $\tsbP$ is an internal choice and $\qmvA$ is empty, or $\qmvB$ is not empty.
  \end{inlinelist}
\end{defi}

\Cref{lem:onduty-culpable} states that, 
in each reachable configuration,
only one participant can be on duty; 
and if no one is on duty nor culpable, 
then both participants have reached success. %

\begin{lem} 
  \label{lem:onduty-culpable}
  If $\tsbP \compliant \tsbQ$ and
  $(\tsbP,\queue{},\clockN[0]) \mmid (\tsbQ,\queue{},\clockE[0]) \mcmove{}^* \gamma$,
  then:
  \begin{enumerate}
  \item there exists at most one participant on duty in $\gamma$,
  \item if there exists some culpable participants in $\gamma$,
    then no one is on duty in $\gamma$,
  \item if no one is on duty in $\gamma$,
    then $\gamma$ is success, or someone is culpable in $\gamma$.
  \end{enumerate}
\end{lem}
\begin{proof}
  Straightforward by~\Cref{def:culpable}, \Cref{lem:tst-monitoring:determinism},
  and by the rules in~\Cref{fig:tst:monitoring-semantics}.
\end{proof}

Note that both participants may be culpable in a configuration.
For instance, let 
$\gamma = (\TsumI{\atomOut{a}}{\guardTrue}{}, \queue{}, \clockE[0]) \mmid (\TsumE{\atomIn{a}}{\guardTrue}{}, \queue{}, \clockE[0])$.
By applying~\nrule{[M-FailA]} twice, we obtain:
\[
\gamma
\;\mcmove{\pmv A \says {\atomIn b}}\;
(\cnil, \queue{}, \clockN[0]) \mmid (\TsumE{\atomIn{a}}{\guardTrue}{}, \queue{} , \clockE[0])
\;\mcmove{\pmv B \says {\atomIn b}}\;
(\cnil, \queue{}, \clockN[0]) \mmid (\cnil, \queue{}, \clockE[0])
\]
and in the final configuration both participants are culpable.

\begin{exa}
  Let $\tsbP = \atomOut{a}\setenum{2 < \clockT < 4}$ 
  and
  $\tsbQ = \atomIn{a}\setenum{2 < \clockT < 5} + \atomIn{b}\setenum{2 < \clockT < 5}$ 
  be the TSTs of participants $\pmvA$ and $\pmvB$, respectively.
  \iftoggle{techreport}{%
  Participant $\pmvA$ declares that she will send $\atom{a}$ between
  2 and 4 time unit (abbr.\ t.u.),
  while
  $\pmvB$ declares that he is willing to receive $\atom{a}$ or $\atom{b}$ 
  if they are sent within 2 and 5 t.u. 
  }
  {}
  We have that $\tsbP \compliant \tsbQ$.
  Let $\gamma_0 = (\tsbP, \queue{}, \clockN[0]) \mmid (\tsbQ, \queue{}, \clockN[0])$. %

  A correct interaction is given by the timed trace
  $\vec{\lambda} = \seq{1.2 ,\ \pmvA \says \atomOut{a} ,\ \pmvB \says \atomIn{a}}$. %
  Indeed, $\gamma_0 \mcmove{\vec{\lambda}} (\success, \queue{}, \clockN[0]) \mmid (\success,\queue{}, \clockN[0])$.
  On the contrary, things may go awry in three cases: 
  \begin{enumerate}
  \item  a participant does something not permitted. %
    E.g., if $\pmvA$  fires $\atom{a}$ at $1$ t.u.,
    by \nrule{[M-FailA]}:
    $
    \gamma_0 
    \mcmove{1}
    \mcmove{\pmv A \says {\atomOut a}} 
    (\cnil, \queue{},\clockN[0]+1) \mmid (\tsbQ, \queue{}, \clockE[0]+1)
    $, %
    where $\pmvA$ is culpable. %
    
  \item a participant avoids to do something she is supposed to do. %
    E.g., assume that after $6$ t.u., 
    $\pmvA$ has not yet fired $\atom{a}$. %
    By rule~\nrule{[M-FailD]}, 
    $\gamma_0 \;\mcmove{6 } (\cnil, \queue{}, \clockN[0]+6) \mmid (\tsbQ,
    \queue{}, \clockE[0]+6)$, %
    where $\pmvA$ is culpable. %
    
  \item a participant does not receive a message as soon as it is sent. %
    For instance, after $\atom{a}$ is sent at $1.2$ t.u.,
    at $5.2$ \text{t.u.} $\pmvB$ has not yet fired $\atomIn{a}$. %
    By~\nrule{[M-FailD]}, 
    $\gamma_0 \mcmove{1.2}\mcmove{\pmvA \says \atomOut{a}}\mcmove{4} 
    (\success, \queue{\atomOut{a}},
    \clockN[0]+5.2) \mmid (\cnil, \queue{}, \clockE[0]+5.2)$, %
    where $\pmvB$ is culpable. %

  \end{enumerate}

\end{exa}

\noindent To relate the monitoring semantics in~\Cref{fig:tst:turn-based} 
with the one in~\Cref{fig:tst:s_semantics} %
we use the notion \emph{turn-simulation} of~\cite{BCPZ15jlamp}.
This relation is between states of two arbitrary LTSs 
$\rightarrow_1$ and $\rightarrow_2$, 
and it is parameterised over two sets $S_1$ and $S_2$ of success states. %
A state $(s_2, \rightarrow_2)$ turn-simulates $(s_1, \rightarrow_1)$
whenever each move of $s_1$ can be matched by a sequence of moves of $s_2$ 
(ignoring the labels), and stuckness of $s_1$ implies that $s_2$ will get stuck
in at most one step. %
Further, turn-simulation must preserve success. %

\begin{defi}[\textbf{Turn-simulation}~\cite{BCPZ15jlamp}]
  \label{def:turn-simulation}
  For $i \in \setenum{1,2}$, 
  let $\rightarrow_i$ be an LTS over a state space $Z_i$, 
  and let $S_i$ be a set of states of $\rightarrow_i$.
  We say that a relation $\relR \;\subseteq Z_1\times Z_2$ 
  is a \emph{turn-simulation} iff 
  $s_1 \relR s_2$ implies:

  \vbox{
    \begin{enumerate}

    \item \label{item:turn-simulation:a}
      $
      s_1 \rightarrow_1 s_1' 
      \implies 
      \exists s_2' : s_2 \rightarrow_2^* s_2' \text{ and } s_1' \relR s_2'
      $

    \item \label{item:turn-simulation:b}
      $
      s_2 \rightarrow_2 s_2' 
      \implies
      s_1 \rightarrow_1 \text{ or } 
      (s_1 \relR s_2' \text{ and } s_2' \not\rightarrow_2)
      $

    \item \label{item:turn-simulation:c}
      $s_2 \in S_2 \implies s_1 \in S_1$
      
    \end{enumerate}
  } % end vbox
  \noindent
  If there is a turn-simulation between $s_1$ and $s_2$ (written $s_1 \relR s_2$), we say 
  that $s_2$ turn-simulates $s_1$.
  We denote with $\tsim$ the greatest turn-simulation.
  We say that $\relR$ is a \emph{turn-bisimulation} iff both 
  $\relR \subseteq Z_1\times Z_2$ 
  and  $\relR^{-1} \subseteq Z_2 \times Z_1$ are  turn-simulations.
\end{defi}

The following~\namecref{lem:st:turn-bisimilar} establishes
that the two semantics of TSTs
(Definitions~\ref{def:tst:semantics} and~\ref{def:tst:turn-based})
are turn-bisimilar.

\begin{reslemma}{lem:st:turn-bisimilar}
  $((\tsbP,\clockN) \mid (\tsbQ,\clockE), \smove{}) $ is turn-bisimilar to 
  $((\tsbP, \queue{}, \clockN) \mmid  (\tsbQ,  \queue{}, \clockE), \cmove{})$.
\end{reslemma}

When both participants behave honestly,
\ie, they never take~\nrule{[M-Fail*]} moves,
the monitoring semantics preserves compliance (\Cref{lem:st:mcompliance}). 
The \emph{monitoring compliance} relation $\mcompliant$
is the straightforward adaptation of that in~\Cref{def:compliance},
except that $\cmove{}$ transitions are used instead of $\smove{}$ ones
(\Cref{def:mcompliance}).

\begin{defi}[\textbf{Monitoring Compliance}] 
  \label{def:mdeadlock}
  \label{def:mcompliance}
  We say that a monitoring configuration $\gamma$ is \emph{deadlock} whenever
  $(i)$ it is not the case that both $\tsbP$ and $\tsbQ$ in $\gamma$ are $\success$,
  and $(ii)$ there is no $\lambda$ such that $\gamma \cmove{\lambda}$.
  We then write $(\tsbP,\qmvA,\clockN) \mcompliant (\tsbQ,\qmvB,\clockE)$
  whenever:
  \[ 
  (\tsbP, \qmvA, \clockN) \mmid (\tsbQ, \qmvB, \clockE)
  \cmove{}^*  
  \gamma
  \quad \text{ implies } \quad  
  \gamma
  \text{ not deadlock}
  \]
  We write $\tsbP \mcompliant \tsbQ$ 
  whenever 
  $(\tsbP,\queue{},\clockN[0]) \mcompliant (\tsbQ,\queue{},\clockE[0])$.
\end{defi}

\newcommand{\lemstmcompliance}{%
  $\compliant \; = \mcompliant$.
}

\begin{thm} \label{lem:st:mcompliance}
  \lemstmcompliance
\end{thm}
\begin{proof}
  For the inclusion $\subseteq$, 
  assume that $\tsbP \compliant \tsbQ$.
  By~\Cref{lem:st:turn-bisimilar}, the states
  $(\tsbP,\clockN[0]) \mid (\tsbQ,\clockE[0])$ 
  and
  $(\tsbP,\queue{},\clockN[0]) \mmid (\tsbQ,\queue{},\clockE[0])$
  are turn-bisimilar.
  By
  $
  (\tsbP,\queue{},\clockN[0]) \mmid (\tsbQ,\queue{},\clockE[0])
  \tsim
  (\tsbP,\clockN[0]) \mid (\tsbQ,\clockE[0])
  $
  and $(\tsbP,\clockN[0]) \compliant (\tsbQ,\clockN[0])$,
  Lemma 5.9 in~\cite{BCPZ15jlamp}
  implies that
  $(\tsbP,\queue{},\clockN[0]) \mcompliant (\tsbQ,\queue{},\clockE[0])$,
  hence $\tsbP \mcompliant \tsbQ$.
  The other inclusion is similar.
\end{proof}

\section{Conclusions and Related work} \label{sec:related-work}

We have studied a theory of session types (TSTs), featuring
timed synchronous communication between two endpoints.
We have defined a decidable notion of compliance between TSTs,
a decidable procedure to detect when a TST admits a compliant,
a computable canonical compliant construction,
a decidable subtyping relation,
and a decidable runtime monitoring of interactions based on TSTs.

The current article corresponds
to a thoroughly revised and improved version of~\cite{Bartoletti15forte}.
Besides presenting the proofs of all statements,
the current work extends~\cite{Bartoletti15forte} in two main directions.
First, in~\Cref{sec:comput-dual} it presents an alternative set of kinding rules,
which we exploit to show the decidability of kind inference,
and from this the computability of the canonical compliant construction. 
Second, in~\Cref{sec:tst-to-ta} it provides a compliance-preserving encoding of TSTs 
into timed automata, which we exploit in our decision procedure for compliance.

The work~\cite{bettybook17tst}
is a step-by-step tutorial on how to exploit timed session types
for programming contract-oriented distributed applications.

Compliance between TSTs is loosely related to the notion of compliance
between \emph{untimed} session types (in symbols, $\compliant_{\erase{}}$).
Let $\erase{\tsbP}$ be the session type obtained by 
erasing from $\tsbP$ all the timing annotations.
It is easy to check that the semantics of 
$(\erase{\tsbP},\clockN[0]) \mid (\erase{\tsbQ},\clockN[0])$ 
in~\Cref{sec:tst}
coincides with the semantics of 
$\erase{\tsbP} \mid \erase{\tsbQ}$ in~\cite{Barbanera15mscs}.
Therefore, if $\erase{\tsbP} \compliant \erase{\tsbQ}$,
then $\erase{\tsbP} \compliant_{\erase{}} \erase{\tsbQ}$.
Instead, \emph{semantic conservation} of compliance does not hold,
\ie it is not true in general that
if $\tsbP \compliant \tsbQ$, then $\erase{\tsbP} \compliant_{\erase{}} \erase{\tsbQ}$.
E.g., let
$\tsbP = \TsumI{\atomOut{a}}{\clockT < 5}{} \sumInt \TsumI{\atomOut{b}}{\clockT < 0}{}$, 
and let $\tsbQ = \TsumE{\atomIn{a}}{\clockT < 7}{}$.
We have that $\tsbP \compliant \tsbQ$ 
(because the branch $\atomOut{b}$ can never be chosen), 
whereas $\erase{\tsbP} = \atomOut{a} \sumInt \atomOut{b} \not\!\!\!\compliant_{\erase{}} \atomIn{a} = \erase{\tsbQ}$.
Note that, for every $\tsbP$, $\erase{\dual{\tsbP}} = \dual{\erase{\tsbP}}$.

In the context of session types, time has been originally introduced in~\cite{Bocchi14concur}.
However, the setting is different than ours
(multiparty and asynchronous, while ours is bi-party and synchronous),
as well as its objectives: 
while we have focussed on primitives for the bottom-up approach
to service composition~\cite{BTZ12sacs}, 
\cite{Bocchi14concur} extends to the timed case 
the \emph{top-down} approach.
There, a \emph{choreography} 
(expressing the overall communication behaviour of a set of participants)
is projected into a set of \emph{session types}, %
which in turn are refined as processes, to be type-checked against their session type
in order to make service composition preserve the properties enjoyed by the choreography.

Our approach is a conservative extension of untimed session types, 
in the sense that a participant which performs an output action chooses not only the branch, 
but the time of writing too; dually, when performing an input, one has to passively
follow the choice of the other participant.
Instead, in \cite{Bocchi14concur} external choices can also delay the reading time. %
The notion of correct interaction studied in~\cite{Bocchi14concur} is called~\emph{feasibility}: 
a choreography is feasible iff all its reducts can reach the success state. %
This property implies progress, but it is undecidable in general, as shown by~\cite{KrcalY06}
in the context of communicating timed automata
(however, feasibility is decidable for the subclass of \emph{infinitely satisfiable} choreographies).
The problem of deciding if, given a local type $T$, there exists a choreography $G$ such that
$T$ is in the projection of $G$ and $G$ enjoys (global) progress is not addressed in \cite{Bocchi14concur}. %
A possible way to deal with this problem would be to adapt our kind system 
(in particular, rule \nrule{[T-+]}). % must be adjusted).

Another problem not addressed by~\cite{Bocchi14concur} 
is that of determinining if a set of session types enjoys progress
(which, as feasibility of choreographies, would be undecidable).
In our work we have considered this problem, 
under a synchronous semantics, 
and with the restriction of two participants.
Extending our semantics to an asynchronous one 
would make compliance undecidable
(as it is for untimed asynchronous session types~\cite{Denielou13icalp}).

Note that our notion of compliance does not 
imply progress with the semantics of~\cite{Bocchi14concur}
(adapted to the binary case).
For instance, consider the TSTs:
\[
\tsbP \; = \; \TsumE{\atomIn a}{\clockX \leq 2}{\TsumI{\atomOut{a}}{\clockX \leq 1}{}}
\hspace{50pt}
\tsbQ \; = \; \TsumI{\atomOut a}{\clockY \leq 1}{\TsumE{\atomIn{a}}{\clockY \leq 1}{}}
\]
These TSTs are compliant according to~\Cref{def:compliance},
while using the semantics of~\cite{Bocchi14concur}:
\[
(\clockN[0],(\tsbP,\tsbQ,\vec{w_0})) 
\; \pmove{\;\;}^* \;
(\clockN,({\TsumI{\atomOut{a}}{\clockX \leq 1}{}},
\TsumE{\atomIn{a}}{\clockY \leq 1}{},\vec{w_0}))
\tag*{with $\clockN(\clockX) = \clockN(\clockY) > 1$}
\]
This is a deadlocked state, 
hence $\tsbP$ and $\tsbQ$ do not enjoy progress according to~\cite{Bocchi14concur}.
% }

Dynamic verification of timed multiparty session types 
is addressed by~\cite{Neykova14beat},
where the top-down approach to service composition is pursued~\cite{Honda08MPS}.
Our middleware instead composes and monitors services 
in a bottom-up fashion~\cite{BTZ12sacs}.

The work~\cite{Bocchi15concur} studies \emph{communicating timed automata},
a timed version of communicating finite-state machines~\cite{Zafiropulo83jacm}.
In this model, participants in a network communicate asynchronously
through bi-directional FIFO channels;
similarly to~\cite{Bocchi14concur}, clocks, guards and resets are used 
to impose time constraints on when communications can happen.
A system enjoys \emph{progress} when no deadlocked state is reachable,
as well as no orphan messages, unsuccessful receptions, and unfeasible configurations.
Since deadlock-freedom is undecidable in the untimed case, 
then \emph{a fortiori} progress is undecidable for communicating timed automata.
So, the authors propose an approximated (sound, but not complete) decidable technique 
to check when a system enjoys progress. 
This technique is based on \emph{multiparty compatibility},
a condition which guarantees deadlock-freedom of untimed systems~\cite{Lange15popl}.

A classical property of timed systems addressed by~\cite{Bocchi15concur}
is \emph{zenoness}, \ie
the situation in which all the paths from a reachable state are infinite and time-convergent.
Again, multiparty compatibility is exploited by~\cite{Bocchi15concur} 
to devise an approximated decidable technique 
which guarantees non-zenoness of communicating timed automata.

In our setting, an example of zenoness is given by the following TSTs:
\[
\tsbP \; = \; \rec {\tsbX}{\TsumI{\atomOut{a}}{\clockX \leq 1}{\tsbX}}
\hspace{40pt} 
\tsbQ \; = \; \rec {\tsbX}{\TsumE{\atomIn{a}}{\clockX \leq 1,\clockX}{\tsbX}}
\] 

Although $\tsbP \compliant \tsbQ$, the total elapsed time cannot exceed 1. 
This implies that, in order to respect $\tsbP$, 
a participant should have to perform \emph{infinitely} many writings in \emph{a single} time unit. %
This problem can be solved imposing some restrictions to TSTs, in order to have the property that
composition of TSTs is always non-zeno. 
For instance, this is guaranteed by the notion of \emph{strong non-zenoness} of~\cite{Tripakis99}, 
which can be computed efficiently but is not complete.
Another possibility is to check non-zenoness directly in the network of timed automata constructed for checking compliance, using one of the techniques
appeared in the literature~\cite{HenzingerNSY94,Tripakis99,BowmanG06}.

In~\cite{DavidLLNTW15} timed specifications are studied in the setting
of \emph{timed I/O transition systems} (TIOTS). %
They feature a notion of correct composition, called \emph{compatibility},
following the \emph{optimistic approach} pursued in~\cite{Alfaro01sigsoft}:
roughly, two systems are compatible whenever 
there exists an environment which, composed with them,
makes ``undesirable'' states unreachable.
A notion of \emph{refinement} is coinductively
formalised as an alternating timed simulation.
Refinement is a preorder,
and it is included in the semantic subtyping relation
(using compatibility instead of $\compliant$).
Because of the different assumptions
(open systems and broadcast communications in~\cite{DavidLLNTW15},
closed binary systems in TSTs),
compatibility/refinement seem unrelated to our notions of 
compliance/subtyping.
Despite the main notions in~\cite{DavidLLNTW15} are defined on semantic
objects (TIOTS), they can be decided on timed I/O automata, 
which are finite representations of TIOTS. %
With respect to TSTs, timed I/O automata are more liberal:
\eg, they allow for \emph{mixed choices}, 
while in TSTs each state is either an input or an output. %
However, this increased expressiveness does not seem appropriate
for our purposes: first, it makes the concept of culpability unclear
(and it breaks one of our main properties,
\ie that at most one participant is on duty at each execution step); 
second, it seems to invalidate any dual construction.
This is particularly unwelcome, since this construction is
one of the crucial primitives of contract-oriented interactions.
 
\section*{Acknowledgement}

This work has been partially supported by
Aut.\ Reg.\ of Sardinia P.I.A.\ 2013 ``NOMAD'',
 by EU COST Action IC1201
``Behavioural Types for Reliable Large-Scale Software Systems'' (BETTY).
We thank the anonymous reviewers for their insightful comments. 

\bibliographystyle{abbrv}
\bibliography{main}

\appendix

\section{Proofs for~\texorpdfstring{\Cref{sec:tst-compliance}}
  {Section 3}}
\label{sec:proofs-compliance}

\begin{lem}\label{lem:compliance:shape}
  For all $(\tsbP,\clockN),(\tsbQ,\clockE)$ such that $(\tsbP,\clockN) \compliant (\tsbQ,\clockE)$:
  \begin{align*}
    & \tsbP = \textstyle 
    \TSumInt[i \in I]{\atomOut[i]{a}}{\guardG[i],\resetR[i]}{\tsbP[i]} \implies
    \tsbQ = \TSumExt[j \in J]{\atomIn[j]{a}}{\guardG[j],\resetR[j]}{\tsbP[j]}
    \\
    & \tsbP = \textstyle
    \TSumExt[i \in I]{\atomIn[i]{a}}{\guardG[i],\resetR[i]}{\tsbP[i]} \implies
    \tsbQ = \TSumInt[j \in J]{\atomOut[j]{a}}{\guardG[j],\resetR[j]}{\tsbP[j]}
  \end{align*}
\end{lem}
\begin{proof}
  Trivial.
\end{proof}

\begin{lem}\label{lem:coind-compliance:sound}
  Let $\relR$ be a coinductive compliance relation, and let $(\tsbP,\clockN)\relR(\tsbQ,\clockE)$.
  Then:
  \begin{enumerate}
  \item $(\tsbP,\clockN) \mid (\tsbQ,\clockE)$ not deadlock \label{lem:coind-compliance:sound:item:i}
  \item $(\tsbP,\clockN) \mid (\tsbQ,\clockE) \smove{\tau} \gamma \implies
    \existsunique \tsbPi,\tsbQi,\clockNi,\clockEi:\gamma \smove{\tau} 
    (\tsbPi,\clockNi)\mid(\tsbQi,\clockEi) \;\land\; (\tsbPi,\clockNi)\relR(\tsbQi,\clockEi)$
    \label{lem:coind-compliance:sound:item:ii}
  \item $(\tsbP,\clockN) \mid (\tsbQ,\clockE) \smove{\delta} (\tsbPi,\clockNi) \mid (\tsbQi,\clockEi)
    \implies (\tsbPi,\clockNi)\relR(\tsbQi,\clockEi)$
    \label{lem:coind-compliance:sound:item:iii}
  \end{enumerate}
\end{lem}
\begin{proof}
  Assume that $(\tsbP,\clockN)\relR(\tsbQ,\clockE)$. 
  We proceed by cases on the form of $\tsbP$, modulo unfolding of recursion
  (note that $\relR$ does not talk about committed choices, \ie terms of the form 
  ${\todo[\TsumI{\atomOut{a}}{\guardG,\resetR}{}]{\tsbP}}$).
  If $\tsbP=\success$, then by \Cref{def:coind-compliance} we have $\tsbQ=\success$,
  and so all items are trivial. 
  If $\tsbP$ is an internal choice, \ie
  $\tsbP = \TSumInt[i \in I]{\atomOut[i]{a}}{\guardG[i],\resetT[i]}{\tsbP[i]}$,
  by \Cref{def:coind-compliance} we have that
  $\tsbQ = \TSumExt[j \in J]{\atomIn[j]{a}}{\guardG[j],\resetT[j]}{\tsbQ[j]}$.
  Take some $\delta$ and $i$ such that 
  $\clockN + \delta \in \sem{\guardG[i]}$
  (their existence is guaranteed by~\Cref{def:coind-compliance}). 
  Let $\clockNi = \clockN + \delta$ and $\clockEi = \clockE + \delta$.
  Then:
  \begin{equation}\label{lem:coind-compliance:sound:proof:eq1}
    \AXC{$\clockNi \in \sem{\guardG[i]}$}
    \RL{\nrule{[$\sumInt$]}}
    \UIC{$(\TSumInt[i \in I]{\atomOut[i]{a}}{\guardG[i],\resetT[i]}{\tsbP[i]},\clockNi) 
      \smove{\; \tau \; } (\todo[\TsumI{\atomOut[i]{a}}{\guardG[i],\resetT[i]}{}]{\tsbP[i]},\clockNi)$}
    \RL{\nrule{[S-$\sumInt$]}}
    \UIC{$(\TSumInt[i \in I]{\atomOut[i]{a}}{\guardG[i],\resetT[i]}{\tsbP[i]},\clockNi) \mid 
      (\tsbQ,\clockEi)
      \smove {\; \tau \;} (\todo[\TsumI{\atomOut[i]{a}}{\guardG[i],\resetT[i]}{}]{\tsbP[i]},\; \clockNi)
      \mid (\tsbQ,\clockEi)$}
    \DisplayProof
  \end{equation}
  Hence, $(\tsbP,\clockN) \mid (\tsbQ,\clockE)$ is not deadlock, 
  which proves~\cref{lem:coind-compliance:sound:item:i}. 

  For \cref{lem:coind-compliance:sound:item:ii}, assume 
  $(\tsbP,\clockN) \mid (\tsbQ,\clockE) \smove{\tau} \gamma$. The derivation of such step must be as 
  in~\eqref{lem:coind-compliance:sound:proof:eq1}, with $\delta = 0$ and 
  $\gamma = (\todo[\TsumI{\atomOut[i]{a}}{\guardG[i],\resetT[i]}{}]{\tsbP[i]},\; \clockN)
  \mid (\tsbQ,\clockE)$. By \Cref{def:coind-compliance}, there exists $j \in J$
  such that $\atom[i]{a} = \atom[j]{a}$,
  $\clockE \in \sem{\guardG[j]}$ and $\tsbP[i] \relR \tsbQ[j]$.
  Hence:
  \[
  \AXC{}
  \RL{\nrule{[\bang]}}
  \UIC{$(\todo[\TsumI{\atomOut[i]{a}}{\guardG[i],\resetT[i]}{}]{\tsbP[i]},\; \clockN)
    \smove{\; \atomOut[i]{a} \;} (\tsbP[i],\reset{\clockN}{\resetT[i]})$}
  \AXC{$\clockE \in \sem{\guardG[j]}$}
  \RL{\nrule{[\qmark]}}
  \UIC{$(\tsbQ,\clockE) \smove{\; \atomIn[j]{a} \;} (\tsbQ[j],\reset{\clockE}{\resetT[j]})$}
  \RL{\nrule{[S-$\tau$]}}
  \BIC{$\gamma \smove{\; \tau \;} (\tsbP[i],\reset{\clockN}{\resetT[i]}) \mid 
    (\tsbQ[j],\reset{\clockE}{\resetT[j]})$}
  \DisplayProof
  \]
  Note that the target state is unique, because branch labels in choices are pairwise distinct.

  For \cref{lem:coind-compliance:sound:item:iii}, assume
  $(\tsbP,\clockN) \mid (\tsbQ,\clockE) \smove{\delta} (\tsbPi,\clockNi) \mid (\tsbQi,\clockEi)$.
  Its derivation must be as follows:
  \begin{equation}\label{lem:coind-compliance:sound:proof:eq2}
    \AXC{$\clockN + \delta \in \rdy{\tsbP} = \past{\bigcup\sem{\guardG[i]}}$}
    \RL{\nrule{[Del]}}
    \UIC{$(\tsbP,\clockN) \smove{\delta} (\tsbP,\clockN+\delta)$}
    \AXC{$\clockE + \delta \in \rdy{\tsbQ} = \Val$}
    \RL{\nrule{[Del]}}
    \UIC{$(\tsbQ,\clockE) \smove{\delta} (\tsbQ,\clockE + \delta)$}
    \RL{\nrule{[S-Del]}}
    \BIC{$(\tsbP,\clockN) \mid (\tsbQ,\clockE) \smove{\delta} 
      (\tsbP,\clockN+\delta) \mid (\tsbQ,\clockE+\delta)$}
    \DisplayProof
  \end{equation}
  Let $\clockNi = \clockN+\delta$ and $\clockEi = \clockE+\delta$. We have to show 
  $(\tsbP,\clockNi) \relR (\tsbQ,\clockEi)$.
  By~\eqref{lem:coind-compliance:sound:proof:eq2} we have that $\clockNi \in \rdy{\tsbP}$.
  It remains to show that, whenever 
  $\clockNi + \delta' \in \sem{\guardG[i]}$, there exist $j$ such that
  $\atom[i]{a} = \atom[j]{a}$ and $\clockEi + \delta'\in \sem{\guardG[j]}$ and
  $(\tsbP[i],\reset{\clockNi+\delta'}{\resetT[i]})\relR(\tsbQ[j],\reset{\clockEi+\delta'}{\resetT[j]})$.
  This follows by the assumption $(\tsbP,\clockN)\relR(\tsbQ,\clockE)$.
  The case where $\tsbP$ is an external choice is similar.
\end{proof}

\begin{proofof}{lem:coind-compliance}
  We prove the more general statement:
  \[
  (\tsbP,\clockN)  \compliant (\tsbQ,\clockE) \iff 
  \exists \relR \text{coinductive compliance} \; : \; 
  (\tsbP,\clockN) \relR (\tsbQ,\clockE)
  \]
  For the $(\Rightarrow)$ direction we proceed by showing that $\compliant$
  is a coinductive compliance relation.
  Assume $(\tsbP,\clockN) \compliant (\tsbQ,\clockE)$. 
  We show the case where $\tsbP = \TSumInt[i \in I]{\atomOut[i]{a}}{\guardG[i],\resetT[i]}
  {\tsbP[i]}$ (the case of external choice is similar, and the case $\tsbP = \success$ is trivial).
  By \Cref{lem:compliance:shape},
  $\tsbQ = \TSumExt[j \in J]{\atomIn[j]{a}}{\guardG[j],\resetT[j]}{\tsbQ[j]}$. 
  Since $(\tsbP,\clockN) \mid (\tsbQ,\clockE)$ is not deadlock,
  it must be $\clockN \in \rdy{\tsbP}$. 
  By \Cref{def:coind-compliance} it remains to show that, for all $\delta,i$:
  \[
  \clockN + \delta \in \sem{\guardG[i]} \implies 
  \exists j:~\atom[i]{a} = \atom[j]{a} 
  \;\land\;\clockE + \delta \in \sem{\guardG[j]}\;\land\;
  (\tsbP[i],\reset{(\clockN+\delta)}{\resetT[i]})
  \compliant (\tsbQ[j],\reset{(\clockE+\delta)}{\resetT[j]})
  \]
  By contradiction, suppose this is not the case, and take a $\delta$ and an $i$ such that: 
  \[
  \clockN + \delta \in \sem{\guardG[i]} \land (\forall j:~\atom[i]{a} \neq \atom[j]{a} 
  \lor \clockE + \delta \not\in \sem{\guardG[j]} \lor (\tsbP[i],\reset{(\clockN+\delta)}{\resetT[i]})
  \notcompliant (\tsbQ[j],\reset{(\clockE+\delta)}{\resetT[j]})
  \]
  There are two cases. 
  If $\delta = 0$, then:
  \[
  {(\tsbP,\clockN) \mid (\tsbQ,\clockE) \smove{\tau}
    (\todo[\TsumI{\atomOut[i]{a}}{\guardG[i],\resetT[i]}{}]{\tsbP[i]},\clockN)\mid
    (\tsbQ,\clockE) = \gamma}
  \]
  If, for all $j$, $\atom[i]{a} \neq \atom[j]{a} \lor \clockN + \delta \not\in \sem{\guardG[j]}$,
  then $\gamma$ is deadlock. Otherwise:
  \[
  \gamma \smove{\tau} (\tsbP[i],\reset{(\clockN+\delta)}{\resetT[i]}) \mid
  (\tsbQ[j],\reset{(\clockE+\delta)}{\resetT[j]}) = \gamma'
  \]
  with $(\tsbP[i],\reset{(\clockN+\delta)}{\resetT[i]}) \notcompliant 
  (\tsbQ[j],\reset{(\clockE+\delta)}{\resetT[j]})$,
  and so, by \Cref{def:compliance}, there exists some deadlock configuration $\gamma''$ such that
  $\gamma' \smove{}^* \gamma''$. Hence $(\tsbP,\clockN) \notcompliant (\tsbQ,\clockE)$.
  If $\delta > 0$:
  \[
  (\tsbP,\clockN) \mid (\tsbQ,\clockE) \smove{\delta}
  (\tsbP,\clockN + \delta) \mid (\tsbQ,\clockE + \delta)
  \]
  The thesis follows by an argument similar to the case with $\delta = 0$.

  \smallskip\noindent
  The $(\Leftarrow)$ direction is a straightforward consequence of~\Cref{lem:coind-compliance:sound}
  \qed
\end{proofof}

\section{Proofs for~\texorpdfstring{\Cref{sec:tst-duality}}{Section4}}
\label{sec:proofs-duality}

\begin{proofof}{lem:every-tst-kindable}
  We have to prove that:
  \[
  \fv{\tsbP} \subseteq \dom{\Gamma}
  \;\; \implies \;\;
  \existsunique \kindK \; : \;
  \Gamma \vdash \tsbP:\kindK
  \]
  Let $\Gamma$ as in the statement. 
  By induction on the structure of $\tsbP$, we have the following cases:
  \begin{itemize}
  \item $\tsbP = \success$. Trivial. 
  \item $\tsbP$ is a choice (internal or external). Straightforward by the induction hypothesis.
  \item $\tsbP = \tsbX$. 
    Since $\fv{\tsbP} \subseteq \dom{\Gamma}$, 
    the thesis follows by rule $\nrule{[T-Var]}$.
  \item $\tsbP = \rec \tsbX \tsbPi$. 
    Since $\fv{\tsbP} \subseteq \dom{\Gamma}$, 
    then for all $\kindKi$: 
    $
    \fv{\tsbPi} \subseteq \fv{\tsbP}\cup\setenum{\tsbX} \subseteq 
    \dom{\Gamma}\cup\setenum{\tsbX} = \dom{\Gamma,\tsbX:\kindKi}
    $.
    Hence, by the induction hypothesis we have that
    $
    \existsunique \kindKii:\Gamma,\tsbX:\kindKi \vdash \tsbPi : \kindKii
    $.
    The thesis follows by rule $\nrule{[T-Rec]}$.
    \qed
  \end{itemize}
\end{proofof}

\begin{lem}
  \label{lem:gamma-monotonicity:past-invReset}
  For all $\kindK,\kindKi$ such that $\kindK \subseteq \kindKi$, and
  for all sets of clocks $\resetT$:
  \[
  \past{\kindK} \subseteq \past{\kindKi} 
  \qquad \text{ and } \qquad
  \invReset{\kindK}{\resetT} \subseteq \invReset{\kindKi}{\resetT}
  \]
\end{lem}
\begin{proof}
  Straightforward by \Cref{def:past-invreset}.
\end{proof}

Hereafter, we assume substitutions to be capture avoiding.

\begin{lem}[\textbf{Substitution}]
  \label{lem:type-preserving-substitution}
  Let $\Gamma \vdash \tsbPi:\kindKi$. Then, for all $\tsbP$:
  \[
  \Gamma,\tsbX:\kindKi \vdash \tsbP:\kindK 
  \;\;\iff\;\; 
  \Gamma \vdash \tsbP\setenum{\bind{\tsbX}{\tsbPi}}:\kindK
  \]
\end{lem}
\begin{proof}
  We prove that, under the assumption $\Gamma \vdash \tsbPi:\kindKi$, the following items hold: 
  \begin{enumerate}
  \item $\Gamma,\tsbX:\kindKi \vdash \tsbP:\kindK \implies
    \exists \kindKii \supseteq \kindK:~\Gamma \vdash \tsbP\setenum{\bind{\tsbX}{\tsbPi}}:\kindKii$.
    \label{lem:type-preserving-substitution:proof:i}
  \item $\Gamma \vdash \tsbP\setenum{\bind{\tsbX}{\tsbPi}}:\kindK \implies
    \exists \kindKii \supseteq \kindK:~\Gamma,\tsbX:\kindKi \vdash \tsbP:\kindKii$.
    \label{lem:type-preserving-substitution:proof:ii}
  \end{enumerate}
  Before proving the two items, we show that together they imply the thesis.
  Assume that $\Gamma,\tsbX:\kindKi \vdash \tsbP:\kindK$. 
  By~\cref{lem:type-preserving-substitution:proof:i},
  $\Gamma \vdash \tsbP\setenum{\bind{\tsbX}{\tsbPi}}:\kindKii$ for some $\kindKii \supseteq \kindK$. 
  Then, by~\cref{lem:type-preserving-substitution:proof:ii}, 
  $\Gamma,\tsbX:\kindKi \vdash \tsbP:\kindKiii$ 
  for some $\kindKiii \supseteq \kindKii$. 
  Therefore, by uniqueness of kinding,
  $
  \kindKii \subseteq \kindKiii = \kindK \subseteq \kindKii
  $.
  The other direction follows by a similar argument.

  \smallskip\noindent
  To prove~\cref{lem:type-preserving-substitution:proof:i},
  assume $\Gamma,\tsbX:\kindKi \vdash \tsbP:\kindK$.
  We proceed by induction on the typing rules.
  \begin{itemize}

  \item \nrule{[T-$\success$]}, \nrule{[T-Var]}. Trivial.

  \item \nrule{[T-$\oplus$]}. We have:
    \[
    \irule{\Gamma,\tsbX:\kindKi \vdash \tsbP[i] : \kindK[i]}
    {\Gamma,\tsbX:\kindKi \vdash 
      \TSumInt{\atomOut[i]{a}}{\guardG[i],\resetT[i]}{\tsbP[i]}:\big(\bigcup\past{\sem{\guardG[i]}}\big)
      \setminus \big(\bigcup\past{(\sem{\guardG[i]} \setminus \invReset{\kindK[i]}{\resetT[i]})}\big)}
    \]
    By the induction hypothesis:
    \[
    \irule{\Gamma \vdash \tsbP[i]\setenum{\bind{\tsbX}{\tsbPi}} : \tildekindK[i] \supseteq \kindK[i]}
    {\Gamma \vdash \TSumInt{\atomOut[i]{a}}{\guardG[i],\resetT[i]}
      {(\tsbP[i]\setenum{\bind{\tsbX}{\tsbPi}}}):
      \big(\bigcup\past{\sem{\guardG[i]}}\big) \setminus \big(\bigcup\past{(\sem{\guardG[i]} 
        \setminus \invReset{\tildekindK[i]}{\resetT[i]})}\big)}
    \]
    The thesis follows by \Cref{lem:gamma-monotonicity:past-invReset}.

  \item \nrule{[T-+]}. Similar to \nrule{[T-$\oplus$]}.

  \item \nrule{[T-Rec]}. 
    $\tsbP$ must have the form $\rec {\tsbY}{\tsbPii}$. 
    If $\tsbX = \tsbY$ the thesis is trivial; otherwise:
    \[\irule{\exists \kindK[0],\kindKi[0]:~\Gamma,\tsbX:\kindKi,\tsbY:\kindK[0] \vdash \tsbPii : \kindKi[0]}
    {\Gamma,\tsbX:\kindKi \vdash \rec {\tsbY} {\tsbPii}:\bigcup{\setcomp{\tildekindK \supseteq \tildekindKi}
        {\Gamma,\tsbX:\kindKi,\tsbY:\tildekindK \vdash \tsbPii: 
          \tildekindKi}} = \kindK}\]
    Since $\tsbX \neq \tsbY$, then 
    $
    \tsbP\setenum{\bind{\tsbX}{\tsbPi}} = 
    \rec {\tsbY} {(\tsbPii\setenum{\bind{\tsbX}{\tsbPi}})}
    $,
    and for all $\kindK[0],\kindKi[0]$:
    \begin{equation}\label{lem:type-preserving-substitution:proof:rec:ii}
      \Gamma,\tsbX:\kindKi
      ,\tsbY:\kindK[0] \vdash \tsbPi : \kindKi[0]
      \iff
      \Gamma,\tsbY:\kindK[0]
      ,\tsbX:\kindKi \vdash \tsbPi : \kindKi[0]
    \end{equation}
    The kinding derivation for $\tsbP\setenum{\bind{\tsbX}{\tsbPi}}$ must be as follows:
    \begin{equation}\label{lem:type-preserving-substitution:proof:rec:i}
      \irule{\exists \kindK[0],\kindKi[0]:~\Gamma
        ,\tsbY:\kindK[0] \vdash (\tsbPii\setenum{\bind{\tsbX}{\tsbPi}}) : \kindKi[0]}
      {\Gamma \vdash \tsbP\setenum{\bind{\tsbX}{\tsbPi}}:\bigcup{\setcomp{\tildekindK}
          {\Gamma,\tsbY:\tildekindK \vdash (\tsbPii\setenum{\bind{\tsbX}{\tsbPi}}): 
            \tildekindKi \land \tildekindK \subseteq \tildekindKi}} = \kindKii}
    \end{equation}
    We first show that the premise of~\eqref{lem:type-preserving-substitution:proof:rec:i} holds.
    Since substitutions are capture avoiding, $\tsbY$ is not free in $\tsbPi$, and hence
    $\Gamma,\tsbY:\kindK \vdash \tsbPi:\kindKi$ as well as $\Gamma$.
    Then, by~\Cref{lem:type-preserving-substitution:proof:rec:ii} together with the induction
    hypothesis, the thesis follows. 
    It remains to show $\kindK \subseteq \kindKii$. 
    If $\kindK$ is empty the thesis holds trivially. 
    Otherwise, take some $\tildekindK,\tildekindKi$ such that:
    \[
    \Gamma,\tsbX:\kindKi,\tsbY:\tildekindK \vdash \tsbPii: 
    \tildekindKi\land \tildekindK \subseteq \tildekindKi
    \]
    By~\eqref{lem:type-preserving-substitution:proof:rec:ii}, together with the induction
    hypothesis, we have that, for some $\tildekindKii$:
    \[\Gamma,\tsbY:\tildekindK \vdash \tsbPii\setenum{\bind{\tsbX}{\tsbPi}}:
    \tildekindKii \supseteq \tildekindKi\] from which the thesis follows.
  \end{itemize}

  \noindent
  To prove~\cref{lem:type-preserving-substitution:proof:ii}, assume
  $\Gamma \vdash \tsbP\setenum{\bind{\tsbX}{\tsbPi}}:\kindK$. We proceed by induction on the
  structure of $\tsbP$.
  \begin{itemize}
  \item $\tsbP=\success$. Trivial.
  \item $\tsbP=\TSumInt{\atomOut[i]{a}}{\guardG[i],\resetT[i]}{\tsbP[i]}$. 
    \[
    \irule{\Gamma \vdash \tsbP[i]\setenum{\bind{\tsbX}{\tsbPi}} : \kindK[i]}
    {\Gamma \vdash \TSumInt{\atomOut[i]{a}}{\guardG[i],\resetT[i]}
      {(\tsbP[i]\setenum{\bind{\tsbX}{\tsbPi}}}):
      \bigcup\past{\sem{\guardG[i]}} \setminus \bigcup\past{(\sem{\guardG[i]} \setminus
        \invReset{\kindK[i]}{\resetT[i]})}}\nrule{[T-$\oplus$]}
    \]
    By induction hypothesis:
    \[
    \irule{\Gamma\setenum{\bind{\tsbX}{\kindK[i]}} \vdash \tsbP[i] : \tildekindK[i] \supseteq \kindK[i]}
    {\Gamma,\tsbX:\kindKi \vdash \TSumInt{\atomOut[i]{a}}{\guardG[i],\resetT[i]}
      {\tsbP[i]}:\bigcup\past{\sem{\guardG[i]}} \setminus 
      \bigcup\past{(\sem{\guardG[i]} \setminus \invReset{\tildekindK[i]}{\resetT[i]})}}\nrule{[T-$\oplus$]}
    \]
    The thesis follows by~\Cref{lem:gamma-monotonicity:past-invReset}.

  \item $\tsbP=\TSumExt{\atomIn[i]{a}}{\guardG[i],\resetT[i]}{\tsbP[i]}$.
    Similar to the internal choice case.

  \item $\tsbP=\tsbY$. If $\tsbY \neq \tsbX$ the thesis follows trivially. Otherwise,
    let $\tsbP = \tsbX$. Then $\Gamma\vdash\tsbP\setenum{\bind{\tsbX}{\tsbPi}}=\tsbPi:\kindKi$ and
    $\Gamma,\tsbX:\kindKi \vdash \tsbP = \tsbX:\kindKi$. 

  \item $\tsbP = \rec {\tsbY}{\tsbPii}$. 
    Assume $\tsbX \neq \tsbY$ (otherwise the thesis holds trivially). 
    \[
    \irule{\exists \kindK[0],\kindKi[0]:~\Gamma
      ,\tsbY:\kindK[0] \vdash (\tsbPii\setenum{\bind{\tsbX}{\tsbPi}}) : \kindKi[0]}
    {\Gamma \vdash \tsbP\setenum{\bind{\tsbX}{\tsbPi}}:\bigcup{\setcomp{\tildekindK}
        {\Gamma,\tsbY:\tildekindK \vdash (\tsbPii\setenum{\bind{\tsbX}{\tsbPi}}): 
          \tildekindKi \land \tildekindK \subseteq \tildekindKi}} = \kindK}
    \]
    As in the proof of~\eqref{lem:type-preserving-substitution:proof:i},
    we have $\Gamma,\tsbY:\kindK \vdash \tsbPi:\kindKi$. 
    Then, by~\eqref{lem:type-preserving-substitution:proof:rec:ii} and the induction hypothesis:
    \[
    \irule{\exists \kindK[0],\kindKi[0]:~\Gamma,\tsbX:\kindKi
      ,\tsbY:\kindK[0] \vdash \tsbPii : \kindKi[0]}
    {\Gamma,\tsbX:\kindKi \vdash \rec {\tsbY} {\tsbPii}:\bigcup{\setcomp{\tildekindK}
        {\Gamma,\tsbX:\kindKi,\tsbY:\tildekindK \vdash \tsbPii: 
          \tildekindKi \land \tildekindK \subseteq \tildekindKi}} = \kindKii}
    \]
    It remains to prove $\kindK \subseteq \kindKii$. If $\kindK$
    is empty the thesis holds trivially. Otherwise, take some $\tildekindK,\tildekindKi$ such that:
    \[
    \Gamma,\tsbY:\tildekindK \vdash \tsbPii\setenum{\bind{\tsbX}{\tsbPi}}:
    \tildekindKi\land \tildekindK \subseteq \tildekindKi
    \]
    By~\eqref{lem:type-preserving-substitution:proof:rec:ii}, together with the induction
    hypothesis, we have that, for some $\tildekindKii$:
    \[
    \Gamma,\tsbX:\kindKi,\tsbY:\tildekindK \vdash \tsbPii:
    \tildekindKii \supseteq \tildekindKi\] from which the thesis follows.
    \qedhere
  \end{itemize}
\end{proof}

\begin{lem}[\textbf{Substitution under dual}]
\label{lem:dual-preserving-substitution}
  Let $\Gamma \vdash {\tsbPi}:\kindK$, with $\tsbX$ not free in $\tsbPi$.
  Then, for all $\tsbP$ such that
  $\tsbP$ is kindable with environment $\Gamma,\tsbX:\kindK$:
  \[\dual[\Gamma,\tsbX:\kindK]{\tsbP}\setenum{\bind{\tsbX}{\dual[\Gamma]{{\tsbPi}}}}
  =  \dual[\Gamma]{\tsbP\setenum{\bind{\tsbX}{{\tsbPi}}}}
  \]
\end{lem}
\begin{proof}
  By induction on the structure of $\tsbP$:
  \begin{itemize}
  \item $\tsbP = \TSumInt{\atomOut[i]{a}}{\guardG[i],\resetT[i]}{\tsbP[i]}$:\\[5pt]
    $\dual[\Gamma,\tsbX:\kindK]{\tsbP}\setenum{\bind{\tsbX}
      {\dual[\Gamma]{{\tsbPi}}}} = $
    \[\begin{array}{cll}
      = & \TSumExt{\atomIn[i]{a}}{\guardG[i],\resetT[i]}{(\dual[\Gamma,\tsbX:\kindK]
        {\tsbP[i]}\setenum{\bind{\tsbX}{\dual[\Gamma]{{\tsbPi}}}})}\\
      = & \TSumExt{\atomIn[i]{a}}{\guardG[i],\resetT[i]}{\dual[\Gamma]{\tsbP[i]
          \setenum{\bind{\tsbX}{{\tsbPi}}}}} & \text{by induction hypothesis}\\
      = & \dual[\Gamma]{\tsbP\setenum{\bind{\tsbX}{{\tsbPi}}}}
    \end{array}\]
  \item $\tsbP = \TSumExt{\atomIn[i]{a}}{\guardG[i],\resetT[i]}{\tsbP[i]}$: 
    Let $\Gamma,\tsbX:\kindK\vdash\tsbP[i]:\kindK[i]$ for all $i \in I$. Then:
    \\[5pt]
    $\dual[\Gamma,\tsbX:\kindK]{\tsbP}
    \setenum{\bind{\tsbX}{\dual[\Gamma]{{\tsbPi}}}} = $
    \[\begin{array}{cll}
      = & \TSumInt{\atomOut[i]{a}}{\guardG[i] \land \invReset{\kindK[i]}{\resetT[i]},\resetT[i]}
      {(\dual[\Gamma,\tsbX:\kindK]{\tsbP[i]}
        \setenum{\bind{\tsbX}{\dual[\Gamma]{{\tsbPi}}}})}\\
      = & \TSumInt{\atomOut[i]{a}}{\guardG[i] \land \invReset{\kindK[i]}{\resetT[i]},\resetT[i]}
      {\dual[\Gamma]{\tsbP[i]
          \setenum{\bind{\tsbX}{{\tsbPi}}}}} & \text{by induction hypothesis}\\
      = & \dual[\Gamma]{\tsbP\setenum{\bind{\tsbX}{{\tsbPi}}}} 
      & \text{by~\Cref{lem:type-preserving-substitution}}
    \end{array}\]
  \item $\tsbP = \success$: Trivial
  \item $\tsbP = \tsbY$: Trivial, whether or not $\tsbY = \tsbX$
  \item $\tsbP = \rec {\tsbY}{\tsbPii}$: If $\tsbY = \tsbX$ the thesis is trivial. Otherwise, assume
    $\Gamma,\tsbX:\kindK \vdash \rec {\tsbY}{\tsbPii}: \kindKi$. First note that, since we are assuming
    capture avoiding substitutions, $\tsbY$ must not be free in $\tsbPi$. Indeed, if $\tsbY$ were free in $\tsbPi$,
    then $\tsbY$ would be captured in $\tsbP\setenum{\bind{\tsbX}{{\tsbPi}}}$.
    Then:\\[5pt]
    $\dual[\Gamma,\tsbX:\kindK]{\tsbP}
    \setenum{\bind{\tsbX}{\dual[\Gamma]{{\tsbPi}}}} = $
    \[
    \begin{array}{cll}
      = & \rec {\tsbY}{(\dual[\Gamma,\tsbX:\kindK,\tsbY:\kindKi]
        {\tsbPii} \setenum{\bind{\tsbX}{\dual[\Gamma]{{\tsbPi}}}})}&\text{by~\Cref{def:dual}}\\
      = & \rec {\tsbY}{(\dual[\Gamma,\tsbY:\kindKi,\tsbX:\kindK]{\tsbPii}
        \setenum{\bind{\tsbX}{\dual[\Gamma]{{\tsbPi}}}})}&
      \text{since } \tsbX\neq\tsbY\\
      = & \rec {\tsbY}{(\dual[\Gamma,\tsbY:\kindKi,\tsbX:\kindK]{\tsbPii}
        \setenum{\bind{\tsbX}{\dual[\Gamma,\tsbY:\kindKi]{{\tsbPi}}}})}
      & \text{since $\tsbY$ is not free in $\tsbPi$}\\
      = & \rec {\tsbY}{\dual[\Gamma,\tsbY:\kindKi]
        {\tsbPii\setenum{\bind{\tsbX}{{\tsbPi}}}}}
      & \text{by induction hypothesis}\\
      = & \dual[\Gamma,\tsbY:\kindKi]{\tsbP\setenum{\bind{\tsbX}{{\tsbPi}}}}&\text{by~\Cref{def:dual}}\\
      = & \dual[\Gamma]{\tsbP\setenum{\bind{\tsbX}{{\tsbPi}}}}
      & \text{since $\tsbY$ is not free in $\tsbP\setenum{\bind{\tsbX}{{\tsbPi}}}$}
    \end{array}
    \]
  \end{itemize}
  Note that the induction hypothesis is trivially applicable: $\tsbP$ is kindable in $\Gamma,\tsbX:\kindK$, necessarily by rule \nrule{[T-Rec]}. The premise
  of rule \nrule{[T-Rec]} implies $\tsbPii$ is kindable in $\Gamma,\tsbX:\kindK,\tsbY:\kindKi$.
\end{proof}

\begin{defi}\label{def:struct-congruence}
We define the structural congruence relation $\equiv$ between TSTs 
as the least congruence relation closed under the following equation:
\[
\rec {\tsbX} {\tsbP} \equiv \tsbP\setenum{\bind{\tsbX}{\rec {\tsbX} {\tsbP}}}
\]
\end{defi}

\begin{lem}\label{lem:dual:struct}
  For all kindable $\tsbP$ and $\tsbQ$, we have that:
  $
  \tsbP \equiv \tsbQ \implies \dual{\tsbP} \equiv \dual{\tsbQ}
  $.
\end{lem}
\begin{proof}
  The thesis follows by the following more general statement: 
  \[
  \forall \Gamma \;\; : \;\;
  \forall \tsbP,\tsbQ \textit{ kindable in $\Gamma$} \;\; : \;\;
  \tsbP \equiv \tsbQ \implies \dual[\Gamma]{\tsbP} \equiv \dual[\Gamma]{\tsbQ}
  \]
  The proof is by easy induction on the structure of $\tsbP$, 
  using~\Cref{lem:dual-preserving-substitution} in the case $\tsbP = \rec {\tsbX} {\tsbPi}$, which is
  applicable because the substitution $\tsbPi\setenum{\bind{\tsbX}{\rec {\tsbX}{\tsbPi}}}$ does not capture the free variables in $\rec {\tsbX}{\tsbPi}$.
\end{proof}

\begin{lem}\label{lem:unfold-shape}
  Every closed $\tsbP$ is structurally equivalent to a TST of the following shape:
  \[
  \begin{array}{ccc}
    \success 
    \hspace{40pt}
    &
    \TSumInt{\atomOut[i]{a}}{\guardG[i],\resetT[i]}{\tsbP[i]}
    \hspace{40pt}
    &
    \TSumExt{\atomIn[i]{a}}{\guardG[i],\resetT[i]}{\tsbP[i]}
  \end{array}
  \]
\end{lem}
\begin{proof}
  Trivial.
\end{proof}

\begin{proofof}[Soundness]{th:dual-sound}
  Let:
  \[\relR = \setcomp{((\tsbP,\clockN),(\dual{\tsbP},\clockN))}
  {\exists \kindK:\;\Gamma \vdash \tsbP:\kindK \land \clockN \in \kindK}\]
  By~\Cref{lem:coind-compliance}, it is enough to show that $\relR$ is a coinductive compliance relation
  (\Cref{def:coind-compliance}).
  Assume that $\Gamma \vdash \tsbP:\kindK$ and $\clockN \in \kindK$. 
  We proceed by cases on the shape of $\tsbP$.
  According to~\Cref{lem:unfold-shape}, we have the following cases:
  \begin{itemize}
  \item $\tsbP \equiv \success$. 
    By~\Cref{lem:dual:struct} $\dual{\tsbP} \equiv \success$, from which the thesis follows.

  \item $\tsbP \equiv \TSumInt{\atomOut[i]{a}}{\guardG[i],\resetT[i]}{\tsbP[i]}$. 
    By \Cref{lem:dual:struct},
    $
    \dual{\tsbP} \equiv \TSumExt{\atomIn[i]{a}}{\guardG[i],\resetT[i]}
    {\dual{\tsbP[i]}}
    $.
    By~rule \nrule{[T-$\oplus$]} it follows that $\nu \in \rdy{\tsbP}$.
    To conclude:
    \[\forall \delta,i:~\clockN + \delta \in \sem{\guardG[i]} \implies 
    (\tsbP[i],\reset{(\clockN + \delta)}{\resetT[i]})\relR
    (\dual{\tsbP[i]},\reset{(\clockN+\delta)}{\resetT[i]})\]
    Let $\vdash \tsbP[i]:\kindK[i]$. By the typing rule \nrule{[T-$\oplus$]} we have that:
    \begin{align*}
      \kindK & = 
      \big(\bigcup_{i \in I}\past{\sem{\guardG[i]}}\big) \setminus 
      \big(\bigcup_{i \in I}\past{(\sem{\guardG[i]} \setminus
        \invReset{\kindK[i]}{\resetT[i]})\big)}
      \\
      & =
      \big(\bigcup_{i \in I}\past{\sem{\guardG[i]}}\big) \setminus 
      \setcomp{\clockN}{\exists \delta,i:\clockN + \delta \in \sem{\guardG[i]} \land
        \reset{\clockN+\delta}{\resetT[i]} \not\in \kindK[i]}
    \end{align*}
    from which the thesis follows.

  \item $\tsbP \equiv \TSumExt{\atomIn[i]{a}}{\guardG[i],\resetT[i]}{\tsbP[i]}$.
    By \Cref{lem:dual:struct},
    $
    \dual{\tsbP} = \TSumInt{\atomOut[i]{a}}{\guardG[i] \land \kindK[i],\resetT[i]}{\dual{\tsbP[i]}}
    $
    with $\vdash \tsbP[i]:\kindK[i]$.
    By rule \nrule{[T-+]}: 
    \[\kindK = \bigcup \past{\big(\sem {\guardG[i]} \cap \invReset{\kindK[i]}{\resetT[i]}\big)}
    = \setcomp{\clockN}{\exists \delta,i:~\clockN + \delta \in \sem{\guardG[i]}
      \land \reset{(\clockN + \delta)}{\resetT[i]}\in \kindK[i]}\] 
    from which the thesis follows. 
  \end{itemize}
\end{proofof}

\begin{defi}\label{def:has-compliant-nu}
  We define the function $\hcn$ from TSTs to sets of clock valuations as follows:
  \[
  \hcn[\tsbP]\mmdef\setcomp{\clockN}{\exists\tsbQ,\clockE:(\tsbP,\clockN)\compliant(\tsbQ,\clockE)}
  \]
\end{defi}

\begin{lem}\label{lem:dual-complete:aux1}
  Let $\Gamma = \tsbX[1]:\hcn[{\tsbP[1]}],\dots,\tsbX[m]:\hcn[{\tsbP[m]}]$, where all $\tsbP[i]$ are closed, and let
  $\vec{\tsbX} = (\tsbX[1],\dots,\tsbX[m]),\vec{\tsbP}=(\tsbP[1],\dots,\tsbP[m])$.
  Then, for all $\tsbP$ such that $\fv{\tsbP} \subseteq \vec{\tsbX}$:
  \[
  \Gamma \vdash \tsbP:\kindK \implies \hcn[\tsbP\setenum{\bind{\vec{\tsbX}}{\vec{\tsbP}}}] 
  \subseteq \kindK
  \]
\end{lem}
\begin{proof}
  We start with an auxiliary definition.
  The recursion nesting level ($\recLevel$) for a TST is inductively defined as follows:
  \iftoggle{lmcs}{%
    \[
    \begin{array}{ll}
      \recLevel[\success] \mmdef \recLevel[\tsbX] \mmdef 0
      &
      \recLevel[{\TSumExt[i \in I]{\atomIn[i]{a}}{\guardG[i],\resetT[i]}{\tsbP[i]}}] 
      \mmdef max~\setenum{\recLevel[{\tsbP[i]}]}_{i \in I}
      \\
      \recLevel[\rec{\tsbX}{\tsbP}] \mmdef 1 + \recLevel[\tsbP]
      &
      \recLevel[{\TSumInt[i \in I]{\atomOut[i]{a}}{\guardG[i],\resetT[i]}{\tsbP[i]}}] 
      \mmdef max~\setenum{\recLevel[{\tsbP[i]}]}_{i \in I}
    \end{array}
    \]
  }
  {%
    \[
    \begin{array}{rcl}
      \recLevel[\success] & \mmdef & 0\\
      \recLevel[{\TSumExt[i \in I]{\atomIn[i]{a}}{\guardG[i],\resetT[i]}{\tsbP[i]}}] & \mmdef 
      & max~\setenum{\recLevel[{\tsbP[i]}]}_{i \in I}\\
      \recLevel[{\TSumInt[i \in I]{\atomOut[i]{a}}{\guardG[i],\resetT[i]}{\tsbP[i]}}] & \mmdef
      & max~\setenum{\recLevel[{\tsbP[i]}]}_{i \in I}\\
      \recLevel[\tsbX] & \mmdef & 0 \\
      \recLevel[\rec{\tsbX}{\tsbP}] & \mmdef & 1 + \recLevel[\tsbP]
    \end{array}
    \]
  }
  \newcommand{\altsqleq}{\prec} %to be moved in macro.tex
  We then define the relation $\altsqleq$ as:
  \[
  \tsbP\;\altsqleq\;\tsbQ \text{ whenever } 
  \recLevel[\tsbP] < \recLevel[\tsbQ] \lor
  (\recLevel[\tsbP] = \recLevel[\tsbQ] \land \tsbP \text{ is a strict subterm of } \tsbQ)
  \]
  It is trivial to check that $\altsqleq$ is a well-founded relation 
  (exploiting the fact that the strict subterm relation is well-founded as well).
  We then proceed by well-founded induction on~$\altsqleq$. 
  We have the following cases,
  according to the form of $\tsbP$:
  \begin{itemize}
  \item $\tsbP = \success$. Since $\kindK = \Val$ (by kinding rule \nrule{[T-$\success$]}),
    the thesis follows trivially.
  \item $\tsbP = \tsbX[i]$, for some $i \in \setenum{1,\dots,m}$.
    $\kindK = \Gamma(\tsbX[i]) = \hcn[{\tsbP[i]}]
    = \hcn[{\tsbX[i]\setenum{\bind{\vec{\tsbX}}{\vec{\tsbP}}}}]$.
  \item $\tsbP = \TSumExt[i \in I]{\atomIn[i]{a}}{\guardG[i],\resetT[i]}{\tsbP[i]}$.
    \[
    \irule{\Gamma \vdash \tsbP[i]: \kindK[i] \qquad \text{for }i \in I}
    {\Gamma \vdash \TSumExt[i \in I]{\atomIn[i]{a}}{\guardG[i],\resetT[i]}{\tsbP[i]}: 
      \bigcup_{i \in I} \past{\big(\sem {\guardG[i]} \cap
        \invReset{\kindK[i]}{\resetT[i]}\big)} = \kindK} 
    \; \nrule{[T-$\sumExt$]}
    \]
    Since, for all $i \in I$ it holds $\tsbP[i]\;\altsqleq\;\tsbP$, by the induction hypothesis:
    \begin{equation}\label{lem:dual-complete:aux1:proof:eq1}
      \hcn[{\tsbP[i]\setenum{\bind{\vec{\tsbX}}{\vec{\tsbP}}}}]\subseteq \kindK[i]
    \end{equation}
    Now suppose, by contradiction, $\hcn[\tsbP\setenum{\bind{\vec{\tsbX}}{\vec{\tsbP}}}]
    \not\subseteq \kindK$. Then there exist $\clockN,\clockE,\tsbQ$ such that 
    $\hcn[\tsbP\setenum{\bind{\vec{\tsbX}}{\vec{\tsbP}}}] \ni \clockN \not\in \kindK$ and
    $(\tsbP\setenum{\bind{\vec{\tsbX}}{\vec{\tsbP}}},\clockN) \compliant (\tsbQ,\clockE)$.
    By~\Cref{def:coind-compliance} we have 
    $\tsbQ = \TSumInt[j \in J]{\atomOut[j]{a}}{\guardG[j],\resetT[j]}{\tsbQ[j]}$, with 
    $\clockE \in \rdy{\tsbQ}$ and
    \[
    \forall \delta,j:~\clockE + \delta \in \sem{\guardG[j]}\implies 
    \]
    \[
    \exists i : \atom[i]{a} = \atom[j]{a}\land \clockN + \delta \in \sem{\guardG[i]} \land
    (\tsbP[i]\setenum{\bind{\vec{\tsbX}}{\vec{\tsbP}}},\reset{(\clockN+\delta)}{\resetR[i]}) 
    \compliant (\tsbQ[j],\reset{(\clockE+\delta)}{\resetR[j]})
    \]
    But then, by \Cref{def:has-compliant-nu} and \Cref{lem:dual-complete:aux1:proof:eq1}:
    \[\reset{(\clockN+\delta)}{\resetR[i]}\in\hcn[{\tsbP[i]\setenum{\bind{\vec{\tsbX}}{\vec{\tsbP}}}}]
    \subseteq \kindK[i]\]
    Since $\kindK$ is past closed (\ie for all $\clockN,\delta$: 
    $\clockN + \delta \in \kindK \implies \clockN \in \kindK$),
    it follows $\nu \in \kindK$, a contradiction.
  \item $\tsbP = \TSumInt[i \in I]{\atomOut[i]{a}}{\guardG[i],\resetT[i]}{\tsbP[i]}$. Similar
    to the external choice case.
  \item $\tsbP = \rec \tsbY \tsbPi$. By rule \nrule{[T-Rec]}:
    \[
    \irule{\Gamma,\tsbY:\kindK[0]
      \vdash \tsbPi : \kindKi[0]}
    {\Gamma \vdash \rec \tsbY \tsbPi:\bigcup{\setcomp{\kindK}
        {\exists \kindKi \supseteq \kindK:\;\Gamma,\tsbY:\kindK 
          \vdash \tsbPi : \kindKi}} = \kindK} 
    \; 
    \]
    To prove $\hcn[\tsbP\setenum{\bind{\vec{\tsbX}}{\vec{\tsbP}}}] \subseteq \kindK$, 
    it is enough to show that, for some $\kindKi \supseteq \hcn[\tsbP\setenum{\bind{\vec{\tsbX}}{\vec{\tsbP}}}]$:
    \[
    \Gamma,\tsbY:\hcn[\tsbP\setenum{\bind{\vec{\tsbX}}{\vec{\tsbP}}}] 
    \vdash \tsbPi: \kindKi
    \]
    Since compliance is preserved by unfolding of recursion (because TSTs are considered up-to $\equiv$ in \Cref{def:tst:semantics}),
    by \Cref{def:has-compliant-nu} it follows that:
    \[\hcn[\tsbP\setenum{\bind{\vec{\tsbX}}{\vec{\tsbP}}}] = 
    \hcn[\tsbPi\setenum{\bind{\vec{\tsbX}\setminus\setenum{\tsbY}}{\vec{\tsbP}}}\setenum{\bind{\tsbY}{\rec{\tsbY}{(\tsbPi\setenum{\bind{\vec{\tsbX}\setminus \setenum{\tsbY}}{\vec{\tsbP}}})}}}]
    = \hcn[\tsbPi\setenum{\bind{\vec{\tsbXi}}{\vec{\tsbPi}}}]\] 
    where $\vec{\tsbXi}$ and $\vec{\tsbPi}$ are defined as follows:
    \[
    \hspace{40pt}
    \vec{\tsbXi} = 
    \begin{cases}
      \vec{\tsbX}\tsbY & \text{if } \tsbY \not\in \vec{\tsbX}\\
      \vec{\tsbX} &\text{otherwise}
    \end{cases}
    \qquad
    \vec{\tsbPi} = 
    \begin{cases}
      (\tsbP[1],\dots,\tsbP[i-1],
      \tsbP\setenum{\bind{\vec{\tsbX}}{\vec{\tsbP}}},\tsbP[i+1],\dots,\tsbP[m]) 
      &\text{if } \tsbY = \tsbX[i]\\
      \vec{\tsbP}(\tsbP\setenum{\bind{\vec{\tsbX}}{\vec{\tsbP}}}) &\text{otherwise}
    \end{cases}
    \]
    Since $\recLevel[\tsbPi] < \recLevel[\tsbP]$, by the induction hypothesis we have:
    \[
    \Gamma,\tsbY:\hcn[\tsbPi\setenum{\bind{\vec{\tsbXi}}{\vec{\tsbPi}}}] 
    \vdash \tsbPi: \kindKi \implies \hcn[\tsbPi\setenum{\bind{\vec{\tsbXi}}{\vec{\tsbPi}}}] \subseteq \kindKi
    \]
    Since $\fv{\tsbP} \subseteq \vec{\tsbX} = \dom{\Gamma}$ it follows that $\fv{\tsbPi} \subseteq \vec{\tsbXi} = \dom{\Gamma,\tsbY:\hcn[\tsbPi\setenum{\bind{\vec{\tsbXi}}{\vec{\tsbPi}}}]}$,
    and hence, by \Cref{lem:every-tst-kindable}, the premise of the induction hypothesis is satisfied, and we can conclude $\hcn[\tsbPi\setenum{\bind{\vec{\tsbXi}}{\vec{\tsbPi}}}] \subseteq \kindKi$.
    \qedhere
  \end{itemize}
\end{proof}

\begin{proofof}[Completeness]{th:dual-complete}
  Suppose $\vdash \tsbP : \kindK$ and $\exists \tsbQ,\clockE .\; (\tsbP,\clockN) \compliant (\tsbQ,\clockE)$. 
  By instantiating~\Cref{lem:dual-complete:aux1} with the empty kinding environment, we obtain
  $\hcn[\tsbP] \subseteq \kindK$. 
  Hence, by \Cref{def:has-compliant-nu} we conclude that $\clockN \in \kindK$. 
  \qed
\end{proofof}

\begin{proofof}{lem:compliance-trans}
  We prove that the relation:
  \[
  \relR 
  \;\; = \;\;
  \setcomp
  {((\tsbP,\clockN),(\tsbQ,\clockE))}
  {\exists \tsbPi,\clockNi:(\tsbP,\clockN)
    \compliant(\tsbPi,\clockNi)\land (\dual{\tsbPi},\clockNi)\compliant(\tsbQ,\clockE)}
  \]
  is a coinductive compliance relation (\Cref{def:coind-compliance}). 
  We proceed by cases on the form of~$\tsbP$:
  \begin{itemize}

  \item $\tsbP = \success$. 
    Trivial.

  \item $\tsbP = \TSumInt{\atomOut[i]{a}}{\guardG[i],\resetT[i]}{\tsbP[i]}$. 
    It must be:
    \[
    \hspace{40pt}
    \tsbPi = \TSumExt{\atomIn[j]{a}}{\guardG[j],\resetT[j]}{\tsbPi[j]}
    \quad
    \text{and} 
    \quad
    \dual{\tsbPi} = \TSumInt{\atomOut[j]{a}}{\guardG[j] \land \invReset{\kindK[j]}{\resetT[j]},
      \resetT[j]}{\dual{\tsbPi[j]}}
    \]
    with $\vdash\tsbP[j]:\kindK[j]$ for all $j \in J$ and, by \Cref{def:coind-compliance},
    for all $\delta,i$ such that $\clockN + \delta \in \sem{\guardG[i]}$,
    there exists $j$ such that $\atom[i]{a} = \atom[j]{a}$,
    $\clockNi + \delta \in \sem{\guardG[j]}$,
    $
    (\tsbP[i],\reset{(\clockN+\delta)}{\resetT[i]}) 
    \compliant
    (\tsbPi[j],\reset{(\clockNi+\delta)}{\resetT[j]})
    $,
    and
    $\clockN + \delta \in \rdy{\tsbP}$.
    Hence,
    $
    \tsbQ = \TSumExt{\atomIn[k]{a}}{\guardG[k],\resetT[k]}{\tsbP[k]}
    $,
    and for all $\delta,j$ such that $\clockNi + \delta \in \sem{\guardG[j]}$,
    there exists $k$ such that $\atom[k]{a} = \atom[j]{a}$,
    $\clockE + \delta \in\sem{\guardG[k]}$,
    $
    (\dual{\tsbPi[j]},\reset{(\clockNi+\delta)}{\resetT[j]})
    \compliant
    (\tsbQ[k],\reset{(\clockE+\delta)}{\resetT[k]})
    $,
    and $\clockNi\in\rdy{\dual{\tsbPi}}$.
    Now, assume $\clockN + \delta \in \guardG[i]$ for some $\delta,i$, and suppose,
    by contradiction, $\clockNi + \delta \not\in \invReset{\kindK[j]}{\resetT[j]}$ for some $j$
    such that $\atom[i]{a} = \atom[j]{a}$. 
    But then we should have
    $(\tsbP[i],\reset{(\clockN+\delta)}{\resetT[i]})
    \compliant(\tsbPi[j],\reset{(\clockNi+\delta)}{\resetT[j]})$ and
    $\reset{(\clockNi+\delta)}{\resetT[j]}\not\in \kindK[j]$ 
    --- contradiction by~\Cref{th:dual-complete}.

  \item $\tsbP = \TSumExt{\atomIn[i]{a}}{\guardG[i],\resetT[i]}{\tsbP[i]}$.
    It must be:
    \[
    \tsbPi = \TSumInt{\atomOut[j]{a}}{\guardG[j],\resetT[j]}{\tsbPi[j]}
    \qquad
    \text{ and } 
    \qquad
    \dual{\tsbPi} = \TSumExt{\atomIn[j]{a}}{\guardG[j],\resetT[j]}{\dual{\tsbP[j]}}
    \]
    and for all $\delta,j$ such that 
    $
    \clockNi + \delta\in\sem{\guardG[j]}
    $,
    there exists $i$ such that $\atom[j]{a} = \atom[i]{a}$,
    $\clockN + \delta\in\sem{\guardG[i]}$,
    $
    (\tsbPi[j],\reset{\clockNi+\delta}{\resetT[j]}) \compliant (\tsbP[i],\reset{\clockN+\delta}{\resetT[j]})
    $,
    and
    $\clockNi\in\rdy{\tsbPi}$.
    Hence,
    $
    \tsbQ = \TSumInt{\atomIn[k]{a}}{\guardG[k],\resetT[k]}{\tsbP[k]}
    $,
    and for all $\delta,k$ such that $\clockE + \delta\in\sem{\guardG[k]}$,
    there exists $j$ such that $\atom[k]{a} = \atom[j]{a}$,
    $\clockNi + \delta\in\sem{\guardG[j]}$,
    $
    (\tsbQ[k],\reset{(\clockE+\delta)}{\resetT[k]}) \compliant (\tsbPi[j],\reset{(\clockN+\delta)}{\resetT[i]})
    $,
    and
    $\clockE\in\rdy{\tsbQ}$.
    The thesis follows by the composition of the above.
    \qed
  \end{itemize}
\end{proofof}

\section{Proofs for~\texorpdfstring{\Cref{sec:comput-dual}}{Section 8}}
\label{sec:proofs-comput-dual}

\begin{thm}[\bf Knaster-Tarski fixed point theorem~\cite{tarski1955}] \label{th:knaster-tarski}
  Let $\mathcal{U} = (U,\leq)$ be a complete lattice, and let $f$ be a monotonic endofunction over $U$.
  Then $(P,\leq)$, where $P$ is the set of fixed points of $f$, is a complete lattice. In particular, we have:
  \[
  \bigsqcup P = \bigsqcup \setcomp{x}{x \leq f(x)} 
  \hspace{50pt}
  \bigsqcap P = \bigsqcap \setcomp{x}{f(x) \leq x} 
  \]
\end{thm}

\begin{defi}[\bf Cocontinuous function] \label{def:cocontinuity}
  An endofunction $f$ over a complete lattice $D$ is said to be cocontinuous iff,
  for all sequences $d_0,d_1,\hdots$ of decreasing points in $D$, we have:
  \[
  \textstyle
  f(\bigsqcap_i d_i) \; = \; \bigsqcap_i f(d_i)
  \]
\end{defi}

\begin{proofof}{lem:domain-theory}
  For~\cref{item:finite-monotonic-is-cocontinuous},
  let $d_0,d_1,\hdots$ be a decreasing sequence of elements in $D$. 
  By finiteness of $D$, there must exist some $n$ such that, for all $n'$,
  $d_n = d_{n + n'}$. 
  Clearly, $f(\bigsqcap_i d_i) = f(d_n)$. 
  By monotonicity, the sequence $f(d_0),f(d_1),\hdots$ is decreasing, and its meet
  has to be $f(d_n)$. 

  \smallskip\noindent
  For~\cref{item:gfp-smaller-glb},
  we show, by induction on $n$, that, for all $n \geq 0$,
  $
  \gfp f \leq f^n(\top)
  $.
  The base case is trivial, since $f^0(\top) = \top$.
  For the induction case, suppose that $\gfp f \leq f^n(\top)$. 
  By monotonicity of $f$ we have that:
  $
  \gfp f = f(\gfp f) \leq f(f^n(\top)) = f^{n+1}(\top)
  $.

  \smallskip\noindent
  For~\cref{item:sangiorgi}, see~\cite{Sangiorgi12bi}.
  \qed
\end{proofof}

\begin{proofof}{lem:gamma-monotonicity}
  We define the partial order $\sqsubseteq$ between environments as follows:
  \[
  \Gamma \sqsubseteq \Gamma' \iff \forall \tsbX \in \dom{\Gamma}:\Gamma(\tsbX) \subseteq \Gamma'(\tsbX)
  \]
  and we show that,
  for all $\Gamma,\Gamma'$ such that $\Gamma \sqsubseteq \Gamma':$
  \begin{equation} \label{eq:gamma-monotonicity}
    \Gamma \vdash \tsbP:\kindK \implies \exists \kindKi \supseteq \kindK:~
    \Gamma' \vdash \tsbP:\kindKi
  \end{equation}
  \iftoggle{lmcs}{%
    This can be done by easy induction on the structure of $\tsbP$.
    To prove the main statement,
  }
  {%
  Suppose $\Gamma \sqsubseteq \Gamma'$ and $\Gamma \vdash \tsbP:\kindK$.
  By induction on the structure of $\tsbP$: 
  \begin{itemize}
  \item $\tsbP = \success$: Trivial.
  \item $\tsbP = \TSumExt[i \in I]{\atomIn[i]{a}}{\guardG[i],\resetT[i]}{\tsbP[i]}$:
    By rule $\nrule{[T-+]}$, it must be:
    \[
    \kindK = \bigcup_{i \in I} \past{\big(\sem {\guardG[i]} \cap \invReset{\kindK[i]}{\resetT[i]}\big)}
    \text{, with } \Gamma \vdash \tsbP[i]:\kindK[i]
    \]
    Then, by the induction hypothesis, $\forall i \in I: \exists \kindKi[i] \supseteq \kindK[i]$ such that: 
    \[\Gamma' \vdash \tsbP[i]:\kindKi[i]\]
    The thesis follows by \Cref{lem:gamma-monotonicity:past-invReset}.
  \item $\tsbP = \TSumInt[i \in I]{\atomOut[i]{a}}{\guardG[i],\resetT[i]}{\tsbP[i]}$: Similar to
    the external choice case.
  \item $\tsbP = \tsbX$: Trivial by definition of $\sqsubseteq$ and kinding rule $\nrule{[T-Var]}$.
  \item $\tsbP = \rec {\tsbX} {\tsbPi}$: By rule $\nrule{[T-Rec]}$, it must be, for some $\kindK[0],\kindKi[0]$:
    \begin{equation}\label{eq:gamma-monotonicity-K}
      \irule{\Gamma,\tsbX:\kindK[0] \vdash \tsbPi : \kindKi[0]}
      {\Gamma \vdash \rec \tsbX \tsbPi:\bigcup{\setcomp{\kindK}{\exists \kindKi \supseteq \kindK:~\Gamma,\tsbX:\kindK 
            \vdash \tsbPi : \kindKi}} = \kindK}
    \end{equation}
    By the induction hypothesis, for some $\kindKii[0]$:
    \begin{equation}\label{eq:gamma-monotonicity-K'}
      \irule{\Gamma',\tsbX:\kindK[0] \vdash \tsbPi : \kindKii[0]}
      {\Gamma' \vdash \rec \tsbX \tsbPi:\bigcup{\setcomp{\kindK}{\exists \kindKi \supseteq \kindK:~\Gamma',\tsbX:\kindK 
            \vdash \tsbPi : \kindKi}} = \kindKi}
    \end{equation}
    It remains to show $\kindK \subseteq \kindKi$. To see this, take some $\clockN \in \kindK$. By \cref{eq:gamma-monotonicity-K}, 
    there exist some $\kindK[0], \kindKi[0]$, with $\clockN \in \kindK[0] \subseteq \kindKi[0]$ such that:
    \[
    \Gamma:\tsbX,\kindK[0] 
    \vdash \tsbPi : \kindKi[0]
    \]
    By the induction hypothesis, there is some $\kindKi[1] \supseteq \kindKi[0]$:
    \[
    \Gamma',\tsbX:\kindK[0] 
    \vdash \tsbPi : \kindKi[1]
    \]
    Then, by \cref{eq:gamma-monotonicity-K'}, $\kindK[0] \subseteq \kindKi$, from which the thesis follows.
  \end{itemize}
  Back to the main statement,
  }
  let $\fv{\tsbP} \subseteq \dom{\Gamma} \cup \setenum{\tsbX}$, and suppose $\kindK \subseteq \kindKi$.
  We have to show $F(\kindK) \subseteq F(\kindKi)$. 
  Expanding~\eqref{def:kinding-functional}: 
  $\Gamma,\tsbX:\kindK \vdash \tsbP:\kindK[0]$ and $\Gamma,\tsbX:\kindKi \vdash \tsbP:\kindK[1]$, 
  for some $\kindK[0],\kindK[1]$ such that $\kindK[0] \subseteq \kindK[1]$.
  Such $\kindK[0],\kindK[1]$ exist and are unique (\Cref{lem:every-tst-kindable}); 
  $\kindK[0] \subseteq \kindK[1]$ follows by~\eqref{eq:gamma-monotonicity}.
  \qed
\end{proofof}

\begin{proofof}{lem:judments-equivalence}
  By structural induction on $\tsbP$. 
  The only non-trivial case is when $\tsbP = \rec{\tsbX}{\tsbPi}$.

  \smallskip\noindent
  For the $(\Rightarrow)$ direction, suppose that:
  \[
  \irule{\Gamma,\tsbX:\kindKii \vdash \tsbPi : \kindKiii}
  {\Gamma \vdash \rec \tsbX \tsbPi:\bigcup{\setcomp{\kindK[0]}{\exists \kindK[1] \supseteq \kindK[0]:~\Gamma,\tsbX:\kindK[0] 
        \vdash \tsbPi : \kindK[1]}} = \kindK}
  \; \nrule{[T-Rec]}
  \]
  and recall that, by \Cref{th:knaster-tarski} and \Cref{lem:gamma-monotonicity}, $\kindK = \gfp{F_{\Gamma,\tsbX,\tsbP}}$.
  By the induction hypothesis, for all zones $\zoneZ$: 
  \begin{equation}\label{eq:judments-equivalence-aux1}
    F_{\Gamma,\tsbX,\tsbPi}(\zoneZ) = \hat{F}_{\Gamma,\tsbX,\tsbPi}(\zoneZ)
  \end{equation}
  Hence, \Cref{lem:gamma-monotonicity} implies that $\hat{F}_{\Gamma,\tsbX,\tsbPi}$ is monotonic 
  (on the lattice of zones). 
  Since this lattice is finite, then
  $\hat{F}_{\Gamma,\tsbX,\tsbPi}$ is cocontinuous 
  (\cref{item:finite-monotonic-is-cocontinuous} of~\Cref{lem:domain-theory}).
  By~\eqref{eq:judments-equivalence-aux1} and~\Cref{lem:every-tst-kindable},
  $\hat{F}_{\Gamma,\tsbX,\tsbP}^i(\kindK[0])$ is defined for all $\kindK[0]$ and $i$. 
  Hence, by rule~\nrule{[I-Rec]} 
  and~\cref{item:sangiorgi} of~\Cref{lem:domain-theory}:
  \begin{equation} \label{eq:judments-equivalence-aux2}
    \textstyle
    \Gamma \vdash_I \rec \tsbX \tsbPi:\bigsqcap_{i \geq 0} \hat{F}_{\Gamma,\tsbX,\tsbPi}^i (\Val) 
    \; = \; 
    \hat{\kindK}
    \; = \;
    \gfp{\hat{F}_{\Gamma,\tsbX,\tsbPi}}
  \end{equation}
  To conclude that $\kindK = \hat{\kindK}$ it is enough to show that $\gfp{F_{\Gamma,\tsbX,\tsbPi}} = \hat{\kindK}$.
  By~\eqref{eq:judments-equivalence-aux1} and~\eqref{eq:judments-equivalence-aux2} 
  it follows that $\hat{\kindK}$ is a fixed point of $F_{\Gamma,\tsbX,\tsbPi}$,
  and so by definition of gfp, $\hat{\kindK} \subseteq \gfp{F_{\Gamma,\tsbX,\tsbPi}}$.
  By~\eqref{eq:judments-equivalence-aux1} 
  and~\cref{item:gfp-smaller-glb} of~\Cref{lem:domain-theory}
  we conclude that $\gfp{F_{\Gamma,\tsbX,\tsbPi}} \subseteq \hat{\kindK}$.

  \smallskip\noindent
  For the $(\Leftarrow)$ direction, suppose that by rule~\nrule{[I-Rec]} we have:
  \[
  \textstyle
  \Gamma \vdash_I \rec \tsbX \tsbPi:\bigsqcap_{i \geq 0} \hat{F}_{\Gamma,\tsbX,\tsbPi}^i (\Val)
  \]
  Since $\bigsqcap_{i \geq 0} \hat{F}_{\Gamma,\tsbX,\tsbPi}^i (\Val)$ is defined, 
  the induction hypothesis gives
  $\Gamma, \tsbX:\Val \vdash \tsbPi : F_{\Gamma,\tsbX,\tsbPi}(\Val)$,
  hence the premise of rule \nrule{[T-Rec]} is satisfied. 
  We can conclude with the same argument as in the previous case. 
\end{proofof}

\section{Proofs for~\Cref{sec:tst-to-ta}} 
\label{sec:proofs-tst-to-ta}

\begin{defi}[\bf \denf transformation]\label{def:mapping:denf}
  Let $\tsbP$ be a TST  according to~\Cref{def:tst:syntax}. 
  Let $V$ be an infinite set of recursion variables not occurring in $\tsbP$.  
  Then, the \denf of $\tsbP$,
  denoted by $\nf{\tsbP}$, is given by :
  \[
  \begin{array}{lllll}
    (a) & \nf{ \ocircle_{i = 1}^n \labL[i] \setenum{\guardG[i],\resetR[i]}.\tsbP[i]}
    &=& \ (\tsbX[0], \ \bigcup_{i=1}^n D_i \cup \setenum{\tsbX[0] \triangleq \ocircle_{i=1}^{n} \labL[i] \setenum{\guardG[i],\resetR[i]}.\tsbX[i]})  
    \vspace{4pt}\\
    & & & \qquad \quad \text{ for } \ocircle \in \setenum{\oplus,\sum},   \text{ with } \tsbX[0] \in V  \text{ and }\\   
    & & & \qquad \quad (\tsbX[i], D_i) = \nf[V \setminus W_i]{\tsbP[i]}   \text{and } \\
    & & & \qquad \quad  W_i = (\setenum{\tsbX[0]}\  \cup \ \bigcup_{j=1}^{i-1} \rv{D_j})\\
    (b) & \nf{ \rec{\tsbX}{\tsbPi} }  &=& (\tsbX[0], D[\bind{\tsbX}{\tsbX[0]}] )  \quad  \text{where } (\tsbX[0],D) =   \nf{\tsbPi}  \\
    (c) & \nf{ \tsbX } &=& (\tsbX, \emptyset)\\
    (d) & \nf{ \success } &=&  (\tsbX[0],\setenum{\tsbX[0] = \success}) \quad \text{where } \tsbX[0] \in V
  \end{array}
  \] 
  \noindent For short, we indicate the normal form of $\tsbP$, with  $\nf{\tsbP}$.
\end{defi}

\begin{exa}
  Consider the following translations of TSTs, where  $V$ is a set of
  recursion variables not occurring in any $\tsbP[i]$
  (we omit guards and reset sets). %  to focus on the transformation.
  \[
  \begin{array}{rcl}
    \nf{\atomOut{a} \oplus \atomOut{b}} 
    & = &
    (\tsbX[0] , \setenum{\tsbX[0]  \eqdef \atomOut{a}.\tsbX[1] \oplus \atomOut{b}.\tsbX[2],\,
      \tsbX[1]  \eqdef \success,  \,
      \tsbX[2] \eqdef \success})
    \vspace{6pt}
    \\
    \nf{\atomOut{a}.\rec{\tsbX}{( \atomOut{b}.\tsbX \oplus \atomOut{c}) }} 
    & = &
    (\tsbX[0] , \setenum{\tsbX[0]  \eqdef \atomOut{a}.\tsbX[1],\,
      \tsbX[1]  \eqdef \atomOut{b}.\tsbX[1] \oplus \atomOut{c}.\tsbX[2],\,
      \tsbX[2] \eqdef \success})
    \vspace{6pt}
    \\
    \nf{\rec{\tsbX}{\atomOut{a}.\rec{\tsbY}{( \atomOut{b}.\tsbX \oplus \atomOut{c}.\tsbY) }}} 
    & = &
    (\tsbX[0] , \setenum{\tsbX[0]  \eqdef \atomOut{a}.\tsbX[1],\,
      \tsbX[1]  \eqdef \atomOut{b}.\tsbX[0] \oplus \atomOut{c}.\tsbX[1]} )
    \vspace{6pt}
    \\
    \nf{\rec{\tsbX}{( \atomOut{a}.\rec{\tsbX}{\atomOut{b}.\tsbX}\ \oplus\ \atomOut{c}.\tsbX) } } 
    & = &
    (\tsbX[0] , \setenum{\tsbX[0]  \eqdef \atomOut{a}.\tsbX[1] \oplus \atomOut{c}.\tsbX[0],\,
      \tsbX[1]  \eqdef \atomOut{b}.\tsbX[1]})
  \end{array}
  \]
\end{exa}

\begin{lem}
  \label{lem:tst-to-denf:closed}
  If $\tsbP$ is  closed, then $\nf{\tsbP}$ is closed.
\end{lem}
\begin{proof}\emph{(Sketch).}
  A \denf is closed if its used recursion variables plus
  the initial one are defined exactly once. In the transformation, we
  see by construction that all the new variable are first declared.
  The only open issue is caused by variables already presents in the
  TST. However, if a TST is closed, then all its used variables are in the
  scope of some $\recsym$ declaration, and hence they will be correctly
  renamed by rule \ref{def:mapping:denf}$(b)$.
\end{proof}

\begin{proofof}{lem:preserv_compliance}\emph{(Sketch).}
  Rules in $\smove{}$ and $\demove{}$ have a strong correspondence:
  first of all, all the terms in configurations but one are the same,
  secondly every move but one in one system corresponds to one move in
  the other.  The only exception concerns rule \nrule{[$\sumInt$]} in
  $\smove{}$ which corresponds to  pair (\nrule{\erule},
  \nrule{\crule}) in $\demove{}$, and vice-versa. 
  Hence, the proof is done proceeding by absurd: if only one of the
  systems were deadlock and the other not, the other could still move
  and since rules and configurations are the same, also could the
  first one, leading to a contradiction.
\end{proofof}

\begin{lem}\label{lem:mapping:inv}
  Let $(\tsbX,D)$ be a \denf, and 
  let $\autA = \sem{D}^{\tsbX} = (\Loc,\UrgLoc,\locInit, \Edg, \Inv)$. 
  Then:
  \[
  \Inv(\locL) = 
  \begin{cases}
    \rdy{\tsbP} & \text{if } \locL = \tsbTauY, \text{ for some }  \tsbY \text{ and } \tsbX \eqdef \tsbP \in D  \\
    true    & otherwise
  \end{cases}
  \tag{$\forall \locL \in \Loc$}
  \]    
\end{lem}
\begin{proof} % (\emph{Sketch})
  By~\Cref{def:aut_composition}, the invariant of $\autA$ is
  given by $\Inv(\locL) = \bigwedge_i \Inv[i] (\locL)$ for all~$\locL$, where each
  $\Inv[i]$ is the invariant of the encoding of some defining equation in $D$. %
  By~\Cref{mapping}, the only location with invariant other than $\guardTrue$ has the form $\tsbTauY$,
  which can only be obtained by encoding $\tsbY \eqdef \tsbP$.
  Since $(\tsbX,D)$ is closed, there is exactly one defining equation for $\tsbY$. 
  Hence, $\tsbTauY$ occurs only in one \TA (say, in $\autA[j]$), and so
  $\Inv(\locL) = \bigwedge_i \Inv[i] (\locL) = \Inv[j] (\locL) = \rdy{\tsbP}$.
\end{proof}

\iftoggle{lmcs}{}{%
In \Cref{bis} we recall the definition of strong bisimulation. 

\begin{defi}[\bf Strong bisimulation] \label{bis} 
  A binary relation  $\mathcal{R}$ between states of an LTS is a  bisimulation if whenever 
  whenever $s_1 \mathcal{R} s_2$, and $\alpha$
  is an action:
  \begin{enumerate}
  \item if $s_1  \xrightarrow{\alpha} s_1'$ then there
    is a transition  $s_2 \xrightarrow{\alpha} s_2'$ such that  $s_1'
    \mathcal{R} s_2'$
  \item if $s_2  \xrightarrow{\alpha} s_2'$ then there
    is a transition  $s_1 \xrightarrow{\alpha} s_1'$ such that  $s_1'
    \mathcal{R} s_2'$
  \end{enumerate}
  Two states $s_1$ and $s_2$ are  bisimilar, written $s_1 \sim s_2$,
  iff there is a  bisimulation that relates them. Henceforth the
  relation $\sim$ will be referred to strong bisimulation equivalence or strong 
  bisimilarity.
\end{defi}
}

\begin{proofof}{lem:bsim}
  Let $(\tsbX,D')$ and $(\tsbY,D'')$ as in the statement, and
  let $\trn{(\tsbX,D')}  = (\Loc[1],\UrgLoc[1],\locInit[1], \Edg[1], \Inv[1])$ 
  and $\trn{(\tsbY,D'')} = (\Loc[2],\UrgLoc[2],\locInit[2], \Edg[2], \Inv[2])$.
  We show that:
  \begin{equation}\label{lem:bsim:rel}
    \relR = \setcomp{((x,\clockN)|(y,\clockE), \quad (x,y,\clockN \sqcup \clockE))  }
    {x,y,z \in \states \text{ and }  \clockN \in \sem{\Inv[1](x)} \text{ and }  \clockE \in \sem{\Inv[2](y)} }
  \end{equation}
  is a bisimulation.
  We denote with $\ck{D}$ the set of clocks used in $D$.
  First, we show that every possible move from
  $(x,\clockN)|(y,\clockE)$ in $\demove{}$ is matched by
  a move of $(x,y,\clockN \cup \clockE)$ in $\umove{}$, and
  that the resulting states are related by~$\relR$.
  We have the following cases, according to the rule used to move:
  \begin{description} 
  \item[\nrule{\serule}]
    \[
    \irule{(x,\clockN) \demove{\; \tau \; }
      (x',\clockN)} {(x,\clockN) \mid (y,\clockE) \demove {\; \tau
        \;}(x',\clockN) \mid (y,\clockE)} 
    \]
    According to this rule, we
    have two sub cases in which the premise can fire a $\tau$ move:
    \begin{itemize} 
    \item \nrule{\erule}. 
      We have that
      $(x,\clockN) \demove{ \tau }(\tsbTauX,\clockN)$
      with $x = \tsbX$ and $\tsbX \eqdef \tsbP \in D$ with $\tsbP =
      \oplus \ldots $ and $\clockN \in \rdy{\tsbP}$.
      By~\Cref{mapping}, the mapping of an internal choice
      is
      $\sem{ \tsbX \triangleq \tsbP} = \pfx(\tsbX, \tau, \emptyset,
      \br(\tsbTauX , \rdy{\tsbP}, \setenum{(\tau, \guardG[i],
        \emptyset, \autA[i])}_i)$
      for some $\autA[i]$ and $\guardG[i]$.  
      Hence, by definition of the \pfx\ pattern, there exists an edge in
      $\Edg[1]$ such as
      $(\tsbX, \tau, \guardTrue, \emptyset, \tsbTauX)$.  
      By definition of the \br\ pattern and by~\Cref{lem:mapping:inv},
      $\Inv[1]({\tsbTauX}) = \rdy{\tsbP}$.
      In our case rule [TA1] of~\Cref{net_sem} states that:
      \[
      \irule{(\tsbX, \tau, \guardTrue, \emptyset, \tsbTauX) \in
        \Edg[1] \qquad \clockN \sqcup \clockE \in \sem{\guardTrue}
        \qquad \reset{(\clockN \sqcup \clockE)}{\emptyset} \in
        \sem{\Inv[1](\tsbTauX) \land \Inv[2](y)} } { (\tsbX,
        y,\clockN \sqcup \clockE) \umove{\quad \tau \quad } (\tau
        \tsbX , y, \clockN \sqcup \clockE) }
      \]
      We must show that $\clockN \sqcup \clockE \in \sem{\guardTrue}$
      and $\clockN \sqcup \clockE \in \sem{\Inv[1](x) \land
        \Inv[2](y)}$.
      The former holds trivially. For the second, according
      to~\Cref{lem:mapping:inv} we have
      $\Inv[1]({\tsbTauX}) = \rdy{\tsbP}$; by premises in
      \nrule{\serule}\ we have $\clockN \in \rdy{\tsbP}$ and,
      since $\ck{D'} \cap \ck{D''} = \emptyset$, we have
      $\clockN \sqcup \clockE \in \rdy{\tsbP}$.
      $\clockN \sqcup \clockE \in \sem{\Inv[2](y)}$ holds by
      hypothesis in the definition of $\relR$ in~\Cref{lem:bsim:rel}.
      Hence $ { (\tsbX, y,\clockN \sqcup \clockE) \umove{ \tau } (\tau
        \tsbX , y, \clockN \sqcup \clockE) }$ and the resulting states
      belong to $\relR$.

    \item[\nrule{\crule}]
      We have that
      $(x,\clockN)|(y,\clockE) \demove{\tau}
      (x',\clockN)|(y,\clockE) $ with $x = \tsbTauX$ and $\tsbX
      \eqdef \TsumI{\atomOut{a}}{\guardG,\resetR}{\tsbY} \oplus \tsbPi$
      and $\nu \in \sem{\guardG}$ and $x' =
      \todo[\TsumI{\atomOut{a}}{\guardG,\resetR}{}]{\tsbY}$.
      By~\Cref{mapping}, the mapping of an internal choice
      is 
      \[ \sem{ \tsbX \triangleq \tsbP} = \pfx(\tsbX, \tau,
      \emptyset,  \br(\tsbTauX , \rdy{\tsbP}, \setenum{(\tau, \guardG, \emptyset,  \autA)} \cup S)
      \]
      for some set $S$, and with 
      $\autA =  \pfx( \todo[\atomOut{a} \setenum{\guardG, \resetR}]{\tsbY}, \atomOut{a}, \resetR, \idle(\tsbY))$. 
      By definition of the \br\ pattern
      and~\Cref{lem:mapping:inv}, there exists an edge in $\Edg[1]$
      such as $(\tsbTauX, \tau, \guardG, \emptyset, x')$ with
      $\Inv[1](x') = \guardTrue$.
      In our case rule [TA1] of~\Cref{net_sem} states that:
      \[
      \irule{(\tsbTauX, \tau, \guardG, \emptyset, x') \in \Edg[1]
        \qquad \clockN \sqcup \clockE \in \sem{\guardG} \qquad
        \reset{(\clockN \sqcup \clockE)}{\emptyset} \in \sem{\Inv[1](x')
          \land \Inv[2](y)} } { (\tsbTauX, y,\clockN \sqcup \clockE)
        \umove{\quad \tau \quad } (x', y, \clockN \sqcup \clockE) }
      \]
      We must show that $\clockN \sqcup \clockE \in \sem{\guardG}$
      and $\clockN \sqcup \clockE \in \sem{\Inv[1](x') \land
        \Inv[2](y)}$.
      The former holds by hypothesis since  $\nu \in \sem{\guardG}$. 
      For the second, according to~\Cref{lem:mapping:inv} we have
      $\Inv[1]({\tsbTauX}) = \guardTrue$-- hence
      $\clockN \sqcup \clockE \in \sem{\Inv[1](x')}$ is trivially
      true; and $\clockN \sqcup \clockE \in \sem{\Inv[2](y)}$ holds by
      hypothesis in the definition of $\relR$ in~\Cref{lem:bsim:rel}.
      Hence $ { (\tsbTauX, y,\clockN \sqcup \clockE) \umove{ \tau } (x'
        , y, \clockN \sqcup \clockE) }$ and the resulting states
      belong to $\relR$. 
    \end{itemize}

  \item \nrule{\sdrule}.
    \[
    \irule 
    {(x,\clockN) \demove{\; \delta \; } (x,\clockNi) 
      \quad (y,\clockE) \demove{\; \delta \; } (y,\clockEi)} 
    {(x,\clockN) \mid (y,\clockE) \demove{\; \delta \;} (x,\clockNi) \mid (y,\clockEi)} 
    \]
    According to this rule, we have several possible combinations of the
    premises to fire a $\delta$ move:
    \begin{itemize}

    \item \nrule{\idrule} applied twice.
      \[
      \irule{ 
        \irule{ \tsbX \eqdef \tsbP \in D \quad \clockN + \delta \in \rdy{\tsbP}} {(\tsbTauX,\; \clockN)\demove{\; \delta \; }(\tsbTauX,\; \clockN+\delta)} \qquad 
        \irule{ \tsbY \eqdef
          \tsbQ \in D \quad \clockN + \delta \in \rdy{\tsbQ}} {(\tau
          \tsbY,\; \clockN)\demove{\; \delta \; }(\tsbTauY,\;
          \clockN+\delta)} } 
      { {(\tsbTauX,\clockN) \mid (\tau
          \tsbY,\clockE) \demove{\; \delta \;} (\tsbTauX,\clockN+\delta)
          \mid (\tsbTauY,\clockE+\delta)} }
      \]
      By~\Cref{mapping} and~\Cref{lem:mapping:inv}, there is
      only a case in which a location name has prefix $\tau$, and in
      that case we have $\Inv[1](\tsbTauX) = \rdy{\tsbP}$ and
      $\tsbTauX \not\in \UrgLoc[1]$. Similarly,
      $\Inv[2](\tsbTauY) = \rdy{\tsbQ}$ and
      $\tsbTauY \not\in \UrgLoc[2]$ .
      In our case rule [TA3] of~\Cref{net_sem} states that:
      \[
      \irule{
        ((\clockN  \sqcup \clockE)+ \delta) \in \sem{\Inv[1](\tsbTauX) \land \Inv[2](\tsbTauY)} \qquad \tsbTauX \not\in \UrgLoc[1] \qquad  \tsbTauY \not\in \UrgLoc[2]
      }
      {(\tsbTauX, \tsbTauY,\clockN \sqcup \clockE) \umove{\quad \delta \quad } (\tsbTauX, \tsbTauY, (\clockN \sqcup \clockE) + \delta)}
      \]

      By \nrule{\idrule}\ rule, we have $\clockN + \delta \in \rdy{\tsbP}$,
      hence, since $\ck{D'} \cap \ck{D''} = \emptyset$ and
      so $(\clockN \sqcup \clockE) + \delta \in \rdy{\tsbP}$.  Similarly, since $\clockE + \delta \in
      \rdy{\tsbQ}$, it follows $(\clockN \sqcup \clockE) + \delta \in
      \rdy{\tsbQ}$. We already noted that $\tsbTauY$ and $\tsbTauX$
      are not urgent, so we conclude $(\tsbTauX, y, \clockN \sqcup
      \clockE) \umove{\delta} (\tsbTauX, y, \clockN + \delta \sqcup
      \clockE + \delta)$. The resulting states belong to $\relR$.
      
      % 2 

    \item \nrule{\idrule}\ and \nrule{\edrule}.
      \[
      \irule{ 
        \irule{  (\tsbX \eqdef   \ldots ) \in D  \lor  (\tsbX \eqdef   \success ) \in D } 
        {( \tsbX,\; \clockN)\demove{\; \delta \; }(\tsbX,\; \clockN+\delta)} \qquad 
        \irule{ \tsbY \eqdef
          \tsbQ \in D \quad \clockN + \delta \in \rdy{\tsbQ}} {(\tau
          \tsbY,\; \clockN)\demove{\; \delta \; }(\tsbTauY,\;
          \clockN+\delta)} } 
      { {(\tsbX,\clockN) \mid (\tau
          \tsbY,\clockE) \demove{\; \delta \;} (\tsbX,\clockN+\delta)
          \mid (\tsbTauY,\clockE+\delta)} }
      \]
      By~\Cref{mapping} and~\Cref{lem:mapping:inv}, there is
      only a case in which a location name has prefix $\tau$, and in
      that case we have $\Inv[2](\tsbTauY) = \rdy{\tsbQ}$ and
      $\tsbTauY \not\in \UrgLoc[2]$ .
      For $\tsbX$ we have two possibilities: either (i) it is an
      internal choice or (i) it is a success term.  
      
      \medskip
      (i) In the first case, according with~\Cref{mapping}
      and~\Cref{lem:mapping:inv}, the mapping for an external choice is:
      $\sem{ \tsbX \triangleq \tsbP} = \br(\tsbX, \rdy{\tsbP},S)$ for
      some set $S$.  Since $\rdy{\tsbP}$ is $\guardTrue$ for external
      choices, then there exists a location $\tsbX$ with invariant
      $\Inv[1](\tsbX) = \rdy{\tsbP} = \guardTrue$ and
      $\tsbTauY \not\in \UrgLoc[1]$ .
      In this case rule [TA3] of~\Cref{net_sem} states that:
      \[
      \irule{
        ((\clockN  \sqcup \clockE)+ \delta) \in \sem{ \guardTrue \land \Inv[2](\tsbTauY)} \qquad \tsbX \not\in \UrgLoc[1] \qquad  \tsbTauY \not\in \UrgLoc[2]
      }
      {( \tsbX, \tsbTauY,\clockN \sqcup \clockE) \umove{\quad \delta \quad } (\tsbX, \tsbTauY, (\clockN \sqcup \clockE) + \delta)}
      \]

      By \nrule{\idrule}\ rule, we have  $\clockE +
      \delta \in \rdy{\tsbQ}$, it follows $(\clockN \sqcup \clockE) +
      \delta \in \rdy{\tsbQ}$. We already noted  that $\tsbTauY$ and
      $ \tsbX$ are not urgent, so we conclude  $( \tsbX, y, \clockN \sqcup
      \clockE) \umove{\delta} ( \tsbX, y, \clockN + \delta \sqcup
      \clockE + \delta)$. The resulting states belong to $\relR$.
      
      % (ii)
      In the second case, according with~\Cref{mapping}, the mapping
      for the success termination is: $\sem{ \tsbX \triangleq \success} =
      \idle(\tsbX)$.  
      By definition of the \idle\ pattern
      and~\Cref{lem:mapping:inv}, the location $\tsbX$ is not urgent,
      with $\Inv[1](\tsbX) = \guardTrue$.  Again
      $\tsbTauY \not\in \UrgLoc[1]$ .
      In this case rule [TA3] of~\Cref{net_sem} states that:
      \[
      \irule{
        ((\clockN  \sqcup \clockE)+ \delta) \in \sem{ \guardTrue \land \Inv[2](\tsbTauY)} \qquad \tsbX \not\in \UrgLoc[1] \qquad  \tsbTauY \not\in \UrgLoc[2]
      }
      {( \tsbX, \tsbTauY,\clockN \sqcup \clockE) \umove{\quad \delta \quad } (\tsbX, \tsbTauY, (\clockN \sqcup \clockE) + \delta)}
      \]

      By \nrule{idrule} rule, we have  $\clockE +
      \delta \in \rdy{\tsbQ}$, it follows $(\clockN \sqcup \clockE) +
      \delta \in \rdy{\tsbQ}$. We already noted  that $\tsbTauY$ and
      $ \tsbX$ are not urgent, so we conclude  $( \tsbX, y, \clockN \sqcup
      \clockE) \umove{\delta} ( \tsbX, y, \clockN + \delta \sqcup
      \clockE + \delta)$. The resulting states belong to $\relR$.

    \item \nrule{\edrule}\ applied twice.  Similar to  previous cases.

    \end{itemize}
    \vspace{10pt}

  \item[\nrule{\strule}]
    \[
    \irule
    {(x,\clockN) \demove{\; \atomOut{a} \; } (x',\clockNi) \quad 
      (y,\clockE) \demove{\; \atomIn{a} \; } (y',\clockEi)}
    {(x,\clockN) \mid (y,\clockE) \demove{\; \atom{a} \;} 
      (x',\clockNi) \mid (y',\clockEi)} 
    \] 
    According to this rule, we can synchronize on $ \atom{a}$ only if  
    \[
    \irule{
      ({\todo[\TsumI{\atomOut{a}}{\guardG,\resetR}{}]{\tsbXi}}, \; \clockN) \;\demove{\atomOut{a}}\; (\tsbXi,\; \reset{\clockN}{\resetR})  \quad
      \irule{(\tsbY \eqdef \TsumE{\atomIn{a}}{\guardF,\resetT}{\tsbYi} + \tsbQ) \in D \quad \clockE \in \sem{\guardF}  }
      {(\tsbY, \; \clockN) \;\demove{\atomIn{a}}\; (\tsbYi, \; \reset{\clockE}{\resetT}) } 
    }
    {(\todo[\TsumI{\atomOut{a}}{\guardG,\resetR}{}]{\tsbXi},\clockN) \mid (\tsbY,\clockE) \demove{\; \atom{a} \;}   (\tsbXi,\reset{\clockN}{\resetR}) \mid (\tsbYi,\reset{\clockE}{\resetT})}
    \]
    Let $x = \todo[\atomOut{a} \setenum{\guardG, \resetR}]{\tsbXi}$.
    By~\Cref{mapping}, the last part of the mapping for an internal choice is  such as
    $\pfx(x, \atomOut{a}, \resetR, \idle(\tsbXi))$.
    By definition of the \pfx\ pattern
    and~\Cref{lem:mapping:inv}, there exists an edge
    $(x, \atomOut{a}, \guardTrue, \resetR, \tsbXi) \in \Edg[1]$ and
    $\Inv[1](\tsbXi) = \guardTrue$.
    By~\Cref{mapping}, the mapping for an external choice
    is:
    $\sem{ \tsbY \triangleq \TsumE{\atomIn{a}}{\guardF,\resetT}{\tsbYi}
      + \tsbQ } = \br(\tsbY, \guardTrue, \setenum{ (\atomIn{a},
      \guardF, \resetT, \idle(\tsbYi))} \cup S)$
    for some $S$.  
    By definition of the \br\ pattern and~\Cref{lem:mapping:inv}, 
    there exists an edge
    $(\tsbY, \atomIn{a}, \guardF, \resetT, \tsbYi)$ and
    $\Inv[1](\tsbYi) = \guardTrue$.
    In this case rule [TA2] of~\Cref{net_sem} states that:
    \[
    \irule{  \begin{array}{c}
        (x, \atomOut{a}, \guardTrue, \resetR, \tsbXi) \in \Edg[1] \quad   (\clockN \sqcup \clockE) \in \sem{\guardTrue  } 
        \\
        (\tsbY, \atomIn{a}, \guardF, \resetT, \tsbYi) \in \Edg[2] \quad  (\clockN \sqcup \clockE) \in \sem{\guardF  } 
      \end{array}
      \quad   \resett{ (\clockN \sqcup \clockE)}{\resetR}{\resetT} \in \sem{\guardTrue \land \guardTrue }
    }
    {(x,\tsbY,\clockN \sqcup \clockE) \umove{\quad \atom{a} \quad } ( \tsbXi,\tsbYi,\resett{(\clockN \sqcup \clockE)}{\resetR}{\resetT})} 
    \]
    Since by hypothesis $ \clockE \in \sem{\guardF}$, since $\ck{D'}
    \cap \ck{D''} = \emptyset$ and we obtain 
    $(\clockE \sqcup \clockN) \in \sem{\guardF}$.
    Hence by rule [TA2] of~\Cref{net_sem} , we have $(x,\tsbY, \clockN \sqcup \clockE)
    \umove{\atom{a}} (\tsbXi, \tsbYi,\resett{\clockN \sqcup \clockE}{\resetR}{\resetT}$. 
    The resulting states, belong to $\relR$.

  \end{description}

\noindent  We now show that every possible move of $(x,y,\clockN \sqcup \clockE)$ in $\umove{}$ is matched by
  a move of $(x,\clockN)|(y,\clockE)$ in $\demove{}$, and that
  the resulting states are related by~$\relR$.
  We have the following cases, according to the rules of~\Cref{net_sem} used to justify the move:
  \begin{itemize}

  \item \nrule{[TA1]}.
    \[ 
    \irule{(x, \tau, \guardG, \resetR, x') \in \Edg[1]
      \qquad (\clockN \sqcup \clockE) \in \sem{\guardG} \qquad
      \reset{(\clockN \sqcup \clockE)}{\resetR} \in \sem{\Inv[1](x')
        \land \Inv[2](y)} } { (x, y,\clockN \sqcup \clockE) \umove{\quad
        \tau \quad } (x', y, \reset{\clockN \sqcup \clockE}{\resetR}) }
    \] 
    By~\Cref{mapping}, there exist two possibilities for
    an edge in network $N$ to display a $\tau$ label, and both derive
    from the mapping of an internal choice.  Let $\tsbP =
    \TsumI{\atomOut{a}}{\guardG,\resetR}{}{\tsbY} \oplus \tsbPi$, then
    $ \sem{ \tsbX \triangleq \tsbP} = \pfx(\tsbX, \tau, \emptyset,
    \br(\tsbTauX , \rdy{\tsbP}, \setenum{(\tau, \guardG, \emptyset,
      \autA)} \cup S))$
    for some set $S$ and \TA $\autA$.  
    By definition of the \pfx\ and \br\ patterns and by~\Cref{lem:mapping:inv}, 
    we have two edges with $\tau$ label in
    $\Edg[1]$: 
    (i) $(\tsbX, \tau, \guardTrue, \emptyset, \tsbTauX)$
    with $\Inv[1]({\tsbTauX}) = \rdy{\tsbP}$; 
    and (ii)
    $(\tsbTauX, \tau, \guardG, \emptyset, \tsbY)$ with
    $\Inv[1](x') = \guardTrue$.  So we have two sub-cases:
    \begin{enumerate}[label=(\roman*)]

    \item Assume that we fire the first edge $(\tsbX, \tau,
      \guardTrue, \emptyset, \tsbTauX)$. Hence $x = \tsbX$, $x' = \tau
      \tsbX$ and $\resetR = \emptyset$.  By hypothesis $\clockN \sqcup
      \clockE \in \sem{\Inv[1](\tsbTauX)} =
      \rdy{\tsbP}$. Since $\ck{D'} \cap
      \ck{D''} = \emptyset$, we derive $\clockN \in \rdy{\tsbP}$.
      Hence, we can use \nrule{\erule}\ and \nrule{\serule}\ rule from~\Cref{demove} to
      obtain:
      \[
      \irule{
        \irule{ (\tsbX \eqdef \tsbP)  \in D \quad  p = \oplus \ldots \quad  \clockN \in \rdy{p}  }
        {  (\tsbX,\; \clockN) \;\demove{\tau}\; (\tsbTauX,\; \clockN) }
      }{
        (\tsbX,\clockN) \mid (y,\clockE) \demove {\; \tau  \;}(\tsbTauX,\clockN) \mid (y,\clockE)
      }
      \]
      Since by \nrule{[TA1]},  $ \reset{(\clockN \sqcup \clockE)}{\emptyset} \in \sem{\Inv[1](\tsbTauX)
        \land \Inv[2](y)}$,  the resulting states, belong to $\relR$.
      
      \vspace{10pt}
    \item Assume that we fire the second edge $(\tsbTauX, \tau,
      \guardG, \emptyset, x')$.  Hence $x = \tsbTauX$, $x' = \tsbY$
      and $\resetR = \emptyset$.  By \nrule{[TA1]}, $\clockN \sqcup \clockE \in
      \sem{\guardG}$, since $\ck{D'} \cap
      \ck{D''} = \emptyset$, we derive $\clockN \in
      \sem{\guardG}$. Hence, we can use \nrule{\crule}\ and \nrule{\serule}\ rules
      from~\Cref{demove} to obtain:
      \[
      \irule{
        \irule{ (\tsbX \eqdef \TsumI{\atomOut{a}}{\guardG,\resetR}{}{\tsbY} \oplus \tsbPi)  \in D  
          \quad  \clockN \in \sem{\guardG} }
        {  (\tsbTauX,\; \clockN) \;\demove{\tau}\; ( \todo[\TsumI{\atomOut{a}}{\guardG,\resetR}{}]{\tsbY},\; \clockN) }
      }{
        (\tsbTauX,\clockN) \mid (y,\clockE) \demove {\; \tau  \;}( \todo[\TsumI{\atomOut{a}}{\guardG,\resetR}{}]{\tsbY},\clockN) \mid (y,\clockE)
      }
      \]
      By hypothesis  
      $\reset{(\clockN \sqcup \clockE)}{\emptyset} \in \sem{\Inv[1](x') \land \Inv[2](y)}$,  
      so the resulting states belong to $\relR$.

    \end{enumerate}

  \item \nrule{[TA2]}.
    \[
    \irule{  
      \begin{array}{c}
        (x, \atomOut{a}, \guardG, \resetR, x') \in \Edg[1]  \quad  
        (\clockN \sqcup \clockE) \in \sem{\guardG} \quad 
        \resett{ (\clockN \sqcup \clockE)}{\resetR}{\resetT} \in \sem{\Inv[1](x)}
        \\
        (y, \atomIn{a}, \guardF, \resetT, y') \in \Edg[2]   \quad  
        (\clockN \sqcup \clockE) \in \sem{\guardF} \quad 
        \resett{ (\clockN \sqcup \clockE)}{\resetR}{\resetT} \in \sem{\Inv[2](y)}
      \end{array}
    }
    { (x,y,\clockN  \sqcup \clockE) \umove{\quad \atom{a} \quad } (x', y',\resett{(\clockN \sqcup \clockE)}{\resetR}{\resetT})} 
    \]
    By~\Cref{mapping}, a synchronization happens only if an
    internal-external choice synchronizes.
    Hence,  
    $x = \todo[\TsumI{\atomOut{a}}{\guardG,\resetR}{}]{\tsbX}$ and 
    $y = \tsbY$ for some $\tsbY \eqdef \atomIn{a}\setenum{\guardF,\resetS}.{\tsbYi} + \tsbP$.  
    By~\Cref{mapping}, the last part of the mapping for an
    internal choice is such as $\pfx(x, \atomOut{a}, \resetR,
    \idle(\tsbXi))$.
    By definition of the \pfx\ pattern
    and by~\Cref{lem:mapping:inv}, 
    there exists an edge $(x, \atomOut{a}, \guardTrue, \resetR, \tsbXi) \in
    \Edg[1]$ and $\Inv[1](\tsbXi) = \guardTrue$.
    By~\Cref{mapping}, the mapping for an external choice is:
    $\sem{ \tsbY \triangleq  \TsumE{\atomIn{a}}{\guardF,\resetT}{\tsbYi} + \tsbQ }  =   
    \br(\tsbY, \guardTrue, \setenum{ (\atomIn{a}, \guardF, \resetT, \idle(\tsbYi))} \cup S )$ for some $S$.
    By definition of the \br\ pattern, there
    exists an edge   
    $(\tsbY, \atomIn{a}, \guardF, \resetT, \tsbYi)$ 
    and $\Inv[1](\tsbYi) = \guardTrue$.
    So in this case \nrule{[TA2]} becomes:
    \[
    \irule{
      \begin{array}{c}
        ( x, \atomOut{a}, \guardTrue, \resetR, \tsbX') \in \Edg[1]  \quad  
        (\clockN \sqcup \clockE) \in \sem{\guardTrue} \quad 
        \resett{ (\clockN \sqcup \clockE)}{\resetR}{\resetT} \in \sem{\Inv[1](\tsbX')}
        \\
        (\tsbY, \atomIn{a}, \guardF, \resetT, \tsbY') \in \Edg[2]   \quad  
        (\clockN \sqcup \clockE) \in \sem{\guardF} \quad 
        \resett{ (\clockN \sqcup \clockE)}{\resetR}{\resetT} \in \sem{\Inv[2](\tsbY')}
      \end{array}
    }
    { ( x, \tsbY,\clockN  \sqcup \clockE) \umove{\quad \atom{a} \quad } (\tsbX', \tsbY',\resett{(\clockN \sqcup \clockE)}{\resetR}{\resetT})} 
    \] 
    Since $\ck{D'} \cap \ck{D''} = \emptyset$ and since $(\clockN
    \sqcup \clockE) \in \sem{\guardF}$, we
    derive $\clockE \in \sem{\guardF}$.
    Hence we can apply \inrule \nrule{\outrule} and \nrule{\strule} to obtain:
    \[
    \irule{
      ({\todo[\TsumI{\atomOut{a}}{\guardG,\resetR}{}]{\tsbXi}}, \; \clockN) \;\demove{\atomOut{a}}\; (\tsbXi,\; \reset{\clockN}{\resetR})  \quad
      \irule{(\tsbY \eqdef \TsumE{\atomIn{a}}{\guardF,\resetT}{\tsbYi} + \tsbQ) \in D \quad \clockE \in \sem{\guardF}  }
      {(\tsbY, \; \clockN) \;\demove{\atomIn{a}}\; (\tsbYi, \; \reset{\clockE}{\resetT}) } 
    }
    {(\todo[\TsumI{\atomOut{a}}{\guardG,\resetR}{}]{\tsbXi},\clockN) \mid (\tsbY,\clockE) \demove{\; \atom{a} \;}   (\tsbXi,\reset{\clockN}{\resetR}) \mid (\tsbYi,\reset{\clockE}{\resetT})}
    \] 
    By hypothesis 
    $\resett{(\clockN \sqcup \clockE)}{\resetR}{\resetT} \in \sem{\Inv[1](\tsbX) \land \Inv[2](\tsbY)}$,  
    so the resulting states belong to $\relR$.

  \item \nrule{[TA3]}. 
    \[
    \irule{ ((\clockN \sqcup \clockE) + \delta) \in \sem{\Inv[1](x)}  \quad  x \not\in \UrgLoc[1] \quad  
      ((\clockN \sqcup \clockE) + \delta) \in \sem{\Inv[2](y)}  \quad   y \not\in \UrgLoc[2] } 
    {(x,y,\clockN \sqcup \clockE) \umove{\quad \delta \quad } (x,y,(\clockN \sqcup \clockE) + \delta)}
    \]
    By \nrule{[TA3]}, time can pass in a network only if
    all the locations of the current state are not urgent. According
    with~\Cref{mapping}, this happens for: the second location in
    the mapping of an internal choice or the first location in the
    mapping of an external choice or a success location. Hence, we
    have several sub-cases:

    \begin{enumerate}[label=(\roman*)]

    \item Assume that $x$ derives from an internal choice and
      $y$ from an external choice.  
      By~\Cref{mapping} and~\Cref{lem:mapping:inv}, $x$ must be of the form
      $x = \Inv[1](\tsbTauX)$ for some $\tsbX \eqdef \tsbP \in D$ with
      $\tsbP = \oplus \ldots$, obtaining
      $\Inv[1](\tsbTauX) = \rdy{\tsbP}$ and
      $\tsbTauX \not\in \UrgLoc[1]$.
      By~\Cref{mapping} and~\Cref{lem:mapping:inv}, $y$ must be of the form $y = \tsbY$
      for some $\tsbY \eqdef \tsbQ \in D$ with $\tsbQ = + \ldots$,
      obtaining $\Inv[2](\tsbY) = \rdy{\tsbQ} = \guardTrue$ and $\tau
      \tsbY \not\in \UrgLoc[2]$.
      So in this case \nrule{[TA2}] becomes: 
      \[
      \irule{((\clockN \sqcup \clockE) + \delta) \in \sem{\Inv[1](\tsbTauX)} \quad   (\tsbTauX) \not\in \UrgLoc[1] \quad 
        ((\clockN \sqcup \clockE) + \delta) \in \sem{\Inv[2](\tsbY)} \quad   (\tsbY) \not\in \UrgLoc[2]  } 
      {(\locL[1],\locL[2],\clockN \sqcup \clockE) \umove{\quad \delta \quad } (\locL[1],\locL[2],(\clockN \sqcup \clockE) + \delta)}
      \]
      By hypothesis of \nrule{[TA2]}, $((\clockN \sqcup \clockE) + \delta) \in
      \sem{\Inv[1](\tsbTauX)} = \rdy{\tsbP}$.  Since $\ck{D'} \cap
      \ck{D''} = \emptyset$, we derive $\clockN + \delta \in
      \rdy{\tsbP}$.  Hence, we have:
      \[
      \irule{  
        \irule{ \tsbX \eqdef \tsbP \in D \quad \clockN + \delta \in \rdy{\tsbP}} 
        {(\tsbTauX,\; \clockN)\demove{\; \delta \; }(\tsbTauX,\; \clockN+\delta)} \quad
        \irule{ (\tsbY = + \ldots)\in D  \lor  (\tsbY = \success)\in D } 
        {( \tsbY,\; \clockN)\demove{\; \delta \; }(\tsbY,\; \clockN+\delta)} 
      }
      { (\tsbTauX,\clockN) \mid (\tsbY,\clockE) \demove{\; \delta \;} (\tsbTauX,\clockN+\delta)
        \mid (\tsbY,\clockE+\delta)
      } 
      \]
      Since by hypothesis,  $ ((\clockN \sqcup \clockE)+\delta) \in \sem{\Inv[1](\tsbTauX)
        \land \Inv[2](\tsbY)}$,  the resulting states, belong to $\relR$.
      
    \item All the other combinations can be proved  similarly to the previous one. 
      
    \end{enumerate}
  \end{itemize}
\end{proofof}

\subsection{Implementation  details in \UPP} 
\label{sec:upp_guards}

\UPP does not allow disjunctions in guards and invariants, because of efficiency issues.
However, with little adjustments it is possible to recover disjunctions.
Disjunctions in guards on edges can be achieved by creating  several
copies of the edge, one for each disjunct. 
E.g., the encoding of 
$\atomOut{a}\setenum{\clockT < 4 \lor \clockT > 10}$ may give an \TA with
edges $(l_0,\atomOut{a}, \clockT < 4, \emptyset, l_1)$ and 
$(l_0,\atomOut{a}, \clockT > 10, \emptyset, l_1) $ for some $l_0$ and $l_1$.

An invariant with disjunctions invariant can  be split into its disjuncts, 
and the location can be multiplied  accordingly. 
In~\Cref{fig:upp_invariant} we see the actual implementation of~\Cref{ex:mapping}: 
location names are different for what seen in the theory, since they are immaterial when checking deadlock-freedom. 
The invariant of location $\tsbTauX$ is $\clockT < 7 \lor \clockFmt{c} \leq 2$, 
so we split it in two clauses $\clockT < 7$ and $\clockFmt{c} \leq 2$,
and each clause is assigned to a new location ($XX0aXX01$ and $XX0bXX21$).
The resulting network
is no longer bisimilar to the previous one:
\eg, in the original network time can delay until one of the two
guards is still true, while in the one in~\Cref{fig:upp_invariant} it
can delay, non-deterministically, sometimes until $2$ and sometimes
until $7$ time unit.  
Nevertheless, both networks are equivalent respect to deadlock-freedom.

Finally, since \UPP supports \emph{urgent} locations, no other
adjustments are needed.

\begin{figure*}[t]
   {\includegraphics[width=0.5\textwidth]{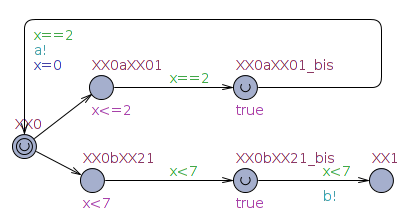}} 
   {\includegraphics[width=0.35\textwidth]{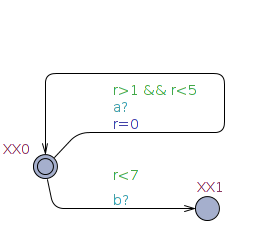}}
\caption{\UPP mapping of the TSTs in~\Cref{ex:mapping}.}
\label{fig:upp_invariant}
\end{figure*}
\newpage

\section{Proofs for~\texorpdfstring{\Cref{sec:tst-monitoring}}{Section 9}}
\label{sec:proofs-monitoring}

\begin{proofof}{lem:st:turn-bisimilar}
  Let us consider the LTS $\smove{}$ of~\Cref{fig:tst:s_semantics} 
  and the set of success states
  $S_1 = \setcomp{(\success,\clockN)\mid (\tsbQ,\clockE)}{\tsbQ \text{ TST}, \clockN, \clockE \in \Val}$. %
  Also, consider the LTS $\mcmove{}$ of~\Cref{fig:tst:turn-based} and 
  the set of success states 
  $S_2 = \setcomp{(\success, \queue{}, \clockN) \mmid (\tsbQ, \queue{}, \clockE)}{\tsbQ \text{ TST}, \clockN, \clockE \in \Val}$. %
  Let $\relR$ be the relation:
  \begin{align*}
    \relR & \; = \; 
    \relR[1] \, \cup \, \relR[2] \,\cup\, \relR[3] \,\cup\; \relR[2S] \cup \relR[3S]
    \\
    \relR[1]
    & \; = \; 
    \setcomp
    {( (\tsbP,\clockN) \mid (\tsbQ, \clockE), \;\; (\tsbP, \queue{}, \clockN) \mmid  (\tsbQ,  \queue{}, \clockE))}
    {\tsbP,\tsbQ \text{ TST}, \clockN,\clockE \in \Val}
    \\
    \relR[2] & \; = \;  
    \{ ( (\todo[\atomOut{a}\setenum{\guardG,\resetR}]\tsbP, \clockN) \mid( \tsbQ, \clockE), \;\;
    (\tsbP,  \queue{\atomOut{a}}, \reset{\clockN}{\resetR}) \mmid (\tsbQ, \queue{}, \clockE)  ) 
    \; | \; (\tsbP, \tsbQ \text{ TST}, \clockN \in \sem{\guardG} \}
    \\
    \relR[3] & \; = \;  
    \{ ( (\todo[ \atomOut{a}\setenum{\guardG,\resetR}]\tsbP,\clockN) \mid (\todo[ \atomOut{b}\setenum{\guardF, \resetS}]\tsbQ , \clockE),  \;\;
    (\tsbP ,  \queue{\atomOut{a}}, \reset{\clockN}{\resetR}) \mmid  (\atomOut{b}\setenum{\guardF, \resetS}.\tsbQ \oplus \tsbQi, \queue{}, \clockE))
    \; | \; \cdots \} 
    \\
    \relR[2S] & \; = \;  
    \{ ( (\tsbP, \clockN) \mid  (\todo[\atomOut{a}\setenum{\guardF,\resetS}]\tsbQ, \clockE), \;\;
     (\tsbP, \queue{}, \clockN) \mmid (\tsbQ,  \queue{\atomOut{a}}, \reset{\clockN}{\resetS})  ) 
    \; | \; \tsbP, \tsbQ \text{ TST}, \clockE \in \sem{\guardF} \}
    \\
    \relR[3S] & \; = \;  
    \{ ( (\todo[ \atomOut{a}\setenum{\guardG,\resetR}]\tsbP,\clockN) \mid (\todo[ \atomOut{b}\setenum{\guardF, \resetS}]\tsbQ , \clockE),  \;\;
    ( \atomOut{a}\setenum{\guardG, \resetR}.\tsbP \oplus \tsbPi,  \queue{}, \clockN) \mmid  (\tsbQ, \queue{\atomOut{b}}, \reset{\clockE}{\resetR})
    \; | \; \cdots \}
  \end{align*}
  Let $s_1 = (\tsbP,\clockN) \mid (\tsbQ,\clockE)$,
  and let $s_2 = (\tsbP, \queue{}, \clockN) \mmid (\tsbQ,  \queue{}, \clockE)$.
  Clearly, $s_1 \relR s_2$,
  hence to obtain the thesis we will prove that $\relR$ is a turn-bisimulation.
  The proof is organised as follows.
  In {\bf Part A} we show that $s_2$ turn-simulates $s_1$ via $\relR$,
  and in {\bf Part B} that $s_1$ turn-simulates $s_2$ via $\relR$.
  Within each part, we proceed by cases on the form of $s_1$ and $s_2$:
  in {\bf Case 1} we assume that $(s_1,s_2) \in \relR[1]$,
  in {\bf Case 2} that $(s_1,s_2) \in \relR[2]$, and 
  in {\bf Case 3} that $(s_1,s_2) \in \relR[3]$.
  We omit cases for $\relR[2S]$ and $\relR[3S]$, since they are specular to  cases $\relR[2]$  and $\relR[3]$.
  For each case, we show that 
  items~\eqref{item:turn-simulation:a}, \eqref{item:turn-simulation:b}, and~\eqref{item:turn-simulation:c} 
  of~\Cref{def:turn-simulation} hold. %
  We will only consider the moves of the LHS of a composition $P \circ Q$;
  all the symmetric cases will be omitted.

  \medskip
  \noindent {\bf Part A:} $s_2$ turn-simulates $s_1$ via $\relR$.

  \smallskip
  \noindent {\bf Case 1}: 
  Let $s_1 = (\tsbP,\clockN) \mid (\tsbQ,\clockE)$ and 
  $s_2 = (\tsbP, \queue{}, \clockN) \mmid  (\tsbQ,  \queue{}, \clockE)$.
  \begin{enumerate}[label=\({\alph*}]

  \item To prove item~\eqref{item:turn-simulation:a} of~\Cref{def:turn-simulation}, 
    we consider the possible moves of $s_1$:
    \begin{itemize}[label=$-$]

    \item \nrule{[S-$\sumInt$]}. %
      We have:
      \[
      \irule
      {(\tsbP,\clockN) \smove{\tau} (\todo[\atomOut{a}\setenum{\guardG,\resetR}]\tsbPi, \clockN)}
      {s_1 \smove{\tau} (\todo[\atomOut{a}\setenum{\guardG,\resetR}]\tsbPi,\clockN) \mid (\tsbQ, \clockE) = s_1'}
      \]      
      where the premise requires 
      $\tsbP = \atomOut{a}\setenum{\guardG,\resetR}.\tsbPi \sumInt \tsbPii$ and $\clockN \in \sem{\guardG}$. %
      Hence, by rule~\nrule{[M-$\sumInt$]} we have:
      \[
      s_2 
      \cmove{\pmvA \says\atomOut{a}} 
      (\tsbPi, \queue{\atomOut{a}}, \reset{\clockN}{\resetR}) \mmid (\tsbQ, \queue{}, \clockE)
      = s_2'
      \]
      Then,
      $(s_1',s_2') \in \relR[2] \subseteq \relR$. 

    \item \nrule{[S-$\tau$]}. %
      This case does not apply.

    \item \nrule{[S-Del]}. 
      We have:
      \[
      \irule
      {(\tsbP, \clockN) \smove{\delta} (\tsbPi,\clockN+\delta) \qquad
       (\tsbQ, \clockE) \smove{\delta} (\tsbQi,\clockE+\delta)}
      {s_1 \smove{\delta} (\tsbPi, \clockN+\delta) \mid (\tsbQi,\clockE+\delta) = s_1'}
      \]
      The only rule which can be used in the premises of the above
      is~\nrule{[Del]}, which implies that
      $\tsbP = \tsbPi$, $\tsbQ = \tsbQi$,
      $\clockN + \delta \in \rdy{\tsbP}$ and 
      $\clockE + \delta \in \rdy{\tsbQ}$.
      Then, the thesis follows by rule~\nrule{[M-Del]}.

    \end{itemize}

  \item To prove item~\eqref{item:turn-simulation:b} of~\Cref{def:turn-simulation}, 
    we consider the possible moves of $s_2$:
    \begin{itemize}[label=$-$]

    \item \nrule{[M-$\sumInt$]}.
      We have
      $\tsbP = \atomOut{a}\setenum{\guardG,\resetR}.\tsbPi \sumInt \tsbPii$,
      $\clockN \in \sem{\guardG}$, and:
      \[
      s_2 
      \cmove{\pmvA \says\atomOut{a}}
      (\tsbPi, \queue{\atomOut{a}}, \reset{\clockN}{\resetR}) \mmid (\tsbQ,\queue{},\clockE)
      = s_2'
      \]
      So, by rules~\nrule{[$\sumInt$]} and~\nrule{[S-$\sumInt$]} we have:
      \[
      \irule
      {\tsbP \smove{\tau} (\todo[\atomOut{a}\setenum{\guardG,\resetR}] \tsbPi, \clockN)}
      {s_1 \smove{\tau} (\todo[\atomOut{a}\setenum{\guardG,\resetR}] \tsbPi, \clockN) \mid (\tsbQ, \clockE)}
      = s_1'
      \]
      and we conclude that $(s_1',s_2') \in \relR[2] \subseteq \relR$.

    \item \nrule{[M-$\sumExt$]}.
      This case does not apply, since both buffers are empty.

    \item \nrule{[M-Del]}.
      We have
      $\clockN + \delta \in \rdy{\tsbP}$, 
      $\clockE + \delta \in \rdy{\tsbQ}$, and:
      \[
      s_2 
      \cmove{\delta}
      (\tsbP, \queue{}, \clockN+\delta) \mmid (\tsbQ,\queue{},\clockE+\delta)
      = s_2'
      \]
      Then, rule~\nrule{[Del]} yields 
      $(\tsbP,\clockN) \smove{\delta} (\tsbP, \clockN+\delta)$
      and
      $(\tsbQ,\clockE) \smove{\delta} (\tsbP, \clockE+\delta)$.
      Hence, by rule~\nrule{[S-Del]} we conclude that:
      \[
      s_1 \smove{\delta} (\tsbP, \clockN+\delta) \mid (\tsbQ, \clockE+\delta) = s_1'
      \]
      and the thesis follows because $(s_1',s_2') \in \relR[1] \subseteq \relR$ .  %

    \end{itemize}

  \item
    To prove item~\eqref{item:turn-simulation:c} of~\Cref{def:turn-simulation}, 
    assume that $s_2 \in S_2$.
    By definition of $S_2$, $s_2$ has the form
    $(\success,\queue{},\clockN) \mmid (\tsbQ, \queue{}, \clockE)$.
    Then, $s_1 = (\success, \clockN) \mid (\tsbQ, \clockE) \in S_1$. 

  \end{enumerate}

  \noindent {\bf Case 2}: 
    Let $s_1 = (\todo[ \atomOut{a}\setenum{\guardG,\resetR}]\tsbP, \clockN) \mid (\tsbQ, \clockE)$ and 
    $s_2 = (\tsbP, \queue{\atomOut{a}}, \reset{\resetR}{\clockN}) \mmid
    (\tsbQ, \queue{}, \clockE)$ with  $\clockN \in \sem{\guardG}$.
 
  \begin{enumerate}[label=\({\alph*}]
   \item
    To prove item~\eqref{item:turn-simulation:a} of~\Cref{def:turn-simulation}, 
    we consider the possible moves of $s_1$:
    \begin{itemize}[label=$-$]
    \item \nrule{[S-$\sumInt$]}. We have: 
      \[
      \irule
      {(\tsbQ,\clockE) \smove{\tau} (\todo[\atomOut{b}\setenum{\guardF,\resetS}]\tsbQi, \reset{\resetS}{\clockE)}}
      {s_1 \smove{\tau} 
        (\todo[\atomOut{a}\setenum{\guardG,\resetR}]\tsbP,\clockN)
        \mid  (\todo[\atomOut{b}\setenum{\guardF, \resetS}]\tsbQ',  \reset{\resetS}{\clockE}) = s_1'}
      \]      
      where the premise requires 
      $\tsbQ = \atomOut{a}\setenum{\guardF,\resetS}.\tsbQi \sumInt \tsbQii$ and $\clockE \in \sem{\guardF}$. %
      Hence, by rule~\nrule{[M-$\sumInt$]} we have:
      \[
      s_2 
      \cmove{\pmvB \says\atomOut{b}} 
      (\tsbPi, b, \clockN) \mmid (\tsbQi, \queue{\atomOut{b}}, \reset{\clockE}{\resetS}) = s_2'
      \]
      Then,
      $(s_1',s_2') \in \relR[2] \subseteq \relR$. 

    \item \nrule{[S-$\tau$]}. We have: 
      \[
       \irule
       { (\todo[ \atomOut{a}\setenum{\guardG,\resetR}]\tsbP, \clockN) \smove{\atomOut{a}} (\tsbP,
          \reset{\clockN}{\resetR}) \quad 
        (\tsbQ, \clockE) \smove{\atomIn{a}} (\tsbQi,\reset{\clockE}{\resetS}) 
       }
       {(\tsbP,\clockN) \mid (\tsbQ,\clockE) \smove{\; \tau \;} 
            (\tsbPi,\clockNi) \mid (\tsbQi,\clockEi) = s_1'}
      \]
      where the premise requires 
      $\tsbQ = \atomIn{a}\setenum{\guardF, \resetS}.\tsbQi + \tsbQii$,
      with $\clockN \in \sem{\guardG}$ and
      $\clockE \in \sem{\guardF}$ . %
      Since $\clockE \in \sem{\guardF}$, by rule~\nrule{[M-$\sumExt $]}
      we have:
      \[
         { (\tsbP, \queue{\atomOut{a}}, \clockN) \mmid 
           (\TsumE{\atomIn{a}}{\guardG,\resetR}{\tsbQi} \sumExt \tsbQii}, \queue{}, \clockE)
           \cmove{\pmv B \says \atomIn{a}}
         (\tsbP, \queue{}, \clockN)  \mmid  (\tsbQi, \queue{},
          \reset{\clockE}{\resetR}) = s_2'
      \]
      Then,
      $(s_1',s_2') \in \relR[2] \subseteq \relR$.

    \item \nrule{[S-Del]}. 
      This case does not apply, since one
      buffers is not  empty.
    \end{itemize}

  \item
    To prove item~\eqref{item:turn-simulation:b} of~\Cref{def:turn-simulation}, 
    we consider the possible moves of $s_2$:
    \begin{itemize}[label=$-$]

    \item \nrule{[M-$\sumInt$]}.
      This case does not apply, given the form of $s_2$.

    \item \nrule{[M-$\sumExt$]}.
      We have:
      \[
      s_2 
      \cmove{\pmv{B} \says\atomIn{a}}  
      (\tsbP, \queue{},\reset{\clockN}{\resetR}) \mmid (\tsbQi,\queue{}, \reset{\clockE}{\resetS})
      = s_2'
      \]
      which requires
      $\tsbQ = \atomIn{a}\setenum{\guardF, \resetS}.\tsbQi + \tsbQii$, with $ \clockE \in \sem{\guardF}$. %
      Since by hypothesis $\clockN \in \sem{\guardG}$, 
      by rule~\nrule{[S-$\sumInt$]} we have:
      \[
      \irule
      {\irule{}{(\todo[\atomOut{a}\setenum{\guardG,\resetR}]\tsbP, \clockN) \smove{\atomOut{a}} (\tsbP,\reset{\clockN}{\resetR}) } \;\nrule{[$\sumInt$]} 
        \quad
       \irule{}{(\atomIn{a}\setenum{\guardF, \resetS}.\tsbQi + \tsbQii, \clockE) \smove{\atomIn{a}} (\tsbQi,\reset{\clockE}{\resetS}) } \;\nrule{[$\sumExt$]}}
      {s_1 \smove{\tau} (\tsbP,\reset{\clockN}{\resetR})  \mid (\tsbQi, \reset{\clockE}{\resetS}) = s_1'}
      \] 
      and the thesis follows because $(s_1',s_2') \in \relR[1] \subseteq \relR$. %

    \item \nrule{[M-Del]}.
      This case does not apply, given the form of $s_2$.  

    \end{itemize}

  \item
    To prove item~\eqref{item:turn-simulation:c} of~\Cref{def:turn-simulation}, 
    assume that $s_2 \in S_2$.
    By definition of $S_2$, $s_2$ has the form
    $(\success,\queue{},\clockN) \mmid (\tsbQ, \queue{}, \clockE)$.
    Then, this case does not apply.

  \end{enumerate}

  \noindent {\bf Case 3}: 
  Let $s_1 = (\todo[\atomOut{a}\setenum{\guardG,\resetR}] \tsbP \mid \todo[\atomOut{b}\setenum{\guardF, \resetS}] \tsbQ, \clockN)$, 
  and let 
  $s_2 = (\tsbP, \queue{\atomOut{a}}, \reset{\clockN}{\resetR}) \mmid$ 
  $(\atomOut{b}\setenum{\guardF, \resetS}.\tsbQ \oplus \tsbQi,\queue{},\clockE)$.
  The thesis follows trivially, since both $s_1$ and $s_2$ are stuck, and neither of them is a success state.
  
  \medskip
  \noindent {\bf Part B:} $s_1$ turn-simulates $s_2$ via $\relR$.
    {

  \noindent{\bf Case 1}: 
  Let $s_1 = (\tsbP,\clockN) \mid (\tsbQ,\clockE)$ and 
  $s_2 = (\tsbP, \queue{}, \clockN) \mmid  (\tsbQ,  \queue{}, \clockE)$.
  \begin{enumerate}[label=\({\alph*}]

  \item To prove item~\eqref{item:turn-simulation:a} of~\Cref{def:turn-simulation}, 
    we consider the possible moves of $s_2$:
    \begin{itemize}[label=$-$]

    \item \nrule{[M-$\sumInt$]}. We have: 
      \[
         s_2 \cmove {\pmvA \says\atomOut{a}} (\tsbPi, \queue{\atomOut{a}},  \reset{\clockN }{\resetR})
                           \mmid  (\tsbQ, \queue{}, \clockE  ) = s_2'
      \] 
      where the premise requires  $\tsbP =  \atomOut{a}\setenum{\guardG,\resetR}.\tsbPi \oplus \tsbPii$ 
      and  $\clockN \in \sem{\guardG}$.
      Hence, \nrule{[S-$\sumInt$]} yields:
      \[
      \irule
      {(\tsbP,\clockN) \smove{\tau} (\todo[\atomOut{a}\setenum{\guardG,\resetR}]\tsbPi, \clockN)}
      {s_1 \smove{\tau} (\todo[\atomOut{a}\setenum{\guardG,\resetR}]\tsbPi,\clockN) \mid (\tsbQ, \clockE) = s_1'}
      \]      
      Then,
      $(s_1',s_2') \in \relR[2] \subseteq \relR$. 

    \item \nrule{[M-$\sumExt$]}.
      This case does not apply.

    \item \nrule{[M-Del]}. We have: 
      \[
        s_2 \cmove{\delta}  (\tsbP, \queue{}, \clockN +\delta \mmid \tsbQ, \queue{},  \clockE +\delta )
      \]
      where the premise requires  
      $\clockN +\delta \in  \rdy{\tsbP}$ and  $\clockN +\delta \in \rdy{\tsbQ}$.
      Hence, \nrule{[Del]} yields:
      \[
      \irule
      {(\tsbP, \clockN) \smove{\delta} (\tsbP,\clockN+\delta) \qquad
       (\tsbQ, \clockE) \smove{\delta} (\tsbQ,\clockE+\delta)}
      {s_1 \smove{\delta} (\tsbP, \clockN+\delta) \mid (\tsbQ,\clockE+\delta) = s_1'}
      \]
       Then,
      $(s_1',s_2') \in \relR$. 
  \end{itemize}

  \item
    To prove item~\eqref{item:turn-simulation:b} of~\Cref{def:turn-simulation}, 
    we consider the possible moves of $s_1$:
    \begin{itemize}[label=$-$]

    \item \nrule{[S-$\sumInt$]}. We have: 
      \[
      s_1 \smove{\todot} (\todo[ \atomOut{a}\setenum{\guardG,\resetR}]\tsbPi, \reset{\clockN}{\resetR})
      \mid (\tsbQ, \clockE)
      \] whose premise requires  
      $\tsbP =
      \atomOut{a}\setenum{\guardG,\resetR}.\tsbPi \oplus \tsbPii$, with
      $\clockN \in \sem{\guardG}$.  Hence $s_2 \cmove{\pmvA \says\atomOut{a}}$.

    \item \nrule{[S-$\tau$]}.
     This case does not apply.

    \item \nrule{[S-Del]}. We have: 
    \[
     s_1 \smove{\delta} (\tsbP, \clockN + \delta) \mid (\tsbQ, \clockE + \delta)
    \] 
    where the premise requires  
      $\clockN +\delta \in  \rdy{\tsbP}$ and  $\clockN +\delta \in \rdy{\tsbQ}$.
    Hence, by \nrule{[M-Del]} we have   $s_2 \cmove{\delta}$.
    \end{itemize}

  \item To prove  item~\eqref{item:turn-simulation:b} of~\Cref{def:turn-simulation},  assume  that  $s_2
    \in S_2$. By definition of $S-2$, $s_2$ has the form   $(\success \mid \success, \queue{},
    \clockN)$.  Then  $s_1 = (\success, \clockN) \mid (\success, \clockE) \in S_1$.
  \end{enumerate}

  \noindent {\bf Case 2}: 
  Let $s_1 = (\todo[ \atomOut{a}\setenum{\guardG,\resetR}]\tsbP, \clockN) \mid (\tsbQ, \clockE)$, 
  $s_2 = (\tsbP, \queue{\atomOut{a}}, \reset{\clockN}{\resetR}) \mmid (\tsbQ, \queue{}, \clockE)$, 
  and $\clockN \in \sem{\guardG}$.
  \begin{enumerate}[label=\({\alph*}]

  \item
    To prove item~\eqref{item:turn-simulation:a} of~\Cref{def:turn-simulation}, 
    we consider the possible moves of $s_2$:
    \begin{itemize}[label=$-$]

    \item \nrule{[M-$\sumInt$]}.
      This case does not apply.

    \item \nrule{[M-$\sumExt$]}. We have: 
      \[ 
      s_2 \cmove{\pmvB \says \atomIn{a}}(\tsbP , \queue{}, \reset{\clockN}{\resetR}), \mmid (\tsbQi,
      \queue{} ,\reset{\clockE}{\resetS})= s'_2 
      \] whose premise requires  
      $\tsbQ =
      \atomIn{a}\setenum{\guardF, \resetS}.\tsbQi + \tsbQii$ with $\clockE
      \in \sem{\guardF}$. %
      Since by hypothesis $\clockN \in \sem{\guardG}$, by \nrule{[S-$\tau$]} we have:
      \[
       \irule
       { (\todo[ \atomOut{a}\setenum{\guardG,\resetR}]\tsbP, \clockN) \smove{\atomOut{a}} (\tsbP,
          \reset{\clockN}{\resetR}) \quad 
        (\tsbQ, \clockE) \smove{\atomIn{a}} (\tsbQi,\reset{\clockE}{\resetS}) 
       }
       {(\tsbP,\clockN) \mid (\tsbQ,\clockE) \smove{\; \tau \;} 
            (\tsbPi,\clockNi) \mid (\tsbQi,\clockEi) = s_1'}
      \]
      Then,
      $(s_1',s_2') \in \relR$. 

    \item \nrule{[M-Del]}.
      This case does not apply.

    \end{itemize}

  \item
    To prove item~\eqref{item:turn-simulation:b} of~\Cref{def:turn-simulation}, 
    we consider the possible moves of $s_1$:
    \begin{itemize}[label=$-$]

    \item \nrule{[S-$\sumInt$]}. We have: 
      \[
      s_1 \smove{\todot} (\todo[
      \atomOut{a}\setenum{\guardG,\resetR}]\tsbP, \clockN) \mid (\todo
      [\atomOut{b}\setenum{\guardF, \resetS}]\tsbQ, \clockE) = s'_1
      \]
      whose premise requires $\tsbQ = \atomOut{b}\setenum{\guardF, \resetS}.\tsbQi \oplus
      \tsbQii$ and $\clockN \in \sem{\guardF}$.  We have that $s_2$ is stuck but
      so is $s'_1$; and $(s_1',s_2) \in \relR$.

    \item \nrule{[S-$\tau$]}. We have: 
      \[
      s_1 \smove{\tau} (\tsbP, \reset{\clockN}{\resetR}) \mid (\tsbQi,
      \reset{\clockE}{\resetS})
      \]whose premise requires $\tsbQ =
      \atomIn{b}\setenum{\guardF, \resetS}.\tsbQi + \tsbQii$ and $\clockE
      \in \sem{\guardF}$. Hence $s_2 \cmove{\pmvB \says \atomIn{a}}(\tsbP
      \mmid \tsbQi, \queue{}, \reset{\clockN}{\resetR}\sqcup \reset{\clockE}{\resetS}) = s'_2$ 
      and $(s'_1,s'_2) \in \relR$.

    \item \nrule{[S-Del]}.
    This case does not apply.
    \end{itemize}

  \item To prove \ref{item:turn-simulation:c}, let us assume $s_2 \in S_2$, which implies $s_2 = (\success \mid \success, \clockN)$.  By hypothesis of \emph{case 2} this is not possible.
  \end{enumerate}

  \noindent {\bf Case 3}: Let $s_1 = (\todo[
  \atomOut{a}\setenum{\guardG,\resetR}]\tsbP, \clockN) \mid (\todo[
  \atomOut{b}\setenum{\guardF, \resetS}]\tsbQ , \clockE)$, 
  $s_2 = (\tsbP, \queue{\atomOut{a}}, \reset{\clockN}{\resetR}) 
          \mmid (  \atomOut{b}\setenum{\guardF, \resetS}.\tsbQ \oplus \tsbQi), \queue{}, \clockE).
  $
  In this case, both $s_1$ and $s_2$ are stuck and neither of them is a
  success state.
  }
  \qed
\end{proofof}

\bigskip
The diagram in~\Cref{fig:graph-proofs} illustrates the 
dependencies among the proofs.

\begin{figure}[t]
  \centering
  \scalebox{0.5}[0.4]{
    \Large
    \input{graph-proofs.tex}
  }
  \caption{Dependencies among the proofs.}
  \label{fig:graph-proofs}
\end{figure}

\end{document}